\documentclass[12pt, oneside, reqno]{amsart}
\usepackage[margin=2.5cm]{geometry}
\usepackage{booktabs} 

\usepackage{hyperref}
\usepackage{amssymb}

\usepackage{amsfonts,color}
\usepackage{mathabx}
\usepackage{amscd}

\usepackage{footnote}
\usepackage[utf8]{inputenc}
\usepackage{mathtools}
\usepackage{amsthm,wasysym}
\usepackage[english]{babel}
\usepackage{bbm}
\usepackage{colonequals}
\usepackage{cancel}
\usepackage{dsfont}
\usepackage{adjustbox}
\usepackage{graphicx}
\usepackage{bm}
\usepackage{amsfonts}
\usepackage{wrapfig}
\usepackage{subcaption}
\usepackage{amsmath}
\usepackage{float}
\usepackage{stmaryrd}
\usepackage{accents}
\usepackage{pdfsync}
\usepackage{todonotes}

\usepackage{algorithm}
\usepackage{algpseudocode}

\usepackage{siunitx}
\newcolumntype{N}{c@{}S}

\usepackage{pgf,tikz,pgfplots}
\pgfplotsset{compat=1.15}
\pgfplotsset{compat=newest}
\usepackage{mathrsfs}
\usetikzlibrary{arrows}
\usetikzlibrary {arrows.meta} 
\usetikzlibrary{calc}
\def\centerarc[#1](#2)(#3:#4:#5)
{ \draw[#1] ($(#2)+({#5*cos(#3)},{#5*sin(#3)})$) arc (#3:#4:#5); }

\usepackage{chngcntr}
\counterwithin{equation}{section}

\newtheorem{theorem}{Theorem}[section]
\newtheorem{lemma}[theorem]{Lemma}
\newtheorem{prop}[theorem]{Proposition}

\theoremstyle{definition}

\theoremstyle{remark}
\newtheorem{remark}[theorem]{Remark}

\def\d{\partial}

\newcommand{\crv}{\mathcal{S}}
\newcommand{\re}[1]{\hat{#1}}

\newcommand{\curve}{\gamma}

\newcommand{\rx}{\re{x}}
\newcommand{\ry}{\re{y}}
\newcommand{\rz}{\re{z}}
\newcommand{\rS}{\re{S}}

\newcommand{\CG}{\mathcal{C}}

\newcommand{\rn}{\re{n}}

\newcommand{\bb}[1]{\mathbb{{#1}}}
\newcommand{\cl}[1]{\mathcal{{#1}}}

\newcommand{\ahel}{a_{\mathrm{H}}}
\newcommand{\bhel}{b_{\mathrm{H}}}

\newcommand{\ii}{\hat\imath}
\newcommand{\hD}{{\hat{D}}}

\newcommand{\om}{\Omega}
\newcommand{\ominf}{{\Omega}}

\newcommand{\cyl}{{\hat{\om}}}
\newcommand{\cylinf}{\hat{\om}}
\newcommand{\mat}[1]{\ensuremath{\begin{bmatrix}#1\end{bmatrix}}}
\newcommand{\R}{\mathbb{R}}
\newcommand{\C}{\mathbb{C}}

\renewcommand{\div}{\mathrm{div}}

\newcommand{\Hone}[1][]{\ensuremath{H^1\ifthenelse{\equal{#1}{}}{}{(#1)}}}
\newcommand{\Honez}[1][]{\ensuremath{H^1_0\ifthenelse{\equal{#1}{}}{}{(#1)}}}

\newcommand{\TF}{\mathcal{D}}
\DeclareMathOperator*{\argmin}{arg\,min}

\newcommand{\FN}{N}
\newcommand{\FB}{B}
\newcommand{\Ftau}{\tau}

\title{Guided modes of helical waveguides}

\author[J.~Gopalakrishnan]{Jay~Gopalakrishnan}
\address{Portland State University, PO Box 751, Portland OR 97201, USA }
\email{gjay@pdx.edu}

\author[M.~Neunteufel]{Michael~Neunteufel}
\address{Portland State University, PO Box 751, Portland OR 97201, USA}
\email{mneunteu@pdx.edu}

\begin{document}
\begin{abstract}
  This paper studies guided transverse scalar modes propagating
through helically coiled waveguides. Modeling the modes as solutions
of the Helmholtz equation within the three-dimensional (3D)
waveguide geometry, a propagation ansatz transforms the mode-finding
problem into a 3D quadratic eigenproblem. Through an untwisting map,
the problem is shown to be equivalent to a 3D quadratic eigenproblem
on a straightened configuration. Next, exploiting the constant
torsion and curvature of the Frenet frame of a circular helix, the
3D eigenproblem is further reduced to a two-dimensional (2D)
eigenproblem on the waveguide cross section. All three
eigenproblems are numerically treated. As expected, significant
computational savings are realized in the 2D model.  A few
nontrivial numerical techniques are needed to make the computation
of modes within the 3D geometry feasible. They are presented along
with a procedure to effectively filter out unwanted non-propagating
eigenfunctions. Computational results show that the geometric effect
of coiling is to shift the localization of guided modes away from
the coiling center. The variations in modes as coiling pitch is
changed are reported considering the example of a coiled optical
fiber.
  \\
	\vspace*{0.25cm}
\end{abstract}

\maketitle

\section{Introduction}
\label{sec:intro}

Waveguides are often bent, deformed, or coiled for varied purposes,
such as to fit it into a desired space, for tuning its resonances to
desired frequencies, or obtaining better separation of fundamental
mode from higher-order modes.  Examples include acoustic wave
propagation of sound waves in moulded pipes of a musical instrument,
and light wave propagation in coiled optical fibers or gas-filled
laser fiber amplifiers.  This paper develops a theoretical model for
computing the guided modes of a helically coiled waveguide, taking
into account its specific geometry.  Although leaky modes and
confinement loss of deformed waveguides are also of great interest,
the scope of this work is limited to perfectly guided modes of no loss
and considers only waveguides which admit such a boundary condition.

A straight (unbent)  waveguide is translationally invariant in its
longitudinal propagation direction. Then the three-dimensional (3D)
domain it occupies admits separation of variables, which allows one to
obtain natural solutions, called (propagating) modes, that propagate along the
longitudinal direction, maintaining the same transverse profile on
waveguide cross sections.  Moreover, the invariant transverse profile can
be found by solving a dimension-reduced equation on the
two-dimensional (2D) transverse cross section. The dimension reduction
(from 3D to 2D) is very attractive for practical computations.  The
main question we seek to answer in this paper is whether such 2D modes
(exist and) are also computable for helical waveguides. The approach we follow
is to find an ``unbending map'' that transforms the helical
waveguide into a straight one, and then apply the known techniques on
the straightened configuration.

To the best of our knowledge, this is the first work to clarify the
role of torsion, in addition to curvature, while computing propagating
helical transverse modes. Arguing that such modes, due to the torsion
of the helix, must incorporate a cross-sectional rotation while
propagating, we develop an ansatz for the modes.  This then leads us
to the dimension reduction to a 2D eigenproblem for the modes. Informally, using
the divergence (div) and gradient (grad) operators on the waveguide
cross section, the 2D eigenproblem for a mode $\psi$ with propagation
constant $\beta$ takes the form
$\div ( J \cl A \mathop{\text{grad}} \psi)
+ J k^2 n^2 \psi + \ii \beta \div(J r \psi ) + \ii
\beta J r\cdot \mathop{\text{grad}} \psi  = \beta^2 J^{-1} \psi$
where $n$ is the waveguide material coefficient, $k$ is the operating wavenumber, and
$J, \cl A$ and $r$ are extra coefficients which depend on
torsion and curvature of the  geometry of  helical deformation, as seen more precisely 
descriptions later in~\eqref{eq:2D-eigenproblem} and
Theorem~\ref{thm:2D-reduction}.
Circularly bent waveguides (with no torsion)
have been extensively investigated in the
optical literature~\cite{HeiblHarri75, Marcu76, Marcu82,
  ScherCole07}. All these works model the bent waveguide as a ring, or
a torus, with the terminal and initial cross sections fused. In
practice, however, long waveguides are coiled at some pitch. Our model
can take into account the effects of both the bend radius and the
pitch. Moreover, prior works used simplified terms potentially
justifiable under a slow bend assumption, an assumption that we have
no need for. We will show a relationship between our model and the
prior state-of-the-art as the geometry  approaches a limiting
case of no torsion.

How transverse mode profiles change under circular bending is a
well-studied topic in optics since optical fibers are often coiled in
spools.  It is interesting to note that in solid-state waveguides,
like dielectric optical fibers, elastic stresses caused by bending can
change the fiber's refractive index. As a fiber bends around some
center (like the center of the spool), according to the laws of
photoelasticity~\cite{Nye85}, the fiber's refractive index is expected
to increase toward the center (where the material is compressed) and
decrease away from the center (where the material is pulled in
tension). Hence, one expects the localization of the fundamental mode
profile to shift toward the center of the bend where the index is
higher, if only the stress-optic effects are taken into
consideration. However, there is also a purely mathematical or
geometric effect on the mode due to the bend.  The fundamental mode
profile {\em shifts exactly in the opposite way,} localizing away from
the center of the bend, if only geometric effects are taken into
account, as will be abundantly clear from our results in
\S\ref{sec:computing_3D}. 
This work is devoted only to capturing the geometric effect
accurately.  The mode shift seen in reality will depend on which among
the photoelastic shift and the geometric shift dominates.  The
strength of the photoelastic effect depends on the Pockel coefficients
of the underlying materials. It takes no effort to modify the
refractive index to account for photoelasticity in our model (once the
material making up the waveguide is known) since our theory admits
such variations in index profiles.

In this study, we have limited our scope to guided modes. The main
resulting drawback is that bend losses (which inevitably exist in
dielectric fibers~\cite{Marcu82a}) cannot be extracted from computed
eigenvalues from the formulation proposed here.  This drawback may be
rectifiable to some extent by using the beam propagation method into
which the mode profiles computed here can be fed in, and losses backed
out. Nonetheless, a more mathematically elegant approach would be
through an eigenvalue formulation that produces a confinement loss
factor from the imaginary part of a computed eigenvalue. To accomplish
this, new formulations with absorbing boundary conditions suitable for
helical waveguides may be needed. This is an interesting topic for
further study, not addressed in this paper.

It is natural to ask if more general waveguide deformations beyond
helical ones can be treated using the techniques presented here. This
is not clear to us. In the proof of our main theorem
(Theorem~\ref{thm:2D-reduction}) accomplishing the dimension
reduction, we make essential use of the fact that the torsion and
curvature of a helix  are constant. By the fundamental theorem of space
curves~\cite{Carmo76} in differential geometry, circular
helices are the only curves in $\bb R^3$ with constant torsion and
curvature. This gives us pause in attempting to generalize the
techniques here to arbitrarily deformed waveguides.

In the next section (\S\ref{sec:modes3D}), we start by modeling
the helically bent waveguide as a tubular expansion of an infinite
helix curve.  Seeking solutions of the Helmholtz equation within such
a domain under zero sources, we derive, after some geometric
preliminaries, a weak form of the resulting boundary value problem and
show that it is equivalent to a quadratic
eigenproblem. \S\ref{sec:untwist} introduces a map that untwists
the helical waveguide into a straightened configuration.  The
eigenproblem on the physical waveguide is then mapped to a 3D
quadratic eigenproblem on the straight cylinder. It is then further
reduced to a 2D quadratic eigenproblem in
\S\ref{sec:dimension-reduction} after clarifying an ansatz for
propagating helical transverse modes. \S\ref{sec:computing_3D}
presents discretizations of the 3D and 2D eigenproblems, cross
verification of results from the multiple eigenproblems for the modes,
discussion of findings, tables of computed propagation constants,
plots of computed eigenmodes, and a study of their variations with the
pitch of the winding.

\section{Guided modes of a helically bent waveguide}
\label{sec:modes3D}

We are concerned with helical waveguides of uniformly circular cross
section.  In this section, we introduce the three-dimensional (3D)
geometrical parameters and identify a suitable 3D equation for a
propagating mode of the waveguide, taking into account the curvilinear
nature of the propagation direction.

\subsection{The geometry}

We model the waveguide geometry 
by thickening an infinite helix curve. Accordingly, first 
let $\gamma: \R \to \R^3$ be a circular helix curve, parameterized by
\begin{align}
  \label{eq:helix'}
  \gamma(s) = \mat{
  \ahel \cos(s/l)\\  
  \ahel \sin (s/l)\\
  \bhel  s/l}, \qquad 
  l= \sqrt{\ahel^2+\bhel^2},
\end{align}
(see Figure~\ref{fig:helix_geometry'}) for some $\ahel, \bhel > 0$.  It
is easily verified that $\gamma(s)$ is an arclength parameterization. Note 
that its trace $\crv$, the image of $\gamma$, has no self-intersections, and that  $\gamma$ is a one-to-one onto map. The radius 
$\ahel$ is called the ``bend radius'' or ``coiling radius.''
The ``pitch'' of the helix $\gamma$ (or the height of
one turn of the helix) is $2\pi \bhel$ and its ``slope'' is
$\bhel /\ahel$, made by the angle
\begin{equation}
  \label{eq:alpha-defn}
  \alpha = \tan^{-1}\left( \frac{\bhel}{\ahel}\right),
\end{equation}
also marked in Figure~\ref{fig:helix_geometry'}. It is evident from the figure that 
the length of one turn of the helix equals $L=2 \pi l$.  Moreover,
$\gamma(s)$ is periodic with $L$ being the period. 

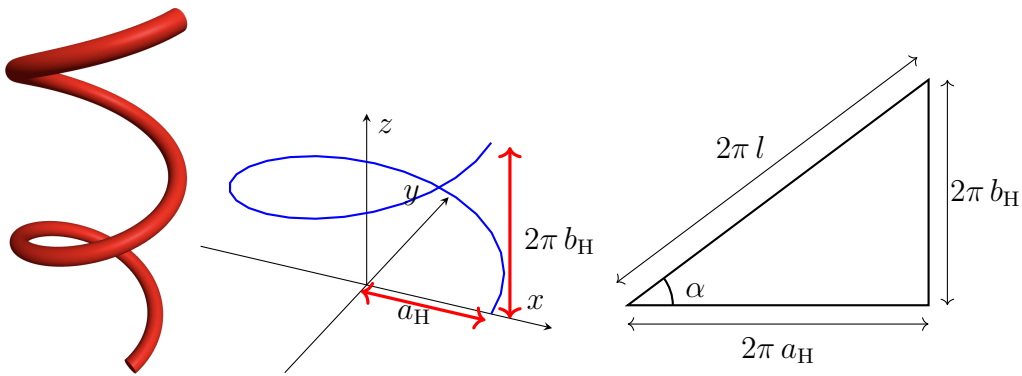
\begin{figure}[ht!]
  \centering
  \includegraphics[width=0.15\textwidth]{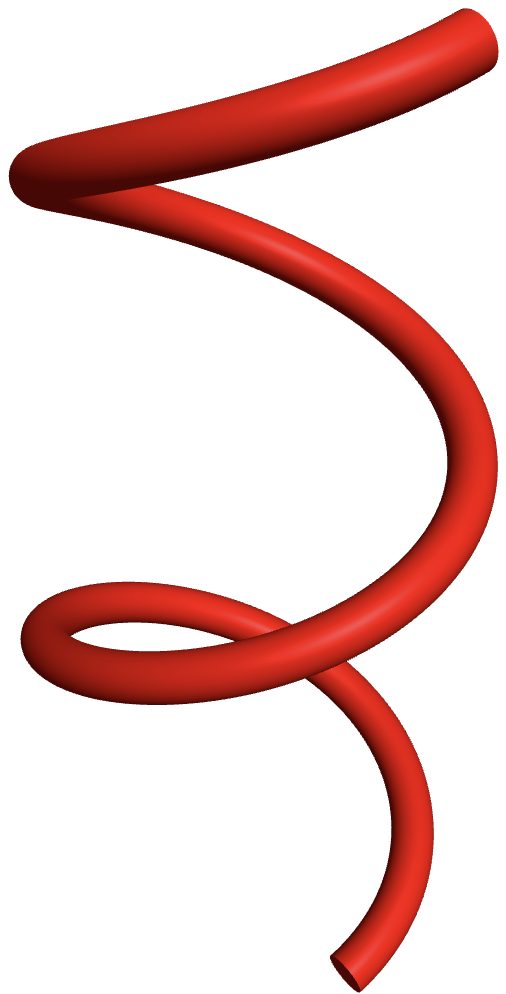}
  \begin{tikzpicture}
    \begin{axis}[
      axis equal,
      ylabel style={yshift=6cm},
      xlabel={$x$}, ylabel={$y$}, zlabel={$z$},
      axis lines=middle, 
      ticks=none, 
      ]
      \addplot3 [
      thick,
      blue,
      samples=30,
      samples y=0,
      domain=0:360, 
      ] 
      ({2*cos(x)}, {2*sin(x)}, {0.008*x});

      \addplot3 [
      very thick,
      red,
      <->,
      ] coordinates {(0,-0.2,0) (2,-0.2,0)}; 
      \node at (axis cs: {1}, -0.5, 0) {$\ahel$};
      \addplot3 [
      very thick,
      red,
      <->,
      ] coordinates {(2.3,0,0) (2.3,0, {0.008*360})}; 
      \node at (axis cs: {2.85}, 0, {0.004*360}) {\quad$2\pi\,\bhel$};
    \end{axis}
  \end{tikzpicture}
  \begin{tikzpicture}
    \draw[thick] (0,0) -- (4,0) -- (4,3) -- cycle;
    
    \draw[thick] (0.6,0) arc[start angle=0,end angle=36.87,radius=0.6];
    \node at (0.9,0.2) {$\alpha$};
    
    \node at (2,-0.6) {$2\pi\,\ahel$};
    \draw[<->] (0,-0.25) -- (4,-0.25);
    
    \node at (4.7,1.5) {$\;2\pi\,\bhel$};
    \draw[<->] (4.25,0) -- (4.25,3);
    
    \node at (1.5,2.1) {$2\pi\,l$};
    \draw[<->] (-0.15,0.3) -- (3.85,3.3);
    
  \end{tikzpicture}
  \caption{{\em Left:} A helical waveguide. {\em Middle:} The set $\crv$ giving
    the geometry of the waveguide centerline, i.e., the image of $\gamma$.
    {\em Right:} If the cylinder containing the curve $\crv$ is unwrapped, the 
    right triangle shown is obtained.}
  \label{fig:helix_geometry'}
\end{figure}

Next, consider the Frenet frame of $\gamma$.  The tangent vector
$T(s) = d \gamma / ds $ has unit length since we parameterized the
curve by arclength. Henceforth we abbreviate $d/ds$ by ``$\prime$'', a
prime. Since $T'(s) \ne 0$, the ``Frenet normal'' is the normalized vector
\begin{equation}
  \label{eq:Frenet-normal}
    \FN(s) = \frac{T'(s) } {\| T'(s)\|}.
\end{equation}
Throughout, we use $X \cdot Y$ and $\| X \|$, respectively, to refer
to the Euclidean inner product and norm for any $X, Y \in \R^3$.  The
moving Frenet frame is then the triple of vector fields 
$(T(s), \FN(s), \FB(s)),$ where $ \FB = T \times \FN$ denotes the
binormal vector. Also, recall that the Frenet structure equations read
\begin{align}
  \label{eq:struct_eq'}
  T^\prime = \kappa_n\,\FN-\kappa_g\,\FB,\qquad \FN^\prime = -\kappa_n\,T + \Ftau\,\FB,\qquad \FB^\prime = \kappa_g\,T - \Ftau\,\FN,
\end{align}
where $\kappa_g, \kappa_n$, and $\Ftau$ are the geodesic curvature,
the normal curvature, and the torsion, respectively.  In particular,
since $T'$ is in the direction of $\FN$ by~\eqref{eq:Frenet-normal}
and $T\cdot B = T \cdot N = 0$, the last equation
of~\eqref{eq:struct_eq'} implies
\begin{subequations}
  \label{eq:curvatures'}
  \begin{equation}
    \label{eq:kappa-g}
        \kappa_g = \FB^\prime\cdot T=-T^\prime\cdot \FB=0,
      \end{equation}
      while the first equation of~\eqref{eq:curvatures'} shows that
      \begin{equation}
        \label{eq:kappa-n}
        \kappa_n = T^\prime\cdot \FN=-\FN^\prime\cdot T= \| T'(s)\|.
      \end{equation}
      The remaining equation of~\eqref{eq:struct_eq'} gives
      \begin{equation}
        \label{eq:tau}
        \Ftau = \FN^\prime\cdot \FB = - \FB^\prime\cdot \FN.  
      \end{equation}
\end{subequations}
The (total) curvature $\kappa$ of a curve is the quantity
\begin{align}
  \kappa = \sqrt{\kappa_g^2 + \kappa_n^2}.
  \label{eq:total_curvature}
\end{align}

The Frenet frame and all quantities in the structure equations can be
immediately computed using the parameterization $\gamma(s)$ and the
above formulas.  Namely,
\begin{align}
    \label{eq:helix_frenet}
  T  = \frac{1}{l}\mat{-{\ahel}\sin(s/l)
  \\\phantom{-}{\ahel}\cos(s/l)\\{\bhel}},\qquad
  \FN = -\mat{\cos(s/l)\\\sin(s/l)\\0},\qquad
  \FB = \frac{1}{l}\mat{\phantom{-}{\bhel}\sin(s/l)
  \\-{\bhel}\cos(s/l)\\{\ahel}}.
\end{align}
Using trigonometric identities involving
the slope angle in~\eqref{eq:alpha-defn}, 
\begin{equation}
  \label{eq:sin-cos-alpha}
  \sin(\alpha)=\frac{\bhel}{l}, \qquad
  \cos(\alpha)=\frac{\ahel}{l}  
\end{equation}
we may also rewrite the above expressions in terms of the slope angle. 
The resulting torsion and geodesic curvature, normal curvature, and
total curvature equal 
\begin{align}
  \label{eq:torsion_frenet_helix}
  \Ftau =
  \frac{\sin(\alpha)}{l},\qquad \kappa_g=0,\qquad \kappa_n=\kappa=\frac{\cos(\alpha)}{l}=\frac{\ahel}{\ahel^2+\bhel^2}.  
\end{align}
by \eqref{eq:curvatures'} and \eqref{eq:total_curvature}.

Next, we thicken the curve to obtain a waveguide. Let $D(s)$ denote
the open disk of radius~$R$, centered at $\gamma(s)$, and lying in the
plane passing through $\gamma(s)$ with its normal vector equal to
$T(s)$.  We refer to $R$ as the ``waveguide radius'', which is
distinct from the ``bend radius'' $\ahel$.  {\em We assume that $R$ is
  small enough so that these transverse disks are all disjoint,} i.e.,
\begin{equation}
  \label{eq:3}
  D(s_1) \cap D(s_2) = \emptyset, \quad \text{for all } s_1 \ne s_2.
\end{equation}
All our results are under this assumption,
which is henceforth tacitly understood to hold.
The infinite helical  waveguide occupies the domain formed by the disjoint
union of these transverse disks,
\[
  \ominf = \bigsqcup_{s \in \R} D(s).
\]
Locally, near any point of the curve, the assumption~\eqref{eq:3}
constrains $R$ inversely in terms of the curvature of
$\gamma$. Specifically, 
\begin{equation}
  \label{eq:R<a}
  R < \frac{\ahel^2 + \bhel^2}{\ahel} = \frac{1}{\kappa}
\end{equation}
is one of the two sufficient conditions stated in
\cite{PrzybPiera01} (see also \cite{OlsenBohr10}) to ensure that a 
helical waveguide is not self-intersecting, as required by
\eqref{eq:3}.  For an alternate justification of the
assumption~\eqref{eq:R<a}, see Remark~\ref{rem:J-positive} later. 
Assumption~\eqref{eq:3} is also a global
constraint. A uniform finite bound on the curvature alone is not sufficient to
determine an $R$ for which~\eqref{eq:3} holds.

In view of the assumption~\eqref{eq:3}, any point
$(x, y,z) \in \ominf$ belongs to exactly one disk $D(s)$ whose center
is $\gamma(s)$. Hence, the mapping that takes any
$(x, y,z) \in \ominf$ to its unique centerline point $\gamma(s)$
is well defined, namely
\begin{equation}
  \label{eq:4}
  P:\; \ominf \to \gamma, \qquad
  P(x, y, z) = \gamma(s) \quad\text{for any } (x,y,z) \in D(s).  
\end{equation}
It gives the closest point projection onto the curve $\gamma$.  We
also need a mapping from $(x, y, z)$ to its centerline arclength value
$s$, namely 
\begin{subequations}
  \label{eq:S-defn}
  \begin{equation}
    \label{eq:S-defn-1}
    S: \; \ominf \to \mathbb R, \qquad
    S(x, y, z) = s \quad \text{for any } (x,y,z) \in D(s),
  \end{equation}
  which is again a proper definition since each
  $(x, y,z) \in \ominf$ belongs to exactly one $D(s)$.  
  Since $\gamma: \R \to \crv \equiv \gamma(\R)$ is a one-to-one and
  onto map (extendable smoothly to a tubular neighborhood of the
  curve), the inverse map $\gamma^{-1}: \crv \to \R$ exists and is
  smooth.  Applying $\gamma^{-1}$ to both sides of the equation
  in~\eqref{eq:4}, we find that
  \begin{equation}
    \label{eq:S-defn-2}
      S = \gamma^{-1} \circ P,
  \end{equation}
\end{subequations}
an alternative representation of $S$.  We do not have a
closed form expression for $S$ or $P$. Except when the
helix degenerates to a straight or toroidal waveguide, one cannot expect analytical expressions for them  since 
a transcendental equation must be solved. In
\S\ref{sec:helically_waveguide_qevp_comp} we present a procedure
to approximate $S$ numerically. We proceed to use $S$ to develop a
mode ansatz in the next subsection.

\subsection{Modes as eigenfunctions}

We are interested in modes propagating along the helical
waveguide~$\ominf$.  They are modeled as solutions of the
Helmholtz equation within the waveguide with zero sources.
Specifically, given a material coefficient function
$n : \ominf \to \R$ and a wavenumber $k >0$, we are interested in
finding a function $u : \ominf \to \C$ satisfying
\begin{subequations}
  \label{eq:helmholtz'}
    \begin{align} \label{eq:helmholtz-inf'}
      \Delta u + k^2 n^2 u &= 0, \quad \text{in } \ominf, \\
      \label{eq:helmholtz-bc-inf'}
      u &= 0, \quad \text{on } \d \ominf.
    \end{align}
\end{subequations}
Such waveguide modes arise when studying time-harmonic wave
propagation within $\ominf$.
We may think of $u$ as modeling a pressure variable in acoustic wave
propagation of sound waves in tubes, or a polarization-maintaining
electromagnetic wave in coiled optical fibers. Our notation
in~\eqref{eq:helmholtz-inf'} is closer to the optical fiber
application where $n$ can be interpreted as the optical refractive
index.

Of particular interest are modes that propagate in a curvilinear direction
through the waveguide. They have the form
\begin{equation}
  \label{eq:5}
    u(x, y, z) = e^{\ii \beta \varphi(x, y, z) } U(x, y, z)
\end{equation}
for some propagation constant $\beta$ and some smooth phase function
$\varphi$ that determines the wavefronts. Here $\ii$ denotes the
imaginary unit. Given a $u$, there are obviously multiple ways to
rewrite it in terms of a $U$ and a $\varphi$ as in~\eqref{eq:5}. A
specific $\varphi$ of interest will be specified shortly, but we begin
with an elementary observation that holds for any smooth $\varphi$.

\begin{prop}
  \label{prop:U-varphi-eqn}
  If $u$ in the form~\eqref{eq:5} solves~\eqref{eq:helmholtz'},
  then $U=0$ on $\d \ominf$ and 
  \begin{gather}
    \label{eq:8}
    \Delta U + \ii \beta \div( U \nabla \varphi) +
    \ii \beta \nabla \varphi \cdot \nabla U
    + (k^2 n^2 - \beta^2 \| \nabla \varphi\|^2) U
     = 0, \quad \text{in } \ominf.
  \end{gather}
\end{prop}
\begin{proof}
  By the product rule,
  \begin{align*}
    & \nabla u = e^{\ii\beta \varphi} \left(\nabla U
      + \ii\beta U\nabla {\varphi}\right),\\
    &\Delta u = e^{\ii\beta {\varphi}}
      \left(\Delta U + \ii\beta\div(U\nabla {\varphi})
      + (\ii\beta\nabla U - \beta^2 U\nabla {\varphi})\cdot\nabla \varphi\right),
  \end{align*}
  so~\eqref{eq:8}  follows by substituting these
  into~\eqref{eq:helmholtz-inf'}.
\end{proof}

Next, we refine the ansatz~\eqref{eq:5} by selecting a phase~$\varphi$
that is constant on each waveguide cross section. Recall that the $S$
defined in~\eqref{eq:S-defn} is constant on each transverse
disk~$D(s)$. Hence a solution of~\eqref{eq:helmholtz'} of the form
\begin{equation}
\label{eq:6}
u(x, y, z) = e^{\ii \beta \,S(x, y, z)\, } U(x, y, z)
\end{equation}
represents a wave that propagates along the fiber, having phases that
continuously increase with the arclength $s$ of the waveguide
centerline. We are interested in perfectly guided modes (i.e., their
energy does not decay as it propagates along the fiber) so we proceed
with the understanding that $\beta$ in~\eqref{eq:6} is a real number.  At this
point $U$ in the ansatz~\eqref{eq:6} depends on all coordinate
variables (but we will refine it further as we proceed).  By
Proposition~\ref{prop:U-varphi-eqn}, $U$ in~\eqref{eq:6}
solves~\eqref{eq:8} with $S$ in place of $\varphi$, i.e.,
\begin{subequations}
  \label{eq:BVP-cell}
\begin{gather}
  \label{eq:BVP-cell-1}
  \Delta U + \ii \beta \div( U \nabla {S}) +
    \ii \beta \nabla {S} \cdot \nabla U
    + (k^2 n^2 - \beta^2 \| \nabla {S}\|^2) U
  = 0 \qquad \text{in } \ominf,
  \\
  \label{eq:BVP-cell-2}
  U|_{\d\ominf} = 0
\end{gather}
\end{subequations}
for a sufficiently regular function $U$.

The system~\eqref{eq:BVP-cell} is a quadratic eigenvalue problem for $\beta$ and $U$. We will have more to say about helical transverse propagating modes in \S\ref{sec:heli-transverse-modes}. For now, we conclude this section by writing out a useful  weak formulation of~\eqref{eq:BVP-cell}. Let $H^1_{0,\mathrm{loc}}(\ominf)$ denote the standard Sobolev space of locally square-integrable
complex-valued functions on $\ominf$ whose first-order derivatives are
also locally square-integrable and whose trace vanishes on the boundary
$\d \ominf$.
For scalar functions $f, g$ on some
measurable domain $O$, we use $(f, g)_O$ to denote the complex
$L^2(O)$ inner product of $f$ and $g$. Even when $f$ and $g$ are
vector fields on $O$, we continue to use the same notation to denote the integral of their appropriate product, i.e., 
$(f, g)_O = \int_O f \cdot \overline g$. All integrals are
computed using the standard Lebesgue measure $dx$ in $O$, which we
omit from the integral notation when no confusion can arise. Let 
$\TF(\ominf)$ denote the space of smooth functions with compact support in $\ominf$. Define, for any $W\in H^1_{0,\mathrm{loc}}(\ominf)$ and $V \in \TF(\ominf)$, 
the following three sesquilinear forms on $H^1_{0,\mathrm{loc}}(\ominf)\times \TF(\ominf)$.
\begin{subequations}
  \label{eq:a-b-c}  
\begin{align}
    &a(W, V) := (\nabla W,  \nabla {V})_\ominf - (k^2n^2W, V)_\ominf, 
    \\
    & b(W, V) := \ii (W\nabla S,   \nabla{V})_\ominf
      - \ii (\nabla W, {V} \nabla S )_\ominf,
    \\
    &c(W, V) := ( W \nabla S ,  V \nabla S)_\ominf.
  \end{align}
\end{subequations}

\begin{prop}[Eigenproblem on helical waveguide]
  \label{prop:weak-form}
  In weak form,  the boundary value problem~\eqref{eq:BVP-cell} is
  the problem of finding a quadratic eigenvalue $\beta$ and its
  corresponding eigenfunction $U \in H^1_{0,\mathrm{loc}}(\ominf)$ such that
  \begin{equation}
    \label{eq:7}
    a(U, V) + \beta \, b(U, V) + \beta^2\,c(U, V) = 0, \qquad
    \text{ for all } V \in \TF(\ominf).
  \end{equation}
\end{prop}
\begin{proof}  
  Multiply~\eqref{eq:BVP-cell-1} by the complex conjugate of a test
  function $V \in \TF(\ominf)$ and integrate over $\ominf$. In the resulting equation,
  replacing  $\Delta U$ by $\div(\nabla U)$, there are 
  two terms with the divergence operator. We integrate
  both by parts. Then, using the compact support of $\TF(\ominf)$, the
  boundary integrals arising from the integration by parts vanish, and
  we obtain
  \begin{align}
    \label{eq:weak_form3d'}
    \int_{\ominf}
    \nabla U\cdot\nabla\overline{V}
    +
    \ii\beta\left(U\nabla{S}\cdot \nabla\overline{V}
    -\overline{V}\nabla S\cdot  \nabla{U}\right)
    + \left(\beta^2\, \|\nabla{S}\|^2 - k^2n^2\right)
    U\overline{V}
    \;= \,0.
  \end{align}
  Rearranging, we obtain~\eqref{eq:7}.
\end{proof}

\section{An unbending map and a straightened configuration}
\label{sec:untwist}

In this section, we introduce a map $\Phi$ that ``unbends'' the
helical deformation of the waveguide.  Using it, we show that the
previous quadratic eigenproblem~\eqref{eq:7} transforms into a variable
coefficient quadratic eigenproblem on a straight cylindrical waveguide,
which we shall refer to as the {\em straightened configuration}.

Recall the normal $\FN(s)$ and binormal $\FB(s)$ fields from
\eqref{eq:helix_frenet}. 
Letting 
\begin{equation}
  \label{eq:cyl-defn}
  \cylinf = \{(\rx,\ry,\rz)\in\R^3\,:\, \rx^2+\ry^2 < R^2\}
\end{equation}
denote an upright infinite cylinder,  define an unbending map that maps the helical waveguide $\ominf$ to the straightened configuration $\cylinf$ by 
\begin{equation}
\label{eq:Phi-defn}
  \Phi:\cylinf\to\ominf,\qquad \Phi(\rx,\ry,\rz) = \curve(\rz) + \rx\,\FN(\rz) + \ry\,\FB(\rz).
\end{equation}
Let $\hat\Gamma(s)$ denote the disk obtained by the
intersection of the plane $\hat z = s$
with $\cylinf$. 
Clearly, $\Phi$ maps $\cylinf$ one-to-one onto  $\ominf$ and furthermore
\begin{equation}
  \label{eq:10a}
  \Phi \big(\hat\Gamma(s) \big) =  D(s), \qquad s\in\R.
\end{equation}
We use $(\rx, \ry, \rz)$ to denote points in $\cylinf$ and $(x,y,z)$ to denote points in $\ominf$.

Given a function $U(x,y,z)$ on $\ominf$, its pullback   by $\Phi$
on $\cylinf$ is the function
\[
  \hat U = U \circ \Phi.
\]
Derivative operators are supplied with a hat
to indicate differentiation with respect to the variables
$\rx, \ry, \rz$ in $\cylinf$, e.g.,
$\hat\nabla \hat U = [\d_{\hat x} \hat U, \d_{\hat y} \hat U, \d_{\hat
  z} \hat U ]^t$.  Let $[\hat\nabla \Phi]$ denote the Jacobian matrix
containing all first-order derivatives of $\Phi$, i.e.,
$[\hat\nabla\Phi] = [\d_{\hat x} \Phi, \d_{\hat y} \Phi, \d_{\hat z}
\Phi].$ It is invertible (since $\Phi$ is).
Let $e_{\hat z}$ denote the unit vector in the $\hat z$
direction and let
\begin{gather}
  \label{eq:C-d-J-0}
  \CG= [\hat\nabla \Phi]^t [\hat\nabla \Phi], \quad
  d = \CG^{-1} e_{\hat z}, \quad 
  J= \big|\det [\hat\nabla \Phi]\,\big|.
\end{gather}
These quantities will appear in a transformed eigenproblem for $\hat U$, so we note the following formulas first.

\begin{lemma}
  \label{lem:C-J}
  For the map $\Phi$ in~\eqref{eq:Phi-defn}, the above quantities are given by 
  $J = |1 - \rx \kappa|,$  
  \begin{gather}
    \label{eq:C-J-d}
    \CG =
    \mat{
      1 & 0 & -\ry \Ftau
      \\
      0 & 1 & \phantom{-}\rx \Ftau
      \\
      -\ry \Ftau & \rx \Ftau & \; (1 - \rx \kappa)^2
                             + (\rx^2 + \ry^2)\Ftau^2
    },
    \quad\text{and}\quad
    d = J^{-2}
    \mat{
      \phantom{-}\ry \Ftau \\ -\rx \Ftau  \\ 1 
    }.
  \end{gather}
\end{lemma}
\begin{proof}
  Differentiating the vector in~\eqref{eq:Phi-defn} with respect to
  $\rx, \ry,$ and $\rz$, we obtain the columns of the Jacobian matrix
  of $\Phi$, 
  \begin{align*}
    \hat \nabla \Phi
    &  =
    \mat{
        \FN & \FB & \; \gamma'(\rz) + \rx \FN'(\rz) + \ry \FB'(\rz)
      }.
  \end{align*}
  Noting that   $\gamma^\prime=T,$  the Frenet structure equations of~\eqref{eq:struct_eq'} imply 
  \begin{align}
    \label{eq:Phi-derivative}
    \hat \nabla \Phi
    & =
    \mat{
        \FN & \FB & (1 - \rx \kappa) T + \rx \Ftau \FB - \ry \Ftau \FN 
    }.
  \end{align}
  The  expression for $\CG=[\hat\nabla \Phi]^t [\hat\nabla \Phi]$ in
  \eqref{eq:C-J-d} now
  follows by the orthonormality of $T$, $\FN$, and $\FB$.  Next, a cofactor
  expansion of  the just obtained expression for $\CG$ shows that
  \begin{equation}
    \label{eq:110}
    \det \CG = (1 - \rx \kappa)^2.
  \end{equation}
  Since we also know from \eqref{eq:C-d-J-0} that
  $\det \CG = (\det [\hat\nabla\Phi])^2 = J^2$, the stated expression
  for $J$ follows from~\eqref{eq:110}.  Finally, to prove the stated
  expression for $d$, it suffices to observe that when the matrix and
  the vector expressions in~\eqref{eq:C-J-d} are multiplied and
  simplified, we find that $ \CG d = e_{\hat z}$.
\end{proof}

\begin{remark}
  \label{rem:J-positive}
  By Lemma~\ref{lem:C-J} and \eqref{eq:torsion_frenet_helix}, the determinant reads
  \begin{equation}
    \label{eq:J-helix}
    J(\rx, \ry, \rz) =
    |1-\rx\kappa| =
    1-\frac{\rx}{l}\cos\alpha.
  \end{equation}
  We have removed the absolute value in the last expression because it
  is positive: indeed, from~\eqref{eq:cyl-defn}, we know that
  $\rx^2 + \ry^2 < R^2$, so 
  \[
    \rx\kappa< R \kappa  < 1
  \]
  due to~\eqref{eq:R<a}.
  It is interesting to note that what we have just shown, namely 
  \begin{equation}
    \label{eq:J>0}
    J > 0, 
  \end{equation}
  also yields, by the implicit function theorem, the local
  invertibility of the map $\Phi$. In particular, this gives a
  justification (different from \cite{PrzybPiera01}) of the
  sufficiency of condition~\eqref{eq:R<a} to avoid self-intersections
  for nearby points in the parameter domain $\cylinf$.
\end{remark}

The eigenproblem of Proposition~\ref{prop:weak-form} on the bent helical
waveguide can now be transformed into an eigenproblem on the straight
cylinder. The pullback map $U \mapsto \hat U = U \circ \Phi$ maps elements in $H^1_{0,\mathrm{loc}}(\ominf)$ one-to-one onto $H^1_{0,\mathrm{loc}}(\cylinf)$. 
Let
\begin{gather*}
  \hat n = n \circ \Phi,
  \qquad
  \hat a(\hat U, \hat V)
  = (J \CG^{-1} \hat \nabla \hat U,  \hat \nabla \hat{V})_{\cylinf} 
  -(J k^2\hat n^2 \hat U, \hat V)_{\cylinf}, 
  \\
  \hat b(\hat U, \hat V) = \ii (J \hat U, d \cdot \hat \nabla V)_{\cylinf}
  - \ii (J d \cdot \hat \nabla \hat U, \hat V)_{\cylinf}, \qquad 
  \hat c(\hat U, \hat V) = ( J^{-1}  \hat U, \hat V)_{\cylinf}.
\end{gather*}

\begin{prop}[Eigenproblem on straightened configuration]
  \label{prop:straight-cyl-ewp}
  An eigenvalue $\beta$ and corresponding eigenfunction $U$
  solves~\eqref{eq:7} on the waveguide $\ominf$
  if and only if $\beta$ and $\hat U = U \circ \Phi \in H^1_{0,\mathrm{loc}}(\cylinf)$ solves
  \begin{equation}
    \label{eq:11}
    \hat a(\hat U, \hat V) + \beta \, \hat b( \hat U, \hat V)
    + \beta^2\, \hat c( \hat U, \hat V) = 0, \qquad
    \text{ for all } \hat V \in \TF(\cylinf),
  \end{equation}
  on the straight cylinder $\cylinf$.
\end{prop}
\begin{proof}
  By the chain rule, 
  \begin{equation}
    \label{eq:13}
    \hat \nabla \hat U = [\hat\nabla \Phi]^t (\nabla U)\circ\Phi. 
  \end{equation}
  Using this within a change of variables formula, we obtain 
  \begin{align*}
    \int_\ominf \nabla U \cdot \nabla {V}
    & =
    \int_{\cylinf} [\hat\nabla \Phi]^{-t} \hat\nabla \hat U
    \cdot [\hat\nabla \Phi]^{-t} \hat \nabla {V}  \, |\det [\hat\nabla \Phi] \,|
    \\
    & = \int_{\cylinf} ([\hat\nabla \Phi]^t [\hat\nabla\Phi])^{-1} \hat \nabla \hat U
      \cdot \hat \nabla {\hat {V}}\; J.
  \end{align*}
  Replacing $V$ by $\overline{V}$, this shows that
  $a(U, V) = \hat a (\hat U, \hat V)$.

  Next, consider $b(U, V)$ in~\eqref{eq:a-b-c}. Since $\Phi$ maps
  $\hat \Gamma(s)$ one-to-one onto $D(s)$---see~\eqref{eq:10a}---the
  definition \eqref{eq:4} of $P$ shows that
  $(P \circ \Phi)(\rx, \ry, \rz) = \gamma(\rz)$. Hence the
  definition~\eqref{eq:S-defn} of $S$ shows that
  \begin{equation}
    \label{eq:14}
    (S \circ \Phi )(\rx, \ry, \rz) = \rz,    
  \end{equation}
  so $\rS = S \circ \Phi$ has $\hat \nabla \rS = e_{\hat
    z}$.  Hence one of the terms in $b(U, V)$, after applying the
  change of variables formula, becomes
  \begin{align*}
    (U \nabla S, \nabla V)_\ominf
    & = (J \hat U [\hat\nabla \Phi]^{-t}\hat\nabla \rS,
    [\hat\nabla \Phi]^{-t}\hat\nabla \hat V)_{\cylinf}
    \\
    & = (J \hat U \CG^{-1} e_{\hat z}, \nabla \hat V)_{\cylinf}
    = (J \hat U , d \cdot \nabla \hat V)_{\cylinf}.
  \end{align*}
  Handling the other term of $b(U, V)$ similarly, we show that
  $b(U, V) = \hat b(\hat U, \hat V)$.
  
  Finally, for the third sesquilinear form $c(U, V)$,
  \begin{align*}
    c(U, V) = ( U \nabla S ,  V \nabla S)_\ominf
    & = (J \hat U [\hat\nabla \Phi]^{-t}\hat\nabla \rS ,
    \hat V [\hat\nabla \Phi]^{-t}\hat\nabla \rS )_{\cylinf}
    \\
    & = (J ( e_{\hat z} \cdot \CG^{-1} e_{\hat z}) \hat U ,   \hat V )_{\cylinf}.
  \end{align*}
  By Lemma~\ref{lem:C-J},
  $ e_{\hat z} \cdot \CG^{-1} e_{\hat z} = e_{\hat z} \cdot d = J^{-2}$ showing that
  $c(U, V) = \hat c (\hat u , \hat v)$.  
\end{proof}

\section{Dimension reduction}
\label{sec:dimension-reduction}

In this section, we introduce a 2D eigenproblem for transverse modes
under further assumptions and show how it can be derived from the
prior 3D eigenproblem.

A basic question we must grapple with is what constitutes a transverse
mode in a helical waveguide. The difficulty is that the torsion of the
helix creates a natural rotation of the cross section as one proceeds in
the propagating direction. Therefore, unlike straight waveguides where
transverse modes do not vary in the propagating direction, for helical
waveguides, we seek to identify transverse modes with a built-in cross
sectional rotation, as defined next.

\subsection{Helical transverse modes}
\label{sec:heli-transverse-modes}

Define a vector field $Z(x, y, z)$ on $\om$ by
\begin{align}
  \label{eq:Z-defn}
  Z = \big(T(\rz)-\rx \kappa T(\rz) + \rx \Ftau \FB(\rz) -\ry \Ftau
  \FN(\rz)\big)\circ\Phi^{-1}.  
\end{align}
It has the property that $Z = T$ along the waveguide
centerline. Moreover, its tangential components on any cross section
$D(\rz)$ form a rotational vortex of strength determined by the
torsion~$\tau$. We define a {\em propagating helical transverse mode}
as any solution of the Helmholtz equation~\eqref{eq:helmholtz'} of the
form~\eqref{eq:6} satisfying an additional transversality condition
that the directional derivative of $U$ in the $Z$~direction vanishes,
i.e., propagating helical transverse modes are of the form
\begin{equation}
  \label{eq:ansatz-transverse}
  u(x, y, z) = e^{\ii \beta \,S(x, y, z)\, } U(x, y, z),
  \quad 
  Z \cdot \nabla U = 0.  
\end{equation}
This is a strengthening  of the prior
ansatz~\eqref{eq:6}. While the ansatz~\eqref{eq:6} posed no
restriction at all on admissible solutions (since any solution can be
brought to that form), now it is not even clear if there are Helmholtz
solutions satisfying the revised ansatz~\eqref{eq:ansatz-transverse}.
Later (in \S\ref{ssec:cross-verification}), we  present ample numerical evidence 
pointing to the existence of such modes. Here we proceed to show why
this is a very natural ansatz when viewed from the straightened
configuration.

First note that coiled waveguides occurring in practice, ignoring any
stress effects, when uncoiled, make a straightened waveguide whose
material properties only vary along transverse cross sections (and not
longitudinally along the propagation direction). To express this as a
mathematical assumption, we use the previously defined straightened
configuration and assume that $n(x, y, z)$ is such that
$\hat n = n \circ \Phi$ has no longitudinal variations there, i.e., 
\begin{align}
  \label{eq:n-assumption}
  \rn(\rx,\ry,\rz) \equiv \rn(\rx,\ry)
  \quad \text{ is independent of $\rz$.}
\end{align}
It is then natural to consider modes that also satisfy a similar assumption, i.e., to  restrict the arbitrary 3D variation of
$U(x, y, z)$ by assuming that $\hat U = U \circ \Phi$ is
independent of $\rz$ in the straightened configuration, i.e.,
\begin{align}
  \label{eq:ansatz'}
  \hat U(\rx,\ry, \rz) \equiv \hat U(\rx,\ry)
  \quad \text{ is independent of $\rz$.}
\end{align}
Such an assumption is equivalent to restricting to propagating helical
transverse modes defined in~\eqref{eq:ansatz-transverse}, as we now
show.

\begin{prop}
  Let $Z$ be as in \eqref{eq:Z-defn}. The condition~\eqref{eq:ansatz'} holds   for $\hat U = U \circ \Phi$ on $\cyl$ if and only if
  \begin{equation}
    \label{eq:ansatz-in-helix}
    Z \cdot \nabla U = 0 \qquad \text{ in } \ominf.
  \end{equation}
\end{prop}
\begin{proof}
  Condition~\eqref{eq:ansatz'} is equivalent to
  $e_{\hat z} \cdot \hat\nabla\hat U =0$, which, by~\eqref{eq:13}, is
  the same as
  \begin{align*}
    (e_{\hat z} \cdot \hat\nabla\hat U) \circ \Phi^{-1} 
    & =
      e_{\hat z} \cdot  \big( [\nabla \Phi]^t \nabla U\big) 
      =
      \big([\nabla \Phi] e_{\hat z} \big) \cdot   \nabla U. 
  \end{align*}
  Since $[\nabla \Phi] e_{\hat z} = Z$ by
  \eqref{eq:Phi-derivative}, the result follows.
\end{proof}

\subsection{Two-dimensional model}

Let $\hD = \hat\Gamma(0)$. It represents the (uniform) cross section
of the straightened configuration $\cyl$. Under the above assumptions,
we are able to reduce the 3D eigenproblem to a 2D eigenproblem on $\hD$. Define 
\begin{subequations}
  \label{eq:r-and-A}
  \begin{equation}
    \label{eq:r-A-defn}
  r =
  \frac{1}{J^2} 
  \begin{bmatrix}
    \phantom{-}\tau \ry \\
    -\tau \rx
  \end{bmatrix},
  \qquad
  \cl A =
  \frac{1}{J^2} 
  \begin{bmatrix}
    J^2 + \tau^2 \ry^2 & - \tau^2 \rx\ry
    \\
    -\tau^2 \rx\ry & J^2 + \tau^2 \rx^2 
  \end{bmatrix},    
\end{equation}
with respect to the $e_{\rx}, e_{\ry}$ basis on $\hD$.
The matrix $\cl A$ can be alternately expressed
as
\begin{equation}
\label{eq:A-alt}
\cl A = I_{2 \times 2 } + J^2 r r^t,
\end{equation}
\end{subequations}
where $I_{2 \times 2}$ is the $2 \times 2$ identity matrix.
Note that for
a function $\hat U$ like in~\eqref{eq:ansatz'}, the matrix $\cl A$ can
be naturally applied to the gradient
$\hat \nabla \hat U = e_{\rx} \d_{\rx} \hat U + e_{\ry} \d_{\ry} \hat U$ since
it has no $\rz$ component. We do not distinguish between the 2D and 3D
vectors when the extra component vanishes. The 2D eigenproblem is to find a $\hat U$
satisfying~\eqref{eq:ansatz'} together with an eigenvalue $\beta$ that 
solves
\begin{equation}
  \label{eq:2D-eigenproblem}
  \begin{aligned}
    \hat\div ( J \cl A \hat \nabla \hat U )
    + \ii \beta \hat \div(J r \hat U) +  \ii \beta J r\cdot \hat \nabla \hat U
    + J k^2 \hat n^2 \hat U
    & = \beta^2 J^{-1} \hat U,
    && \text{ in } \hD,
    \\
    \hat U
    & = 0,
    && \text{ on } \d\hD,
  \end{aligned}
\end{equation}
where the divergence and the gradient are with respect to just the two
variables $\rx$ and $\ry$.  Its weak form~\eqref{eq:12} is displayed
in the next result using $H_0^1(\hD)$, the Sobolev subspace of
$L^2(\hD)$-functions on the 2D domain $\hD$ whose derivatives are in
$L^2(\hD)$ and whose trace vanishes on the boundary $\d\hD$. Any
$\hat U$ in $H_0^1(\hD)$, being a function of $(\rx, \ry)$ in $\hD$,
is obviously independent of $\rz$ and satisfies~\eqref{eq:ansatz'}.

\begin{theorem}[2D eigenproblem for helical transverse modes]
  \label{thm:2D-reduction}
  Suppose~\eqref{eq:n-assumption} holds. Then any $\hat U$ in
  $H_0^1(\hD)$ and a real number $\beta$ satisfies
  \begin{equation}
    \label{eq:12}
    \begin{aligned}
      (J \cl A \hat \nabla\hat U, \hat \nabla \hat V)_{\hD}
      -(J k^2 \hat n^2 \hat U, \hat V)_{\hD}
      & + \beta \,\ii \Big[ (J \hat U, r \cdot \hat \nabla\hat V)_{\hD}
        -  (J r \cdot \hat \nabla \hat U, \hat V)_{\hD}\Big]
      \\
      & + \beta^2 (J^{-1} \hat U, \hat V)_{\hD}      = 0 
    \end{aligned}
  \end{equation}
  for all $\hat V \in H_0^1(\hD)$ 
  if and only if $\beta$ and $U = \hat U \circ \Phi^{-1}$
  solves~\eqref{eq:7} on the helical waveguide $\ominf$.
\end{theorem}
\begin{proof}
  By Proposition~\ref{prop:straight-cyl-ewp},
  $U = \hat U \circ \Phi^{-1}$ solves \eqref{eq:7} on $\ominf$ if and
  only if $\hat U$ solves \eqref{eq:11} on~$\cylinf$. We proceed to
  simplify the terms of~\eqref{eq:11} when~\eqref{eq:n-assumption}
  and~\eqref{eq:ansatz'} hold. For any $W $ in
  $\TF(\cylinf)$,   
  \begin{align*}
    \hat a(\hat U, W)
    & = (J \CG^{-1} \hat \nabla \hat U,  \hat \nabla {W})_{\cylinf} 
      -(J k^2\hat n^2 \hat U, W)_{\cylinf}
    \\
    & =
      \int_{-\infty}^{\infty}\int_{\hD}
      \left( J \CG^{-1} \hat \nabla \hat U
      \cdot  \hat \nabla \overline{W} 
      - J k^2 \hat n^2 \hat U \overline W \right) \; d\rx \, d\ry\,d\rz
    \\
    & = \int_{\hD}
      \left(
      J\CG^{-1} \hat \nabla \hat U  \cdot
      \left(\int_{-\infty}^{\infty} \hat \nabla \overline{W}\,d \rz\right)
      - J k^2 \hat n^2 \hat U
      \left(\int_{-\infty}^{\infty}  \overline{W}\,
      d \rz \right) \right)\, d\rx \, d\ry
    \\
    & = \int_{\hD}
    \left(
    J\CG^{-1} \hat \nabla \hat U  \cdot \hat \nabla \overline{\hat V }
      - J k^2 \hat n^2 \hat U \overline{\hat V }
      \right)\, d\rx \, d\ry,
  \end{align*}
  where we have used in the penultimate step, the fact that $J$ and
  $\CG^{-1}$ are independent of $\rz$ (as is clear from the
  expressions in Lemma~\ref{lem:C-J}, since the torsion and curvature
  of a helix are constant) as well as the $\rz$-independence of
  $\hat U$ and $\hat n$ given by 
  \eqref{eq:n-assumption}--\eqref{eq:ansatz'}, and have put
  \begin{equation}
    \label{eq:hatV-W}
    \hat V = \int_{-\infty}^{\infty} W\,d\rz
  \end{equation}
  in the last step.  Obviously, $\hat V$ is smooth, has no
  $\rz$-dependence, and is in $\cl D(\hD)$. Moreover, since $\hat W$
  has compact support, by the fundamental theorem of calculus,
  $\d_{\rz} \hat V = \int_{-\infty}^{\infty}\d_{\rz}\hat{W}\;
  d\rz=0$. Also noting that $\hat\nabla \hat U$ has no
  $e_{\rz}$~component, we conclude that while computing the value of
  the integrand term
  $\CG^{-1} \hat \nabla \hat U \cdot \hat \nabla \overline{\hat V}$,
  the last row and column of $\CG^{-1}$ are multiplied by zero and are
  not seen. So we may replace $\CG^{-1}$ with its $2 \times 2$
  submatrix obtained by removing the last row and column of
  $\CG^{-1}$, which is precisely the matrix $\cl A$. Thus we have
  proven that the first term of~\eqref{eq:11} simplifies to
  \[
    \hat a(\hat U, W) = (J \cl A \hat \nabla\hat U, \hat \nabla \hat V)_{\hD}
      -(J k^2 \hat n^2 \hat U, \hat V)_{\hD}.
  \]
  Analogously, since $r$ is the vector obtained from $d$ by removing
  its last component, we prove that
  \begin{align*}
    \hat b(\hat U, W)
    & = \ii (J \hat U, r \cdot \hat \nabla\hat V)_{\hD}
      -  \ii(J r \cdot \hat \nabla \hat U, \hat V)_{\hD},
    \\
    \hat c (\hat U, W)
    & =  (J^{-1} \hat U, \hat V)_{\hD}.
  \end{align*}  

  Thus, given a $\rz$-independent $\hat U$, if it satisfies
  \eqref{eq:11} for all $W \in \cl D (\cyl)$, then~\eqref{eq:12} holds
  for all $\hat V$ of the form~\eqref{eq:hatV-W} and hence for all
  $\hat V \in \cl D(\hD)$. Since $\cl D(\hD)$ is dense in
  $H_0^1(\hD)$, one implication is proved. For the converse, if
  $\hat U $ in $H_0^1(\hD)$ solves~\eqref{eq:12}, then retracing the
  above argument for $\hat U(\rx,\ry)$ extended to $\cyl$ constantly
  in $\rz$, we find that~\eqref{eq:11} must hold for all 3D test
  functions $W \in \cl D (\cyl)$. Hence, the result is proved.
\end{proof}

\subsection{The toroidal limiting case}
\label{sec:torus}

In the limiting case when the angle $\alpha$ goes to zero, the helical
waveguide becomes a ring, or a torus.  Then~\eqref{eq:sin-cos-alpha}
shows that $\bhel \to 0$ and both $\ahel$ and $l$ approach the same
value. (In this limit, as is clear from
Figure~\ref{fig:helix_geometry'}, the helical centerline becomes a
circle, so the waveguide becomes a torus.) In this limit, $\tau \to 0$
and $\kappa \to 1/\ahel$, as seen from~\eqref{eq:torsion_frenet_helix}. Hence, by
\eqref{eq:r-and-A},  we have 
\begin{align}
  \label{eq:J-limit-torus}
  \lim\limits_{\alpha\to 0} J= 1-\frac{\rx}{\ahel},\qquad
  \lim\limits_{\alpha\to 0}r=
  \begin{bmatrix}
    0 \\ 0 
  \end{bmatrix},
  \qquad
  \lim\limits_{\alpha\to 0}\cl A
  = I_{2 \times 2}.
\end{align}
Then the 2D {quadratic} eigenproblem~\eqref{eq:12} reduces to the
following {\em linear eigenproblem}
\begin{align}
  (J  \hat \nabla\hat U, \hat \nabla \hat V)_{\hD}
  -(J k^2 \hat n^2 \hat U, \hat V)_{\hD}
  =-  \beta^2 (J^{-1} \hat U, \hat V)_{\hD}     
  \label{eq:torus_lin_evp}
\end{align}
for the eigenvalue $-\beta^2$.
Since there is no torsion in this torus limit case, there is no need for a
cross-sectional rotation when considering 
mode transversality. Indeed, our general 
transversality condition $Z \cdot \nabla U = 0$  in
\eqref{eq:ansatz-transverse}, in the absence of torsion, simply
reduces to the standard condition
\begin{equation}
\label{eq:torsion-free-tranversality}
  T \cdot \nabla U = 0.
\end{equation}
We caution however that~\eqref{eq:torsion-free-tranversality} may not be appropriate
for general helical waveguides.

Prior works on bent optical fiber waveguides---see~\cite{ScherCole07}
and its antecedents~\cite{Marcu82, HeiblHarri75}---have considered
this toroidal case, where the fiber is deformed 
so that its inlet and outlet are fused to form a torus. For ``slow
bends,'' they proposed to account for the geometry of bending solely
by replacing $n$ with a modified effective index. In our notation
(taking into account the direction of our $N$), their modified
effective index equals the original index $n$ scaled by a factor of
$( 1 - \rx / \ahel)$, which matches the above-found limiting value of
$J$ from our analysis. Nonetheless, as seen
from~\eqref{eq:torus_lin_evp}, even for this simpler toroidal bending
model, to be fully accurate, one must modify not only the term
with~$n$, but also the remaining two terms in the equation for modes.

\subsection{The upright limit}
\label{sec:upright-limit}

The limiting case as $\alpha$ approaches $\pi/2$ is what we refer to
as the upright limit.  From~\eqref{eq:sin-cos-alpha}, we see that this
limit can be achieved either by fixing $\ahel$ to some positive number
and letting $\bhel$ grow, 
\begin{equation}
\label{eq:upright-1}
\ahel>0, \qquad \bhel \to \infty, \qquad \alpha \to \frac{\pi}{2},
\end{equation}
or, by fixing $\bhel$ to some positive number and making 
$\ahel$ small, 
\begin{equation}
  \label{eq:upright-2}
  \bhel > 0, \qquad \ahel \to 0, \qquad \alpha \to \frac{\pi}{2}.
\end{equation}
In both cases, the helical waveguide approaches the shape of a straight
upright cylinder centered around the $z$-axis. However, the two cases
differ in their {\em limiting torsion:} in the case of~\eqref{eq:upright-1},
since $l \to \infty$, we see from~\eqref{eq:torsion_frenet_helix}
that the torsion $\tau$ approaches zero. In the case
of~\eqref{eq:upright-2}, the torsion approaches $1/ \bhel.$ In both
cases, the curvature $\kappa$ approaches zero, and \eqref{eq:r-and-A} gives 
\begin{align*}
  \lim\limits_{\alpha\to \pi/2} J=1,
  \qquad
  \lim\limits_{\alpha\to \pi/2} r =
  \begin{bmatrix}
    \phantom{-}\tau y \\ -\tau x
  \end{bmatrix},\qquad
  \lim\limits_{\alpha\to \pi/2}
  \cl A = I_{2\times 2 } + r r^t.
\end{align*}
where $\tau $ is understood to be the limiting torsion. 
The 2D {quadratic} eigenproblem~\eqref{eq:12}, in this upright
limiting case, reduces to
\begin{equation}
    \label{eq:upright_evp}
\begin{aligned}
  (  \hat \nabla\hat U, \hat \nabla \hat V)_{\hD}
  & +
  ( r \cdot  \hat \nabla\hat U, r \cdot \hat \nabla \hat V)_{\hD}
  -(k^2 \hat n^2 \hat U, \hat V)_{\hD}
  \\
  & + \beta \,\ii \Big[ (\hat U, r \cdot \hat \nabla\hat V)_{\hD}
  -  (r \cdot \hat \nabla \hat U, \hat V)_{\hD}\Big]
    +  \beta^2 (\hat U, \hat V)_{\hD} = 0.     
\end{aligned}  
\end{equation}
In the case the upright limit is approached as
in~\eqref{eq:upright-1}, since the limiting torsion $\tau$ vanishes,
$r$ becomes the zero vector, and~\eqref{eq:upright_evp} becomes a
linear eigenproblem with $-\beta^2$ as the eigenvalue,
\[  (  \hat \nabla\hat U, \hat \nabla \hat V)_{\hD}  
  - (k^2\hat n^2 \hat U, \hat V)_{\hD}
  =- \beta^2 (\hat U, \hat V)_{\hD}.
\]
This is the familiar {\em linear eigenproblem for modes of an unbent
  straight waveguide}~\cite{Reide16} with longitudinally constant
material properties; and moreover, vanishing torsion also implies that
the transversality condition reduces to the standard
one~\eqref{eq:torsion-free-tranversality}.

However, if the upright limit is approached as in the other
case~\eqref{eq:upright-2}, we must solve the quadratic
eigenproblem. Note that we have {\em not} assumed that the material
coefficient $n(x, y, z)$ is independent of $z$ in the physical
configuration $\om$ (even when $\om$ is a straight cylinder in the
upright limit case). We have only assumed that~\eqref{eq:n-assumption}
holds in the straightened configuration~$\cyl$. The map $\Phi$,
just (un)twists the straight cylinder $\om$ onto itself  $\cyl=\om$
in the situation of~\eqref{eq:upright-2}.
Our setting allows inclusion of twisted
material configurations contained within a straight cylinder, as in
Figure~\ref{fig:straight_twisted}, in this upright limit.
\begin{figure}[htb]
  \centering
  \includegraphics[width=0.3\textwidth]{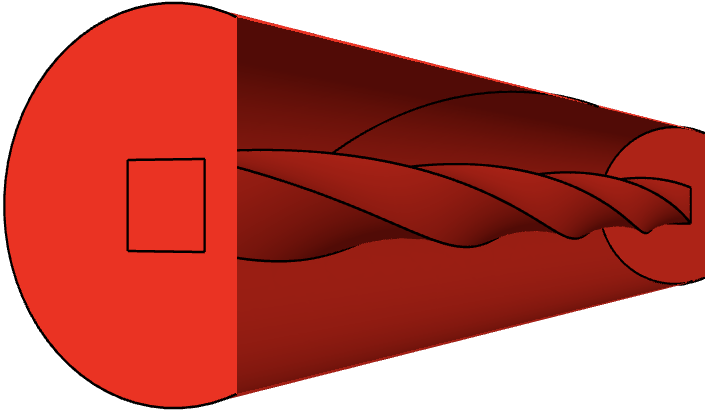}  
  \qquad \qquad \qquad 
  \includegraphics[width=0.3\textwidth]{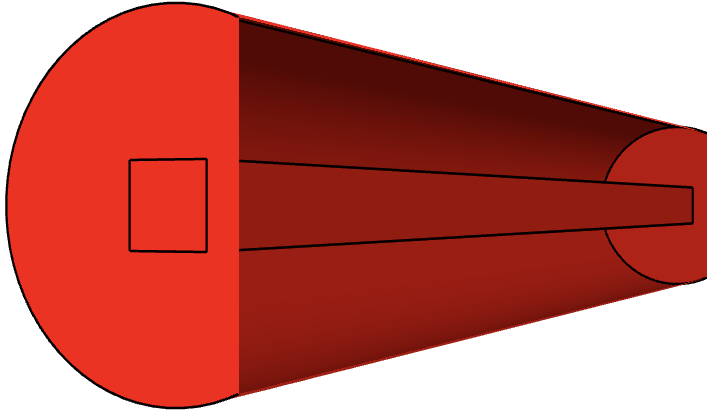}
  \caption{Assumption \eqref{eq:n-assumption} allows for $z$-varying
    material configurations like the twisted rectangular core in the
    left figure ($\om$) visible after clipping. The $z$-axis is into
    the page. The map $\Phi$ untwists it to the straightened
    configuration $\cyl$ on the right. }
\label{fig:straight_twisted}  
\end{figure}

\section{Numerical techniques and results}
\label{sec:computing_3D}

In this section, we use established numerical methods combined with a
few new tricks to solve the previously formulated eigenproblems, and
discuss the computational results, relating some to prior findings in
the optics literature.
In prior sections, we have introduced three eigenproblems for modes of
a helical waveguide: {\em (i)}~Proposition~\ref{prop:weak-form} gave a
3D eigenproblem on the physical helical waveguide, {\em
  (ii)}~Proposition~\ref{prop:straight-cyl-ewp} gave a 3D eigenproblem
on straightened configuration, and {\em
  (iii)}~Theorem~\ref{thm:2D-reduction} gave a 2D eigenproblem on a
waveguide cross section after limiting the longitudinal variation of
material coefficient by~\eqref{eq:n-assumption}.  In this section, we
briefly describe how each of these eigenproblems can be numerical
solved.

Of course, the 2D quadratic eigenvalue formulation~\eqref{eq:12} is
the most attractive in practice since it is the least expensive and
can be numerically solved with minimal tools. But we also compute with
the 3D quadratic eigenvalue formulations in~\eqref{eq:11}
and~\eqref{eq:7} for illustration and cross verification. The
additional tools needed for the 3D eigenproblems are described in
\S\ref{sec:helically_waveguide_qevp_comp} (and the reader only
interested in the 2D results may skip
\S\ref{sec:helically_waveguide_qevp_comp}).

\subsection{Common tools for all three eigenproblems}
\label{sec:common}

For all three eigenproblems, we discretize the underlying 2D or 3D
spatial domain using a simplicial conforming finite element mesh
$\cl T$, using curved elements as needed to fit the geometry
accurately. Let $\cl P_\ell(\cl T)$ denote the space of functions that
are polynomials of degree at most $\ell$ in each mesh element.  Each
of the eigenproblems~\eqref{eq:7}, \eqref{eq:11}, and \eqref{eq:12},
are discretized by replacing the infinite-dimensional Sobolev space
there by its intersection with $\cl P_\ell(\cl T)$, i.e., the Lagrange
finite element subspace of order $\ell$ on the mesh. The default
setting in our computations is $\ell=4$. Using the Lagrange basis, one
obtains in a standard way~\cite{Ihlen98}, a numerically solvable
matrix version of the quadratic eigenproblem
\begin{equation}
  \label{eq:matrix-quad-ewp}
  (A_0 + \beta A_1 + \beta^2 A_2) x =0,
\end{equation}
where $x$ is the vector of coefficients in the basis expansion of the
wanted eigenfunction, and $A_i$ are the finite element stiffness
matrices of the three sesquilinear forms in the prior eigenproblem
formulations. Since this process is completely standard, we omit
further details.

To solve the sparse quadratic eigenproblem \eqref{eq:matrix-quad-ewp},
we use the ``FEAST algorithm''~\cite{GopalGrubiOvall20a, Poliz09},
specifically its adaptation for polynomial eigenproblems
in~\cite{GopalParkeVande22}.  Finite element basis function
calculations, matrix assembly, and visualization are performed using
the open-source finite element library NGSolve~\cite{ngsolve,
  Sch97}. For reproducibility, the code for the benchmarks and
computational results are publicly available in~\cite{GN2025}.

\subsection{Further computational techniques for the 3D formulations}
\label{sec:helically_waveguide_qevp_comp}

The 3D eigenvalue formulation in~\eqref{eq:7} was posed on an
unbounded domain~$\om$. We need to rewrite it on a bounded domain
before discretizing with finite elements. This can be done by
exploiting the periodicity of the helix.  A single period cell of the
helical waveguide is obtained by starting from a cross section, say
$D(0)$, and performing one complete round of the helix to arrive at
$D(2\pi l)$. Note that the planes containing $D(0)$ and $D(2\pi l)$
are parallel (since $T(0) = T(2\pi l)$). The region of the waveguide
$\om$ in between these parallel planes is the single period cell on
which we compute, with periodic identification of corresponding points
in $D(0) $ and $D(2 \pi l)$.

To compute with~\eqref{eq:7}, we also need the phase function $S$ defined in
\eqref{eq:S-defn} and its gradient, or at least approximations of
them. Since $S$ is not available analytically, this poses a
problem. Our strategy to overcome it is to construct a finite element approximation
to $S$ using its values at some points.  Since $\gamma \circ S$ equals
the (orthogonal) closest point projection~$P$,
the value of $S(p)$ at any point  $p \in \om$ must satisfy
\begin{align}
  \label{eq:closest_point_projection}
  (p-\gamma(S(p)))\cdot T(S(p))=0,\qquad p\in \om,
\end{align}
where $\gamma$ and $T$ are
as in~\eqref{eq:helix'} and~\eqref{eq:helix_frenet}. On the tetrahedral mesh  $\mathcal{T}$, we solve \eqref{eq:closest_point_projection} on a collection of points $\{p_i\}$ with Newton's method. The points are selected to  be the  Gauss points used for numerical integration on the tetrahedra. To provide a good initial guess for the Newton iteration, we observe that in the simple case of a straight waveguide,  $S(p)=z$,
while in the other limit case of a torus, we have
$S(p)=\ahel\,\theta(p)$, where $\theta(p)$ is the cylindrical angle of
the point $p$ measure from the $x$-axis in the $xy$-plane.  Together
they lead us to choose as initial guess
$S_0(p) = z\,\sin\alpha+\ahel\,\theta(p)\, \cos\alpha$.
Our numerical experience shows  this  to be a 
solid choice as long as $\ahel>R$, i.e., a ``hole'' is present when
looking at the waveguide from above; otherwise, more sophisticated
guesses or post-processing schemes might be needed.

Once $S_i:=S(p_i)$ at the Gauss points have been (approximately) found
by the Newton iterations, we construct the $L^2$ best approximation
$S_h\in \mathcal{P}_\ell(\mathcal{T})$, with an $\ell$ that is
consistent with the order of the selected Gauss points $\{p_i\}$, by
\begin{align}
  S_h=\argmin_{\substack{R_h\in \mathcal{P}_\ell(\mathcal{T})\\ R_h(x_i)=S_i}}\frac{1}{2}\|R_h\|_{L^2(\om)}^2. 
\end{align}
Having the approximation $S_h$ at hand, we can compute its element-wise gradient $\nabla S_h$ enabling us to solve the quadratic eigenproblem \eqref{eq:BVP-cell} approximately.

The next difficulty specific to the 3D case is that there are
irrelevant (standing wave and non-propagating type) solutions
of~\eqref{eq:7}. We need a filtering mechanism to isolate the modes of
interest from all computed modes. To isolate the propagating helical
transverse modes, we use the
definition~\eqref{eq:ansatz-transverse} which uses
the vector field $Z$ in~\eqref{eq:Z-defn}.
However, computation of~$Z$ requires the inverse of $\Phi$. It is not available
analytically, but can be expressed using $S$. For any point $p$ in
$\om$, $\hat p = \Phi^{-1}(p)$ in $\cyl$ is given by 
\begin{align*}
  \Phi^{-1}(p)
  =\left((p-\gamma(S(p)))\cdot \FN(S(p)),(p-\gamma(S(p)))\cdot \FB(S(p)),S(p) \right).
\end{align*}
Using the approximation $S_h$ in place of $S$ in this expression,
we compute an approximation $Z_h$ of $Z$ and
use  $Z_h \cdot \nabla U =0$ to isolate  the relevant modes~$U$.

We now turn to the 3D eigenproblem~\eqref{eq:11} on the straightened
configuration.  Here again, to get a computable version, we may
consider one period cell, which in this case would be the bounded
subset $ \cyl_1 = \hat \Gamma(0) \times [0, 2\pi l]$ of the unbounded
domain~$\cyl$. However, now the periodicity is no longer relevant to
the straightened geometry nor seen in the coefficients
of~\eqref{eq:11}. Hence, for efficiency, we compute on a smaller
rescaled domain in $\rz$ direction,
$\hat \om_\sigma:=\hat \Gamma(0) \times [0, 2\pi\,l/\sigma]$ with $\sigma>1$.  When
\emph{reducing} the height by a factor $\sigma$, the eigenvalues $\beta$ of
interest \emph{scale} to $\sigma\beta$, since $\hat c(\hat U, \hat V)$
scales by $\sigma^{-2}$, $\hat b (\hat U, \hat V)$ scales by $\sigma^{-1}$, and
$\hat a(\hat U, \hat V)$ does not change for $\hat U$ and $\hat V$
that are constant in $\rz$, a property that modes of interest possess.
The identification of relevant modes is easy in this case: we only
need to check that assumption~\eqref{eq:ansatz'} holds true, i.e., for
a sequence of refined meshes we isolate modes $\hat U$ that satisfy
$ e_{\rz} \cdot \hat\nabla \hat U = 0.  $

\subsection{An optical fiber example}
\label{ssec:optical-fiber}

Some of our numerical results are with parameters from a fiber optic
waveguide. To facilitate description of such parameters, it is useful
to note that a ``step-index optical fiber'' consists of a cylindrical
core $D_0$ with cross section radius $r_0$ and a surrounding
cladding $D_1=D\setminus D_0$ of radius $r_1>r_0$.  Its
refractive index $n$ is modeled as a piece-wise constant function
\begin{align}
  \label{eq:ref-index-helix}
    n(x,y,z) = \begin{cases}
        n_0, & (x,y,z)\in D_0,\\
        n_1, & (x,y,z)\in D_1,
    \end{cases}
\end{align}
where $n_0$ and $n_1$ are positive constants.
Guided modes~\cite{Bures09} of a (straight, unbent) step-index fiber
can be analytically computed in terms of Bessel functions. They are
localized around the core and exponentially decay away from the core
within the cladding $D_1$. The cladding radius is usually an order
of magnitude higher than the core radius, which permits the use of zero
Dirichlet boundary conditions at the cladding's boundary.  The
operating signal wavelength determines the wavenumber~$k$. Together
these quantities define the nondimensional ``$V$-number'' (also called the ``normalized frequency'') \cite{Reide16} of the
fiber,
\begin{equation}
  \label{eq:V-number}
  V=r_0 k(n_0^2-n_1^2).
\end{equation}
The number of guided modes in the core
depends \cite[\S5.2.1]{Reide16} on how large $V$ is (and when $V$ is
smaller than the first root of the Bessel function $J_0$, there can
only be one guided mode). The guided mode with the largest propagation
constant $\beta$ is called the fundamental mode.

In some numerical simulations (see \S\ref{ssec:pitch}) we use
realistic values of $r_i$ and $n_i$, but we begin reports of
computational results in the next subsection by showing results from a
simulation with unit-sized parameters.  We shall use the above setting
to compare the results from the 3D helical waveguide
model~\eqref{eq:7}, the 3D straightened configuration
model~\eqref{eq:11}, and the 2D~model~\eqref{eq:12}. Results from the
2D model are presented first.

\subsection{Numerical results from the 2D eigenproblem} 
\label{ssec:2d-results}

We use the
setting from \S\ref{ssec:optical-fiber} but with (artificial) unit-sized
parameters that correspond to a tightly wound thick helically coiled
waveguide (see Figure~\ref{fig:cross_verification_2D_results}), namely
\begin{align}
  \label{eq:cross_verification_parameters}
  r_0=1,\quad r_1=2.2,\quad \ahel=3,\quad \bhel=5/(2\pi),
  \quad
  V = 15, 
  \quad
  n_1=1,
  \quad
  k = 1.
\end{align}
The value of $n_0$ is determined from~\eqref{eq:V-number}. Thus, $n$ is
fixed from~\eqref{eq:ref-index-helix}.

First, we report results obtained by solving the quadratic
eigenproblem of the 2D model~\eqref{eq:12} with polynomial order
$\ell=4$ and a sequence of meshes where a representative mesh-size $h$ is halved in
each step.  In the core, we used the mesh-size $h_c=h/2$.
Five eigenmodes localized near the core that were found
are displayed in Figure~\ref{fig:cross_verification_2D_results} and
their corresponding propagation constants $\beta_i$ are reported in
Table~\ref{tab:cross_verification_2D_results}. A clear
convergence is observed in the table as the mesh size $h$ is halved.
\begin{table}[ht!]
  \centering
  \begin{footnotesize}    
  \begin{tabular}{c|ccccccc}
    \toprule
    $h$  &  ne  &  ndof  &  $\beta_1^2$ &  $\beta_2^2$  &  $\beta_3^2$  &  $\beta_4^2$  &  $\beta_5^2$ \\
    \midrule
    1.6 & 72 & 605 & 13.735725166 & 7.554370381 & 6.381870597 & 1.213584120 & 0.599749013 \\
    0.8 & 91 & 763 & 13.735750272 & 7.554465932 & 6.381985947 & 1.213910004 & 0.599975761 \\
    0.4 & 349 & 2863 & 13.735759236 & 7.554485905 & 6.382002118 & 1.213986996 & 0.600030486 \\
    0.2 & 1509 & 12211 & 13.735759478 & 7.554486208 & 6.382002348 & 1.213987628 & 0.600030938 \\
    0.1 & 6108 & 49141 & 13.735759478 & 7.554486208 & 6.382002348 & 1.213987629 & 0.600030939 \\
    0.05 & 23842 & 191289 & 13.735759478 & 7.554486208 & 6.382002348 & 1.213987629 & 0.600030939 \\
    \bottomrule
  \end{tabular}
  \end{footnotesize}
  \caption{Eigenvalues $\beta^2$ of the 2D model for a sequence of meshes with mesh-size $h$, number of triangular elements ``ne'', and  number of degrees of freedom ``ndof''.}
  \label{tab:cross_verification_2D_results}
\end{table}

The corresponding eigenmodes are displayed (together with the 2D mesh)
in the bottom row of
Figure~\ref{fig:cross_verification_2D_results}. All mode plots show
the absolute value of the computed complex modes.  The shift in
localization from the core is clearly visible in these mode profiles.
We have also mapped the 2D mode profiles onto a 3D waveguide schematic
(top row of Figure~\ref{fig:cross_verification_2D_results}) so the
off-centered modal features can be placed correctly in the context of
the physical waveguide's bend. Note that the Frenet normal, being the
$\rx$-axis in the 2D model, points to the origin of the
helix. Therefore, modes which shift to the left in the 2D setting,
shift {\em outwards}, away from the coiling center,
when placed in the 3D waveguide.
\begin{figure}[ht!]
  \includegraphics[width=0.19\textwidth]{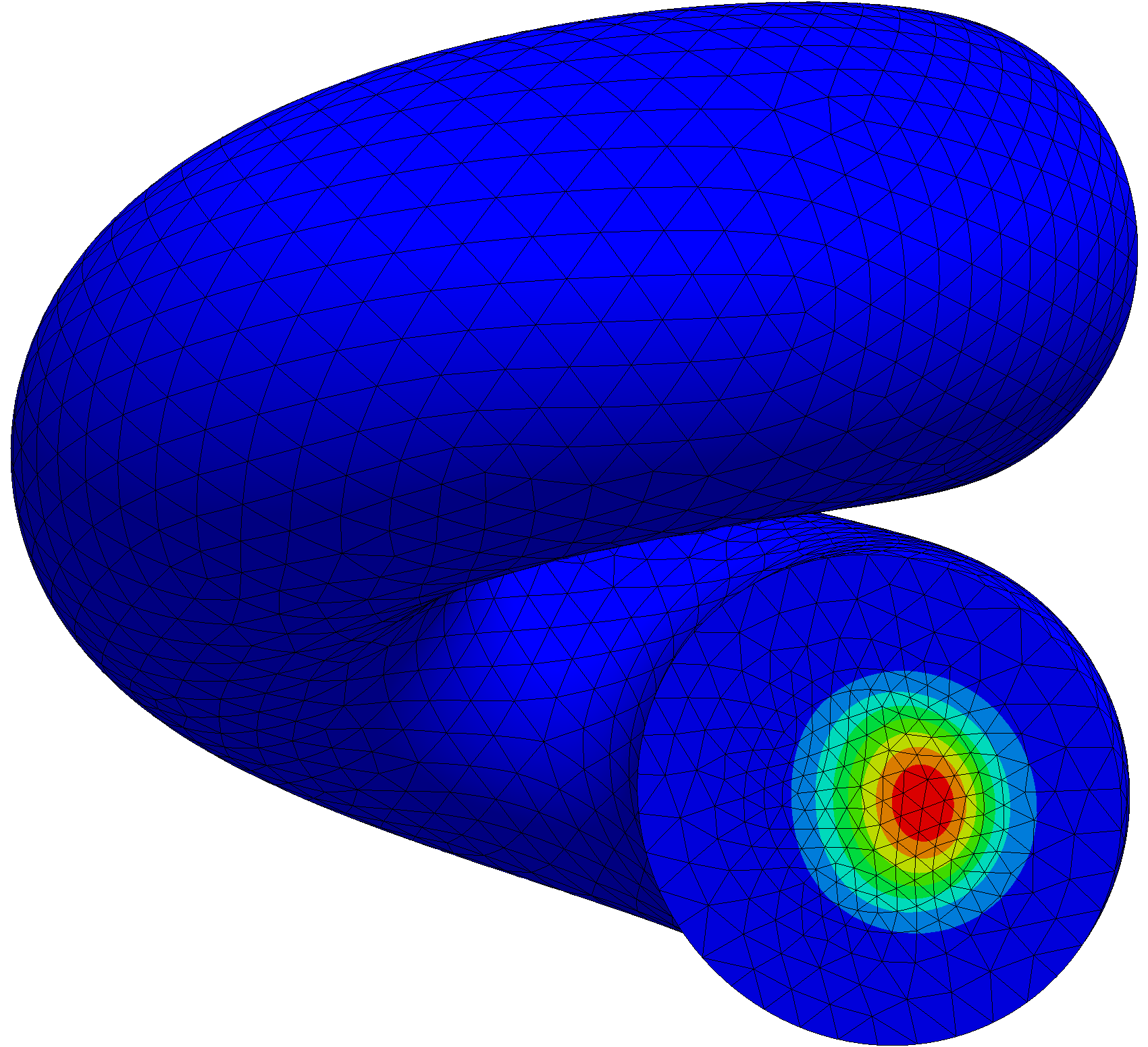}
  \includegraphics[width=0.19\textwidth]{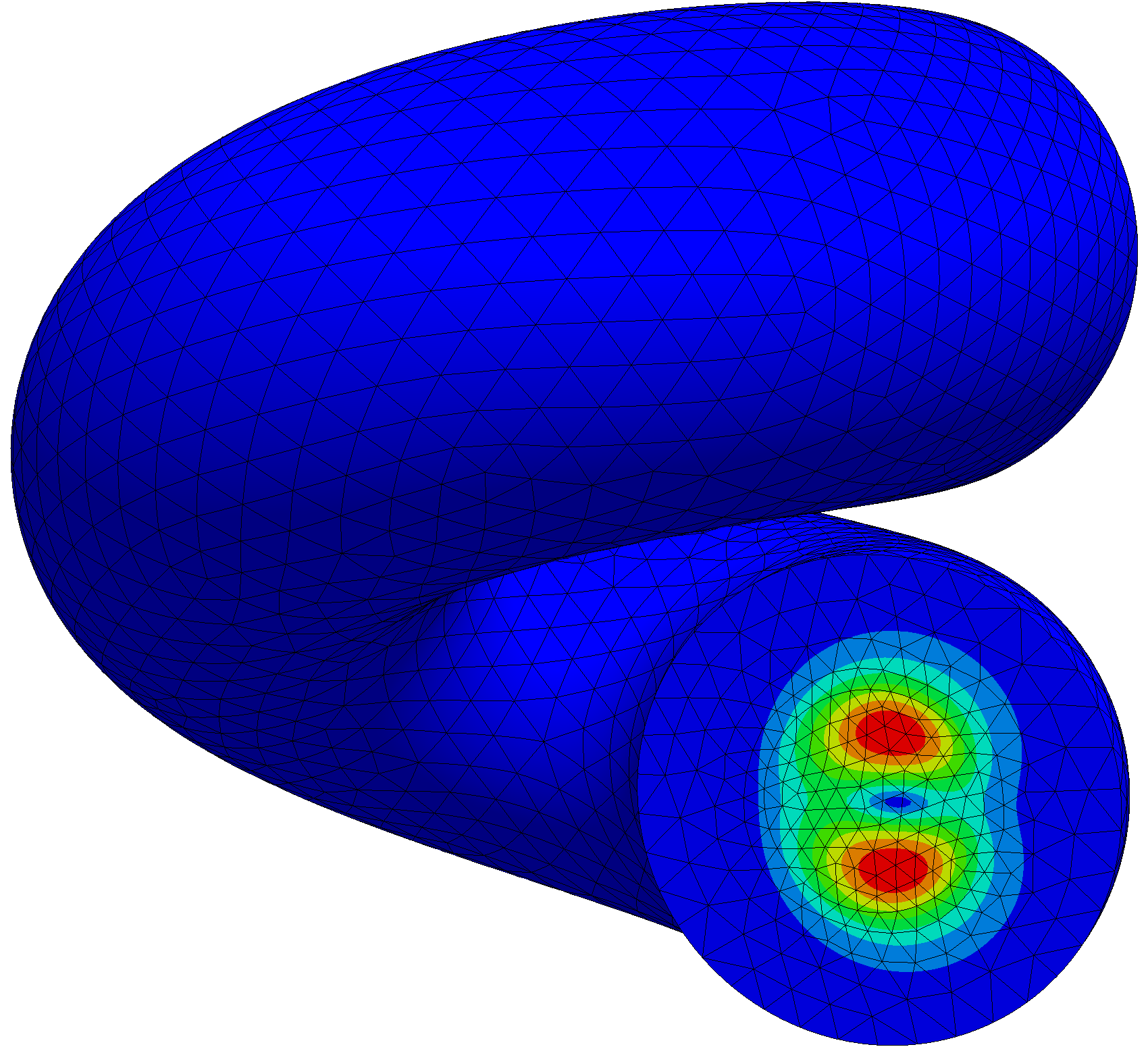}
  \includegraphics[width=0.19\textwidth]{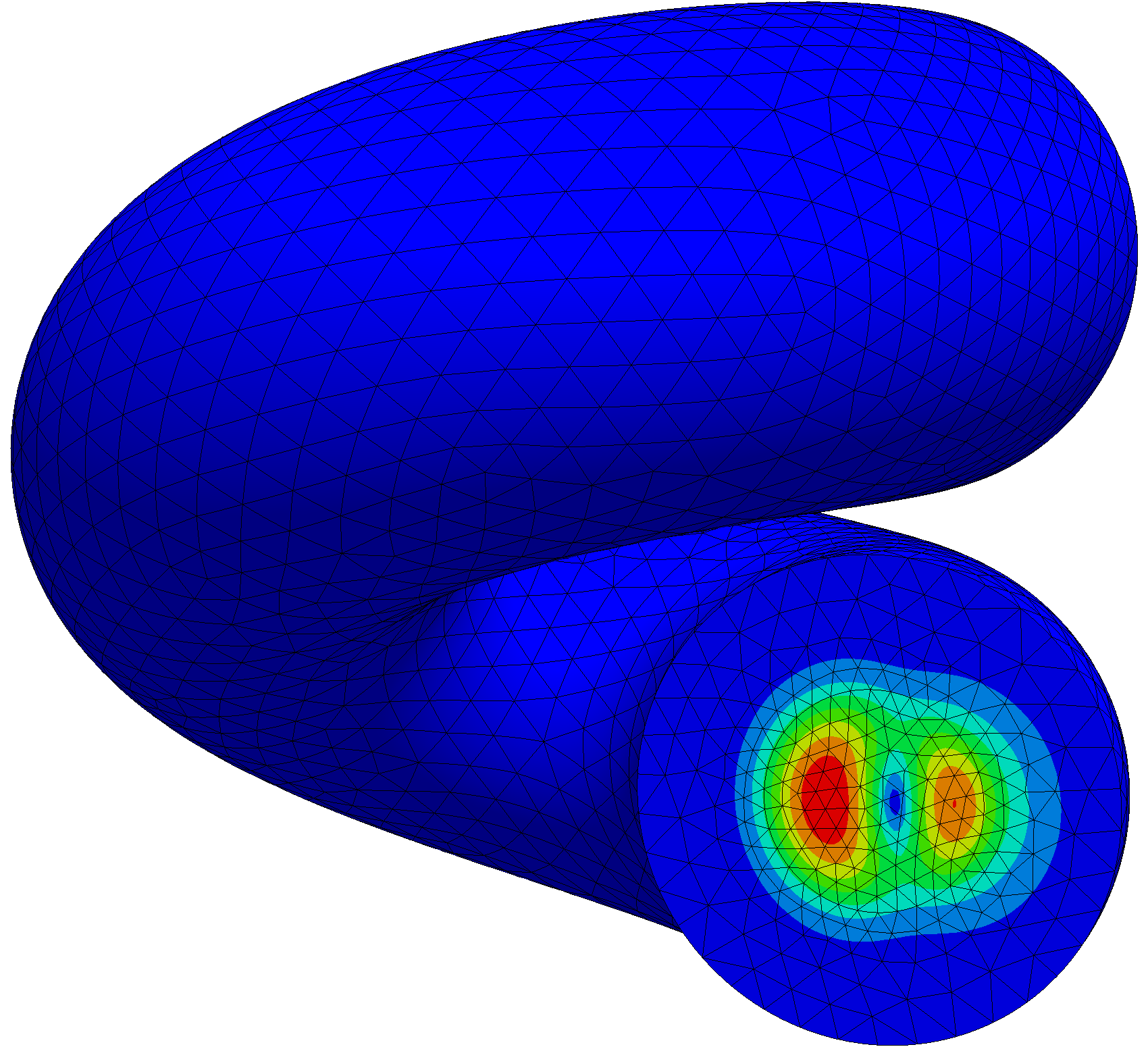}
  \includegraphics[width=0.19\textwidth]{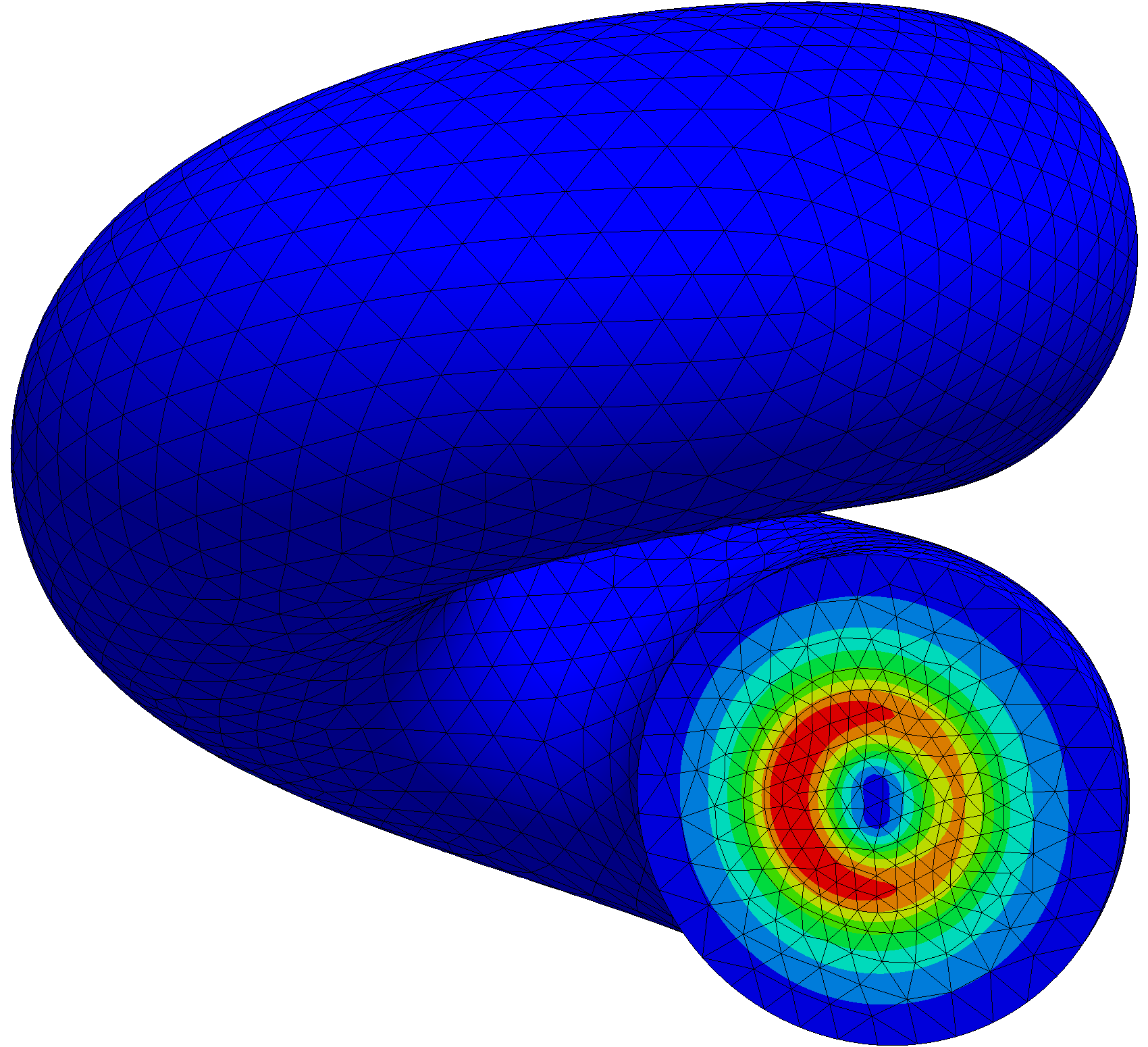}
  \includegraphics[width=0.19\textwidth]{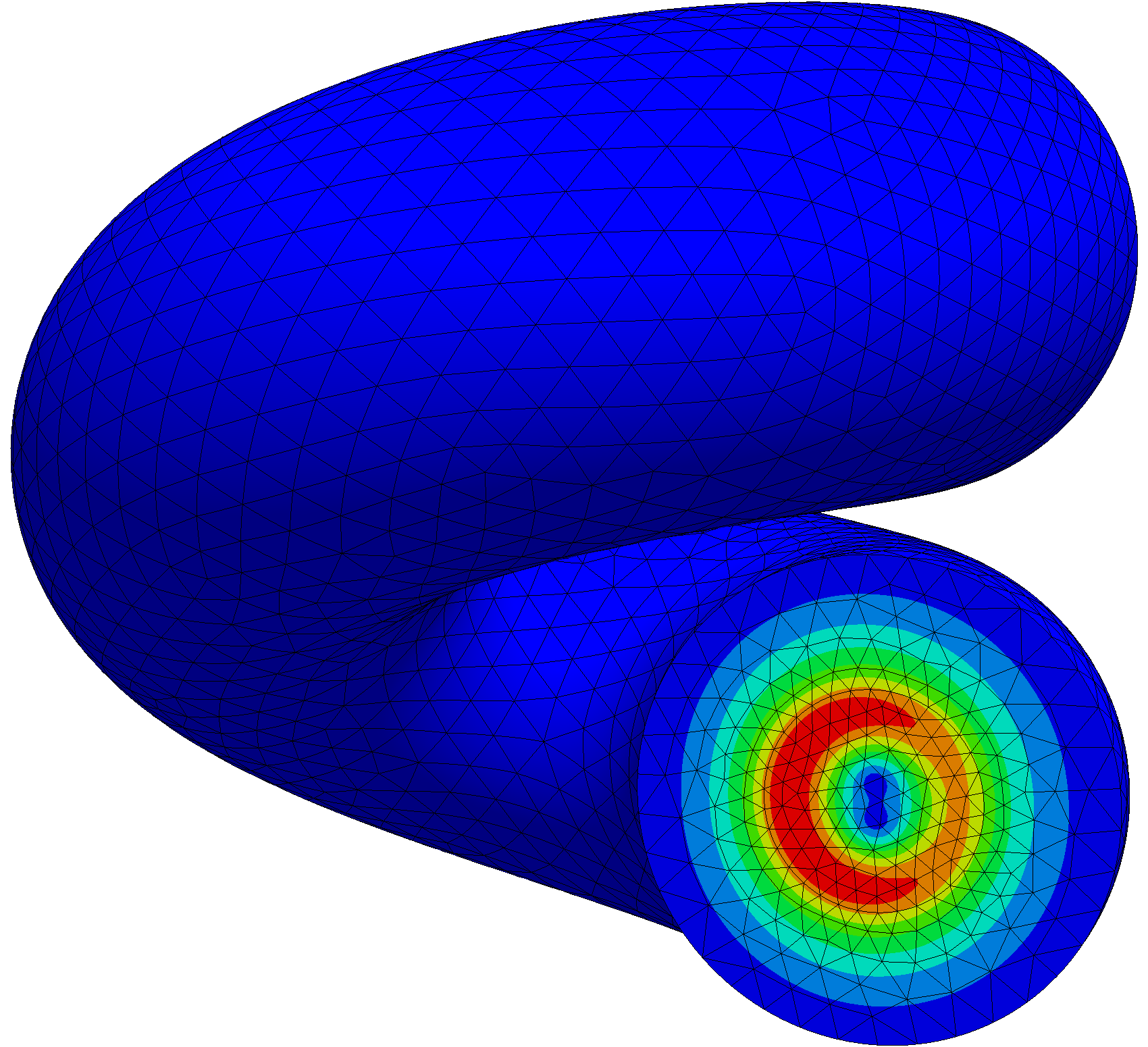}
  
  \includegraphics[width=0.19\textwidth]{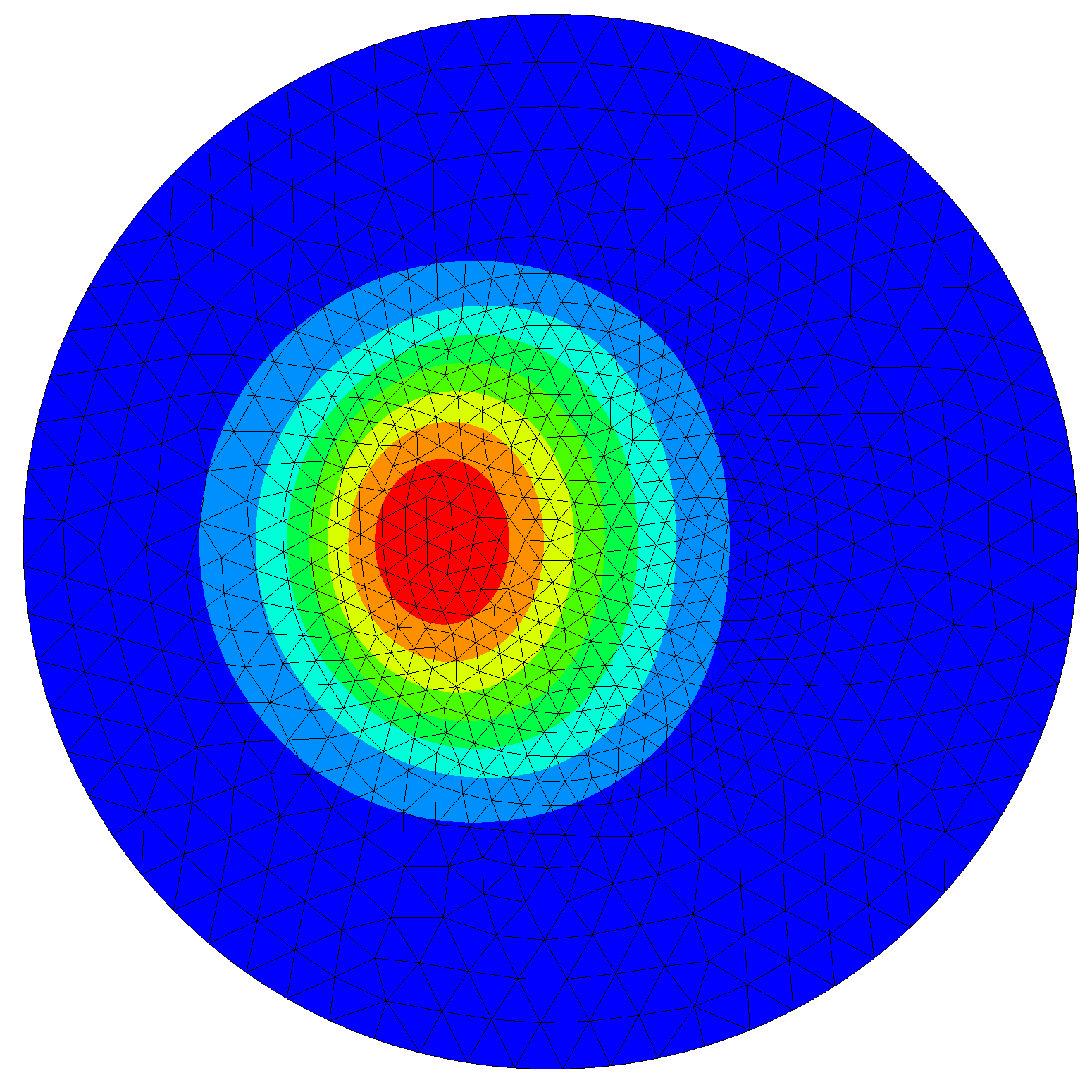}
  \includegraphics[width=0.19\textwidth]{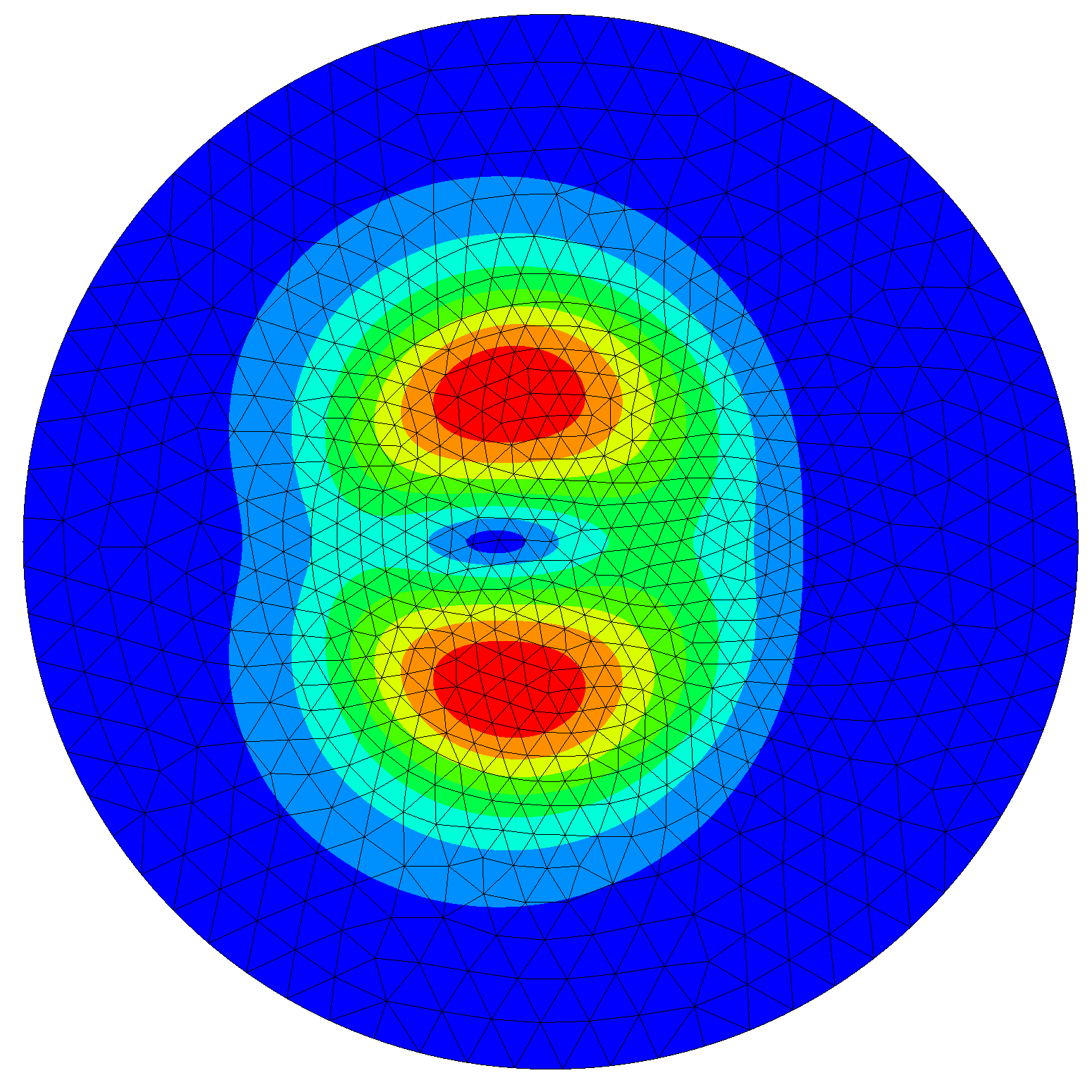}
  \includegraphics[width=0.19\textwidth]{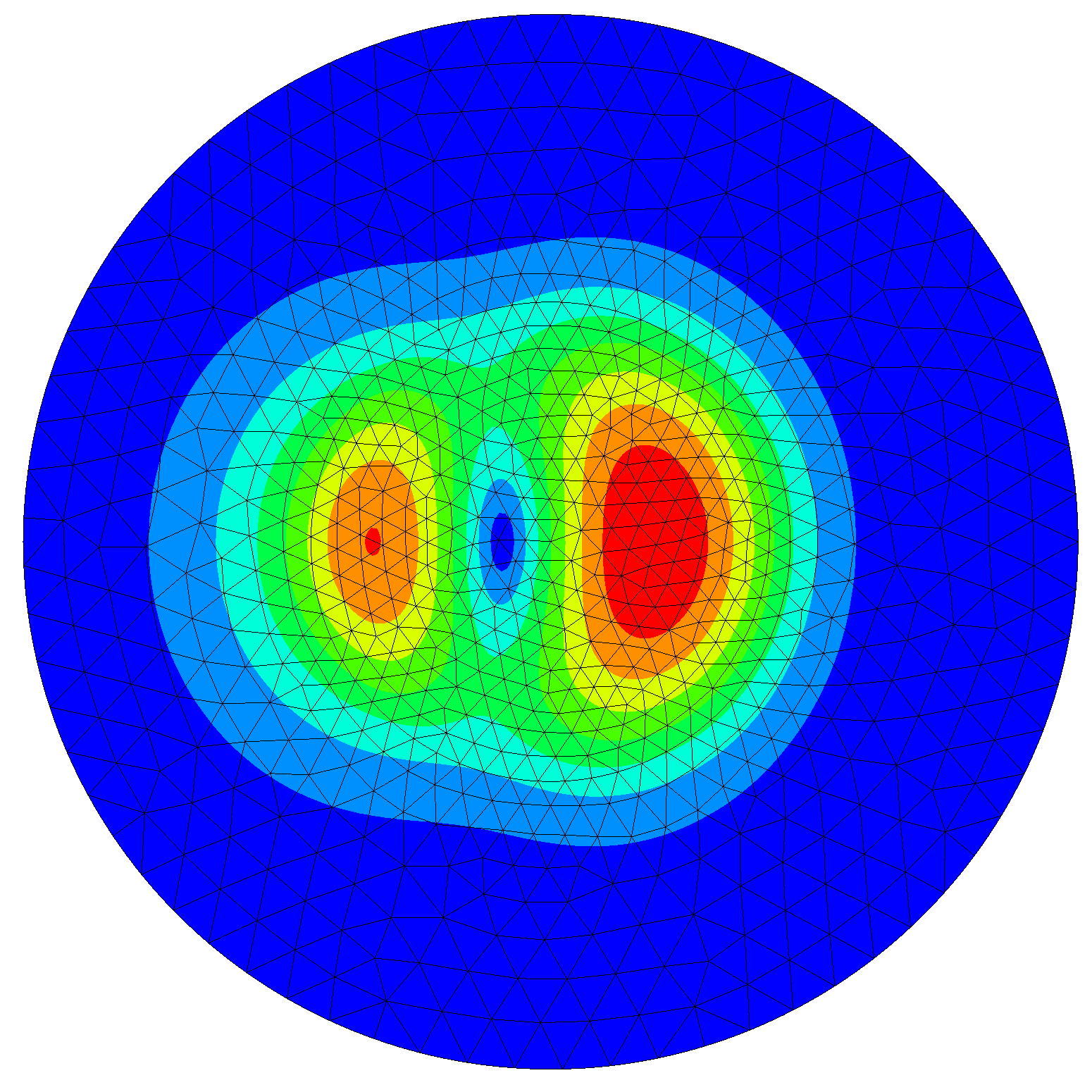}
  \includegraphics[width=0.19\textwidth]{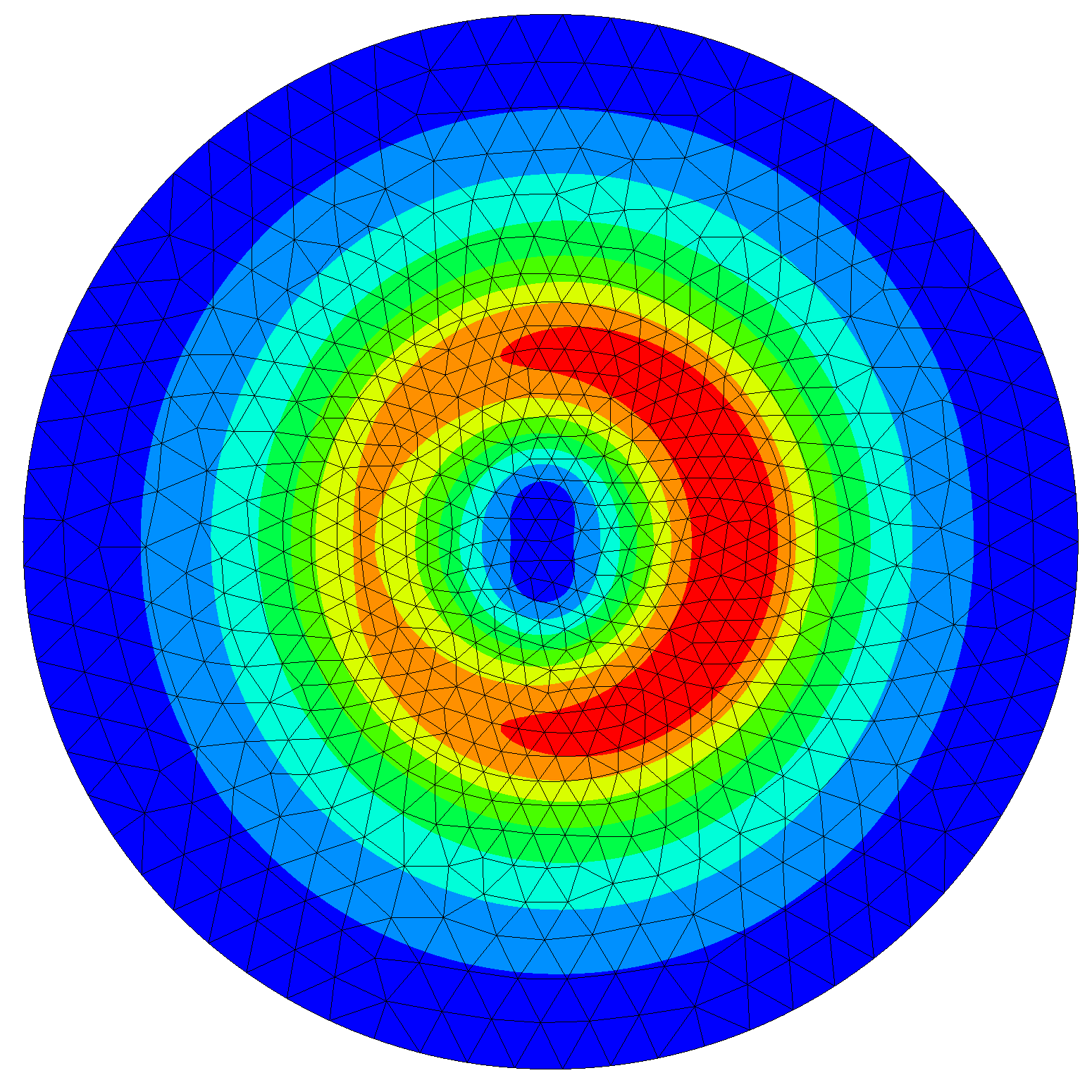}
  \includegraphics[width=0.19\textwidth]{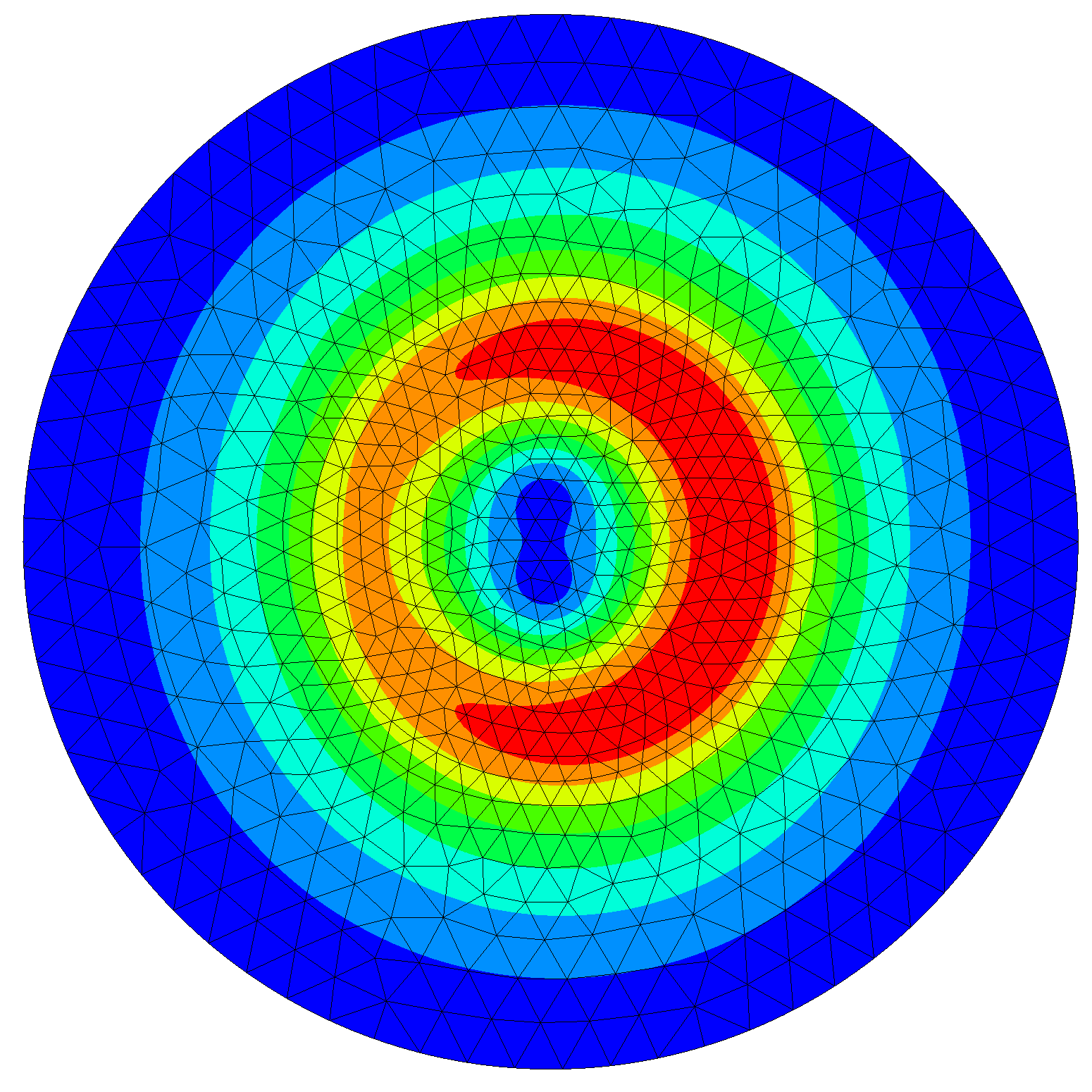}
  \caption{Magnitude of eigenmodes
    computed from the 2D model are illustrated in
    two ways: the bottom row shows them on the computed mesh; the top
    row shows them mapped onto the helically coiled waveguide. They
    are ordered left to right in descending order of their
    corresponding $\beta$-values in
    Table~\ref{tab:cross_verification_2D_results}. }
  \label{fig:cross_verification_2D_results}
\end{figure}

One reason for curving waveguides is to separate the fundamental mode
from the higher-order modes. For a straight fiber with parameters
\eqref{eq:cross_verification_parameters}, the eigenvalue of
fundamental mode (LP01) can be analytically computed to be
approximately 12.415638, whereas the next two higher order modes
(LP11) are near 7.188311 (a difference of about 5.22) and the final
two higher modes (LP21) are around 0.980593. From
Table~\ref{tab:cross_verification_2D_results} we see that upon helical
coiling, the fundamental mode's eigenvalue increased to 13.735759.
Its closest higher-order mode is at 7.554486, making a separation of
about 6.18. We see that the geometric effect of the helical coiling
leads to better separation of the fundamental mode from the
higher-order modes.

\subsection{Cross verification of results from the three eigenproblems} 
\label{ssec:cross-verification}

Next, we proceed to cross verify the above results from the 2D model
with 3D numerical
results from the straightened configuration~\eqref{eq:11} and the
physical helical waveguide~\eqref{eq:7}.
When computing with the straightened configuration we choose the scaling
$\sigma=4$ (see \S\ref{sec:helically_waveguide_qevp_comp}) to quarter the length of the cylinder in $\rz$-direction
(allowing for resources to be used to obtain better
cross section resolution). We also compute
the relative $L^2$-contribution of the $\rz$-derivative of the
eigenmodes, via
\begin{align}
  \mathrm{err}_{\rz}:=\|\hat\nabla \hat U\cdot e_{\rz}\|_{L^2(\hat\om)}/\|\hat\nabla \hat U\|_{L^2(\hat\om)}
  \label{eq:rel_z_contribution}
\end{align}
to isolate modes of interest.

In Table~\ref{tab:cross_verification_intermediate_results} we show the results for the straightened configuration for three different meshes. We setup FEAST to search around the 2D eigenmodes with a search radius of 1 percent of the eigenvalue. In contrast to the 2D computations, we now obtain several eigenfunctions close to the desired eigenvalues. However, we found that
 \eqref{eq:rel_z_contribution} goes to zero
only for five of them.
The eigenvalues for those five modes converge to the 2D results in Table~\ref{tab:cross_verification_2D_results}. This can be viewed as  a numerical verification of our assumption~\eqref{eq:ansatz'}
that there exists $\rz$-independent eigenfunctions in the straightened configuration.
The other computed modes, which have a high $\rz$-derivative, are not propagating modes (but appear to be modes like standing waves). Their eigenvalues cluster  around the relevant ones making it difficult to separate them looking only at the eigenvalues. However, \eqref{eq:rel_z_contribution} does effectively filter  out the propagating modes of interest. 
The five modes are shown in Figure~\ref{fig:straightened_shift}. We see the fundamental mode shifts along the Frenet normal, which is the $\rx$-axis in the straightened configuration.

\begin{table}[ht!]
  \centering
  \begin{footnotesize}   
\begin{tabular}{c|cc||cc||cc}
  \toprule
    &ne=1376 & ndof=15402 &  ne=9650 & ndof=105604 &  ne=23309 & ndof=253058  \\
    & $\beta^2$ & $\mathrm{err}_{\rz}$ & $\beta^2$ & $\mathrm{err}_{\rz}$ & $\beta^2$ & $\mathrm{err}_{\rz}$ \\
  \midrule
  $\beta_1^2$ & 13.733441557  &  0.076566 &  13.735758524  &  0.000252 &  13.735759215  &  0.000125 \\
  & 13.763696206  &  0.866226 &  13.802052436  &  0.834760 &  13.522030691  &  0.957915 \\
   & &   &  &   & 13.802079895  &  0.834752 \\
  &  &   &  &   & 13.664261602  &  0.993535 \\
   \midrule
  $\beta_2^2$ & 7.547275365  &  0.522232 &  7.554484798  &  0.000232 &  7.554485858  &  0.000112 \\
  & 7.701110444  &  0.965881 &  7.505431057  &  0.904869 &  7.505437106  &  0.904868 \\
  & 7.494815328  &  0.961366 &   &  &   &  \\
  \midrule
  $\beta_3^2$ & 6.374972271  &  0.170528 &  6.382001254  &  0.000220 &  6.382002007  &  0.000110 \\
   & &   & 6.338427995  &  0.972767 &  6.261246381  &  0.991823 \\
   \midrule
   $\beta_4^2$ & 1.196704393  &  0.126262 &  1.213983256  &  0.000268 &  1.213986873  &  0.000130 \\
  & 1.232101019  &  0.933150 &  1.202981316  &  0.383337 &  1.202986086  &  0.383336 \\
  \midrule
  $\beta_5^2$ & 0.594864213  &  0.389903 &   0.600027824  &  0.000271 &  0.600030404  &  0.000130 \\
   & &   & 0.607803431  &  0.383419 & 0.607806776  &  0.383419 \\
  \bottomrule
  \end{tabular}
\end{footnotesize}
  
\caption{Eigenvalues $\beta^2$ of the 3D
  straightened configuration and their  $\mathrm{err}_{\rz}$ (defined in~\eqref{eq:rel_z_contribution}) on
  a mesh with 1376 elements (15402 ndof, left), a mesh with 9650 elements (105604 ndof, middle), and a mesh with 23309 elements (253058 ndof, right).}
  \label{tab:cross_verification_intermediate_results}
\end{table}

\begin{figure}[ht!]
  \centering
  \includegraphics[width=0.19\textwidth]{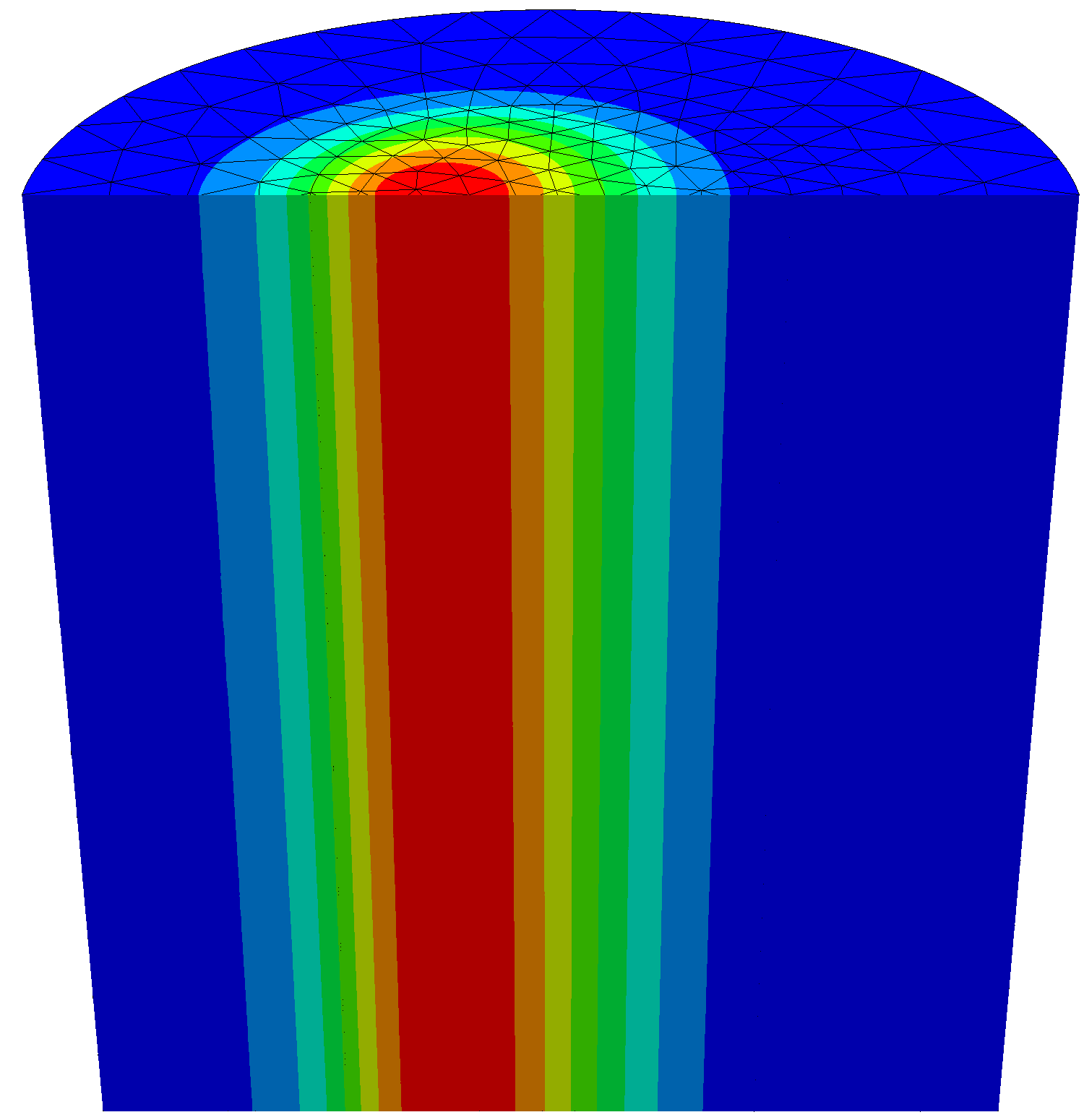}
  \includegraphics[width=0.19\textwidth]{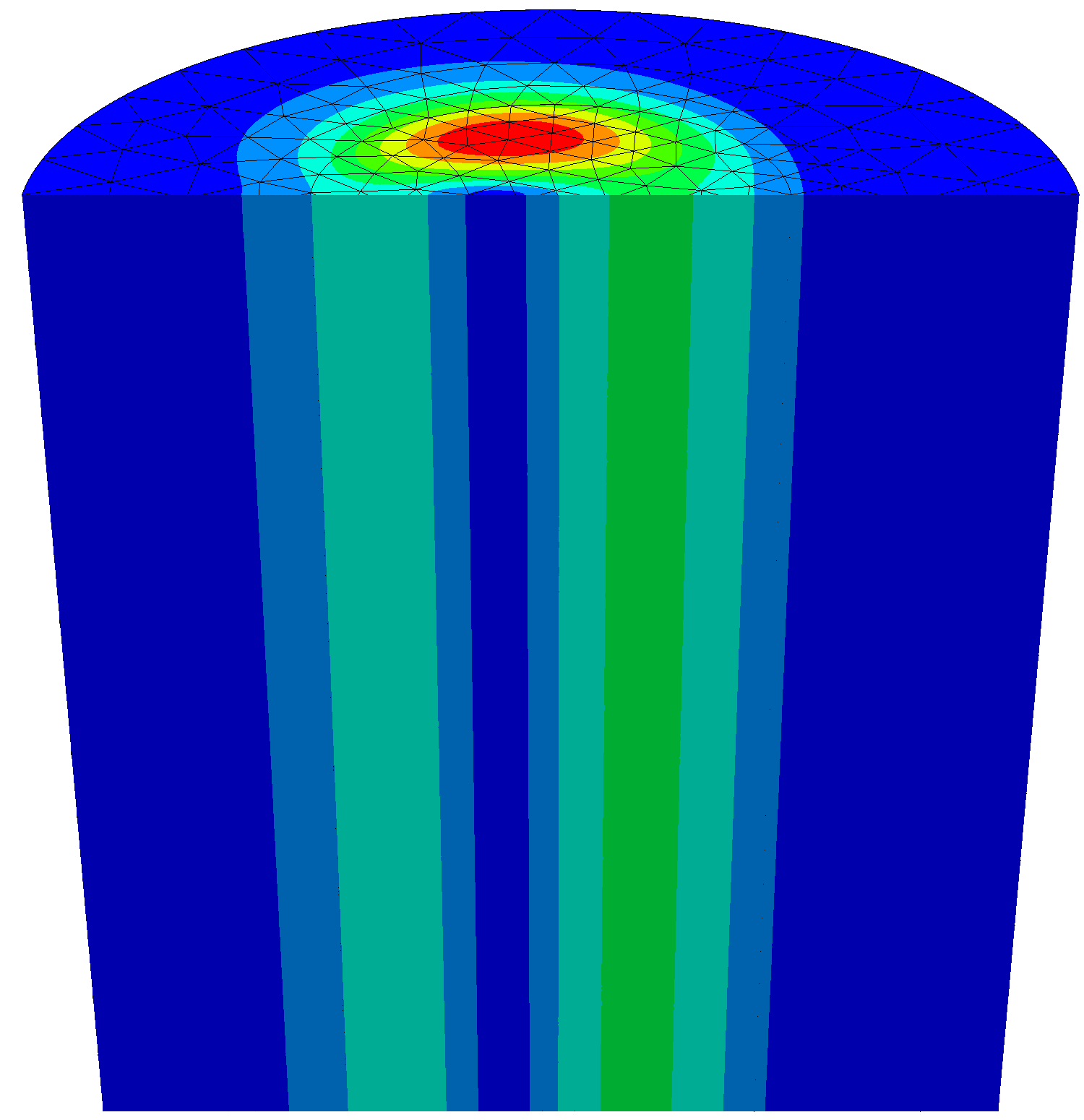}
  \includegraphics[width=0.19\textwidth]{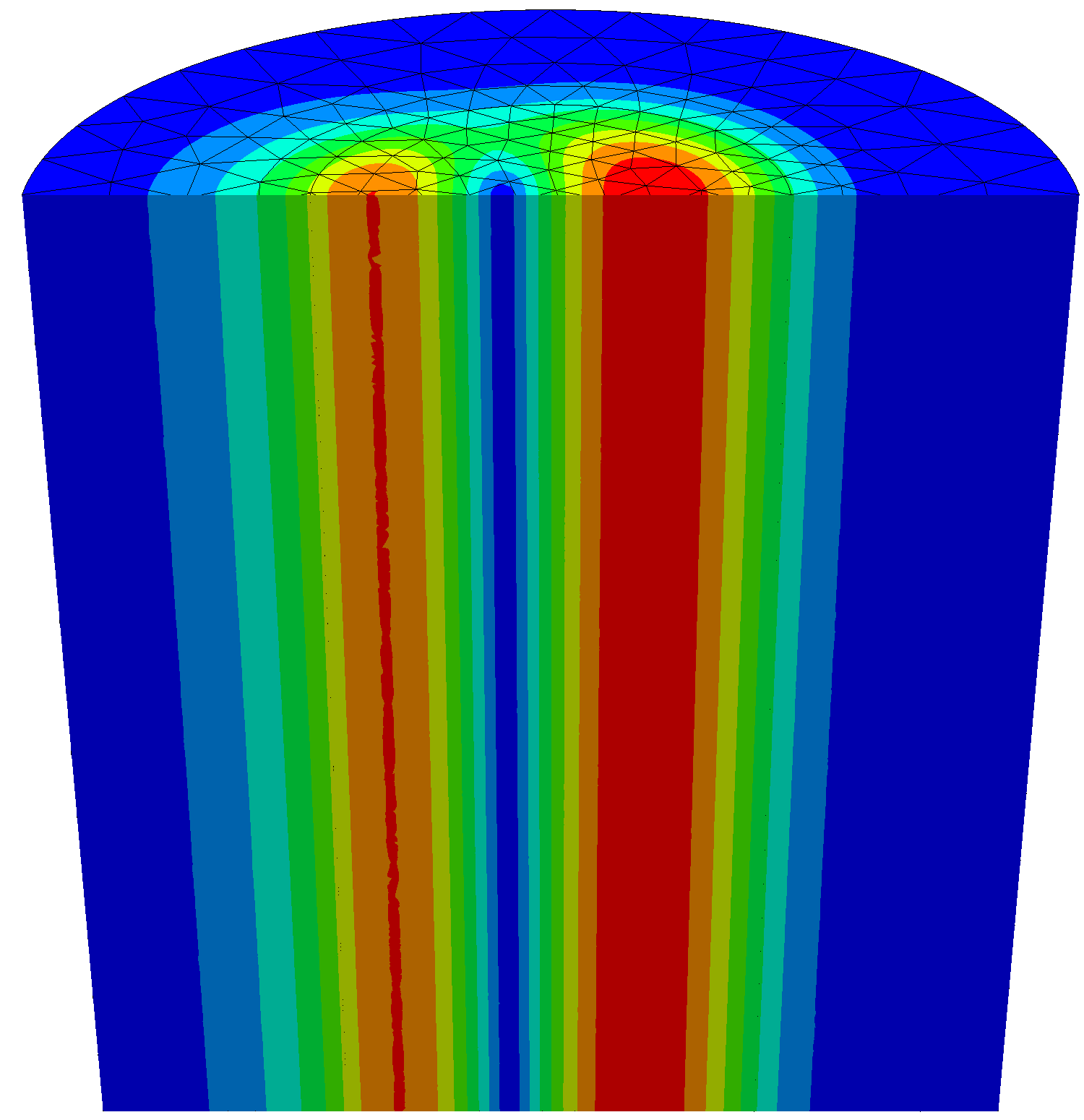}
  \includegraphics[width=0.19\textwidth]{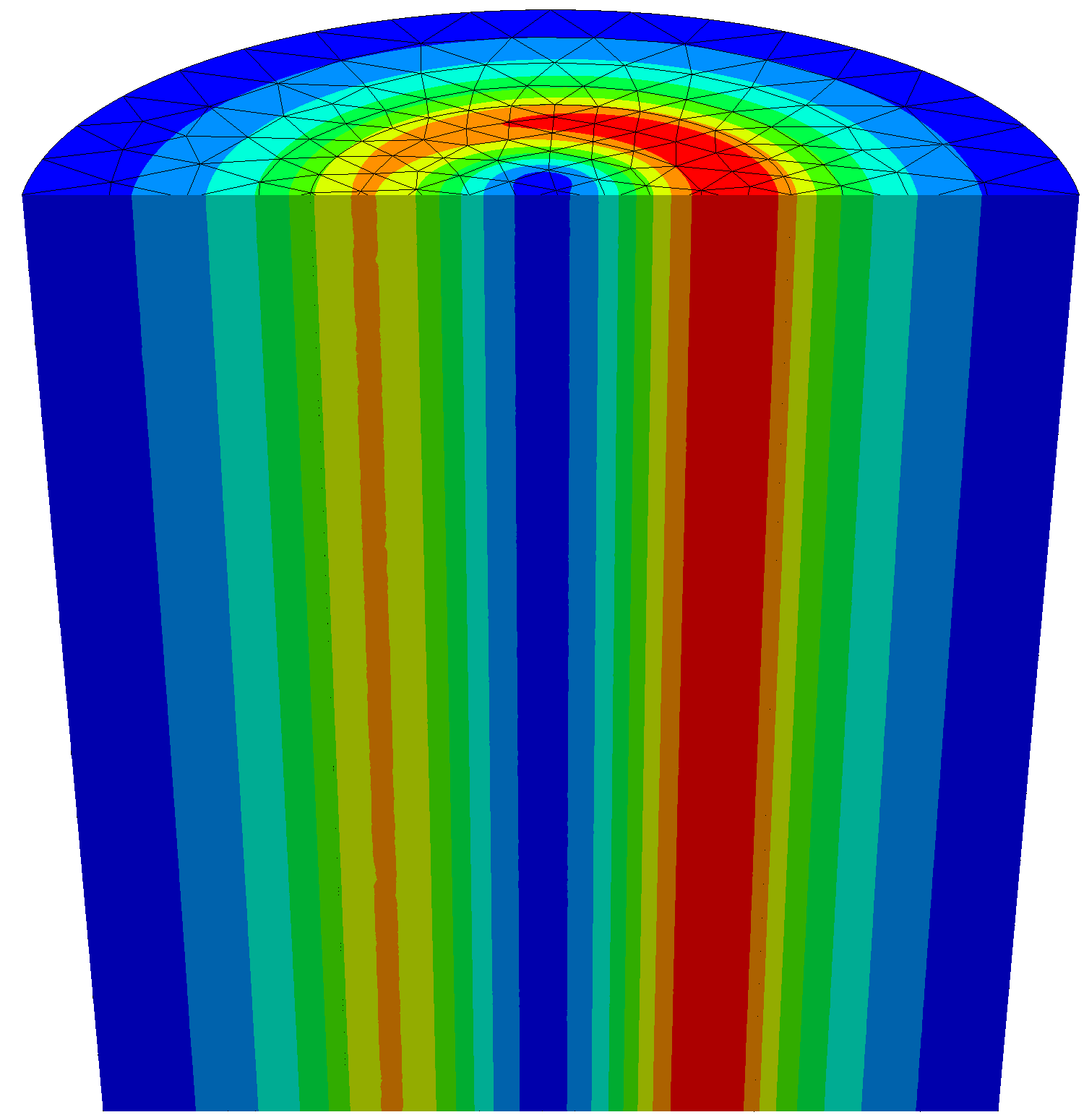}
  \includegraphics[width=0.19\textwidth]{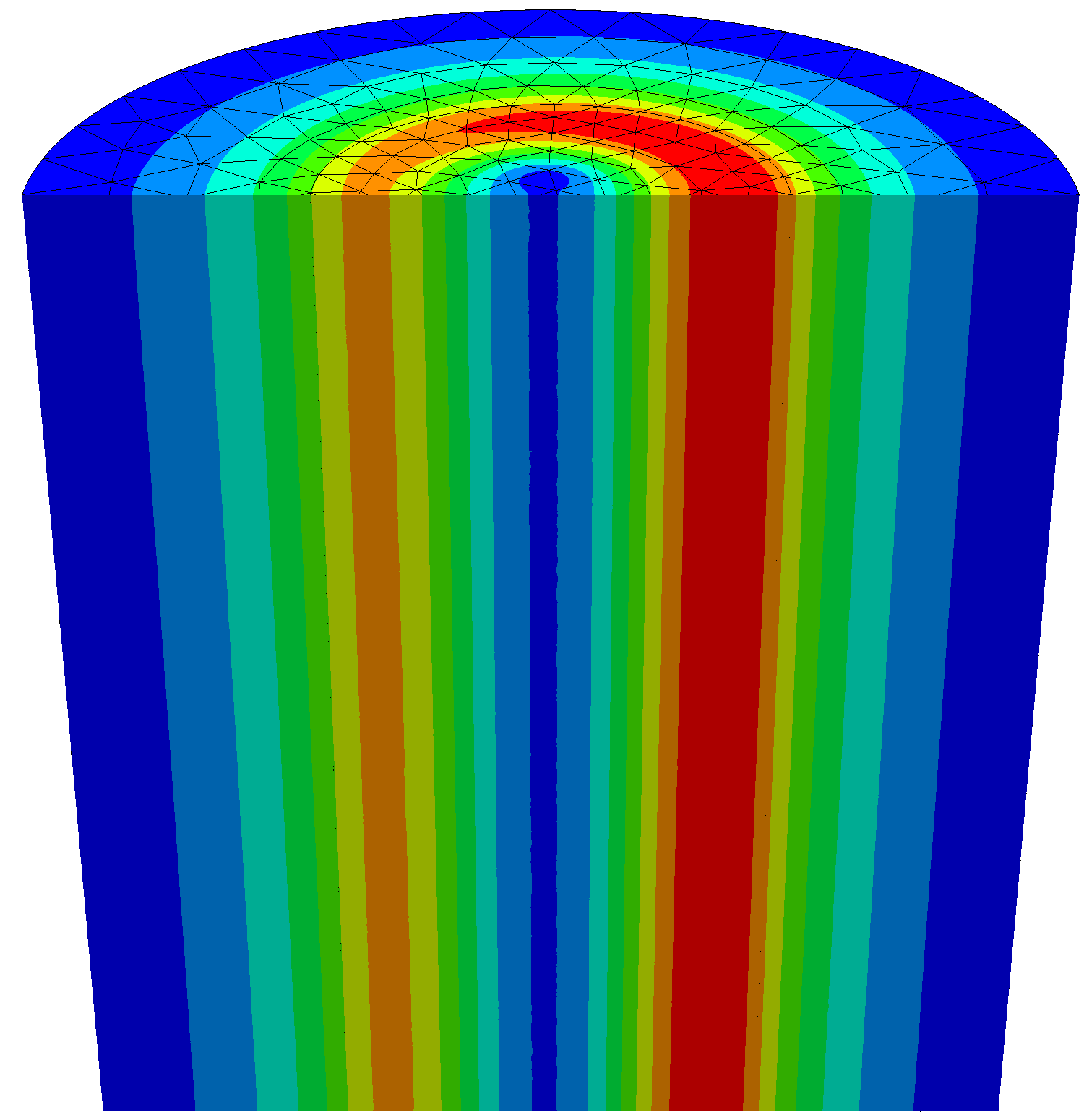}
  \caption{Visualization of selected modes and their off-center shift in the  straightened configuration on the finest mesh. The  $\rz$-independence of these modes is apparent.}
  \label{fig:straightened_shift}
\end{figure}

As an additional cross verification,  we compare the 2D results to the 3D modes computed using the original helical waveguide geometry~\eqref{eq:7}. Again, three different meshes are used and the results from all three are  in Table~\ref{tab:cross_verification_undeformed_results}. 
To filter out propagating modes, we implement a numerical verification of~\eqref{eq:ansatz-transverse} using 
the relative measure 
\begin{align}
  \mathrm{err}_Z := \|\nabla U\cdot \tilde{Z}_h\|_{L^2(\om)}/\|\nabla U\|_{L^2(\om)},
  \label{eq:rel_z_contribution2}
\end{align}
where
$\tilde{Z}_h:= Z_h/\|Z_h\|$ and 
$Z_h$ is as described in \S\ref{sec:helically_waveguide_qevp_comp}.
Filtering out the modes using \eqref{eq:rel_z_contribution2} we again
find five modes.  As can be seen from
Table~\ref{tab:cross_verification_undeformed_results}, the eigenvalues
computed on the bent waveguide do appear to converge to the 2D results
of Table~\ref{tab:cross_verification_2D_results}. The 3D results
appear to be not as accurate because the mesh is comparably
coarse. Also note that finer meshing of the waveguide appear to
produce many more irrelevant eigenfunctions near the relevant ones,
but we are able to use~\eqref{eq:rel_z_contribution2} to identify the
propagating modes and only their eigenvalues appear to converge to the
corresponding 2D eigenvalues.  Needless to mention are the tremendous
computational gains brought about by the 2D reduction (e.g., using a Mac
M2 processor, the 2D computations for Table~\ref{tab:cross_verification_2D_results} took about 30 second, whereas
the 3D computations in the straightened configuration to generate the results in Table~\ref{tab:cross_verification_intermediate_results} took about 1.25 hours).
\begin{table}[ht!]
  \centering
  \begin{footnotesize}
\begin{tabular}{c|cc||cc||cc}
  \toprule
  &ne=4607 & ndof=52318 &  ne=7904 & ndof=88372 &  ne=17342 & ndof=191630  \\
  & $\beta^2$ & $\mathrm{err}_{Z}$ & $\beta^2$ & $\mathrm{err}_{Z}$ & $\beta^2$ & $\mathrm{err}_{Z}$ \\
  \midrule
  $\beta^2_1$&13.733206936     &  0.0953 &  13.734946168    &  0.0037 &  13.735732309   &  0.0011 \\
  &13.791778155     &  0.3427 &  13.480673401   &  0.7176 &  13.992433105     &  0.8126 \\
  &13.708566140    &  0.9150 &  13.518840973   &  0.6388 &  13.950479969    &  0.8705 \\
  & &   & 14.004442036    &  0.9569 &  13.499207778   &  0.7173 \\
  & &   & 13.860708691     &  0.9292 &  13.535777888    &  0.6385 \\
  & &   &   13.800827527   &  0.3403 &  13.677100791   &  0.9296 \\
  & &   &  &   & 13.753184231    &  0.9633 \\
  &    &   &     &   & 13.802025335  &  0.3401 \\
   \midrule
  $\beta^2_2$& 7.547335104    &  0.0237 &  7.553670129  &  0.0060 &  7.554449886  &  0.0008 \\
  & 7.680087108  &  0.9041 &  7.688733573   &  0.8144 &  7.675878453    &  0.7443 \\
  & 7.503380976   &  0.4295 &  7.570153365    &  0.9518 &  7.647479028   &  0.8140 \\
  & &   & 7.504762148   &  0.4299 &  7.450916889   &  0.9057 \\
  & &   & 7.500999833   &  0.9060 &  7.502148982    &  0.9528 \\
  & &   &   &   & 7.505415507     &  0.4289 \\
   \midrule
  $\beta^2_3$& 6.374600241   &  0.0190 &  6.381326289  &  0.0031 &  6.381974976   &  0.0007 \\
  \midrule
  $\beta^2_4$& 1.198328310     &  0.0244 &  1.212183510  &  0.0040 &  1.213903444  &  0.0009 \\
  & &  & 1.201164013   &  0.1087 &  1.202903051    &  0.1085 \\
   \midrule
  $\beta^2_5$&  0.588948792   &  0.0252 &  0.598743993    &  0.0039 &  0.599970767   &  0.0009 \\
  & 0.596613234  &  0.1115 & 0.606525591   &  0.1091 &  0.607748146   &  0.1090 \\
  \bottomrule
  \end{tabular}
\end{footnotesize} 
\caption{Eigenvalues $\beta^2$ of the 3D physical helical waveguide
  and their $\mathrm{err}_Z$ (defined in
  \eqref{eq:rel_z_contribution2})
  on a mesh with 4607 elements (52318 ndof, left),
  a mesh with 7904 elements (88372 ndof, middle), and
  a mesh with 17342 elements (191630 ndof, right).}
\label{tab:cross_verification_undeformed_results}
\end{table}

\subsection{Comparing quadratic and linear eigenproblems for toroidal waveguide}
\label{ssec:comparison-qevp-levp}

Consider the toroidal limit case $\bhel \to 0$ case described in \S\ref{sec:torus}.
The
limiting values of the geometric coefficients are described
in~\eqref{eq:J-limit-torus}. We numerically verify that in this case the 2D quadratic eigenproblem~\eqref{eq:12} is equivalent to the linear eigenproblem~\eqref{eq:torus_lin_evp}. We set the following parameter values 
\begin{align}
  \label{eq:torus-eg-prm}
  r_0=1,\quad r_1=2.2,\quad \ahel=3,\quad \bhel=0,\quad n_0=4,\quad n_1=1,\quad k=1.
\end{align}
Note that this fiber has exactly the same $V$-number $V=15$ as the previously considered case~\eqref{eq:cross_verification_parameters}.

\begin{figure}[ht!]
  \centering
  \includegraphics[width=0.19\textwidth]{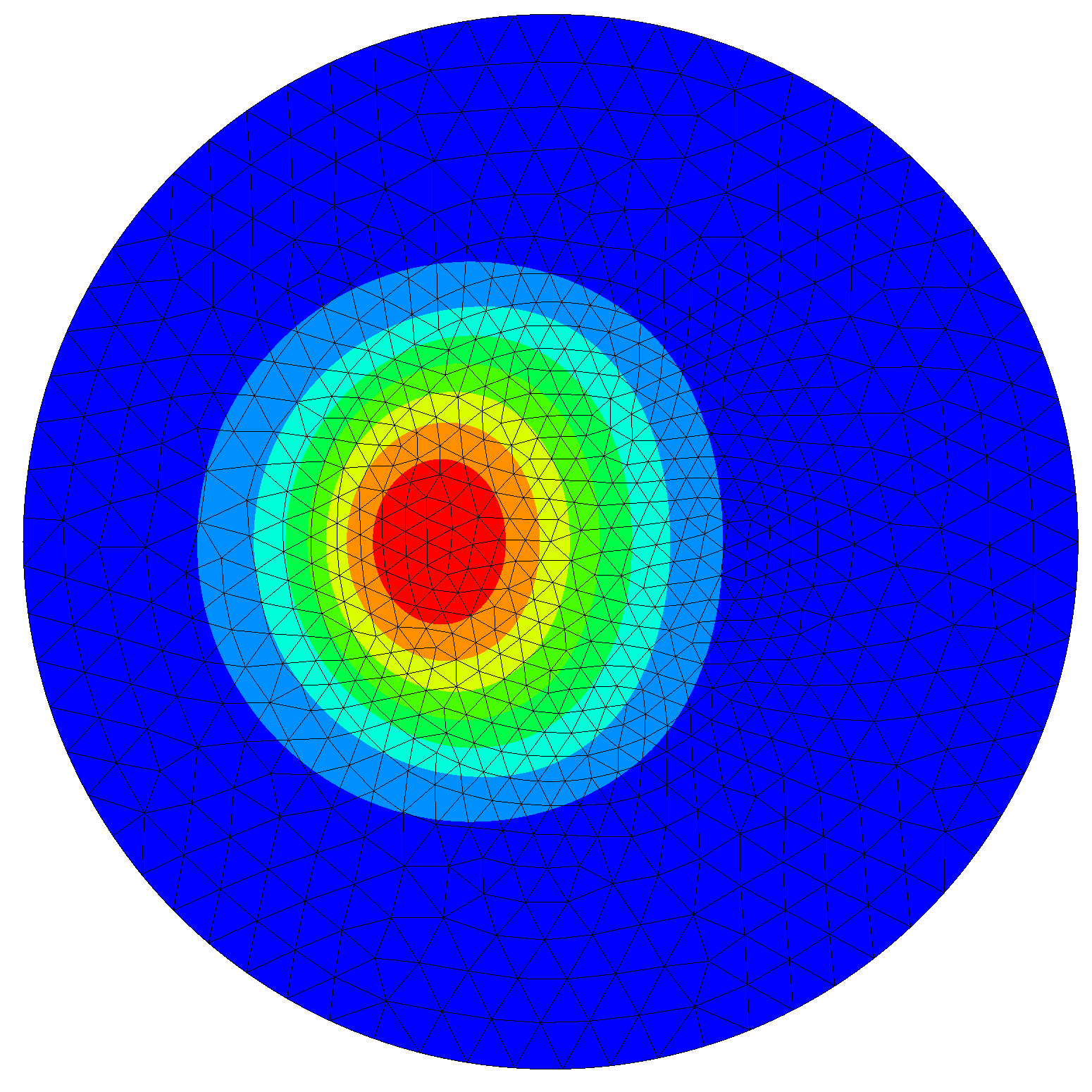}
  \includegraphics[width=0.19\textwidth]{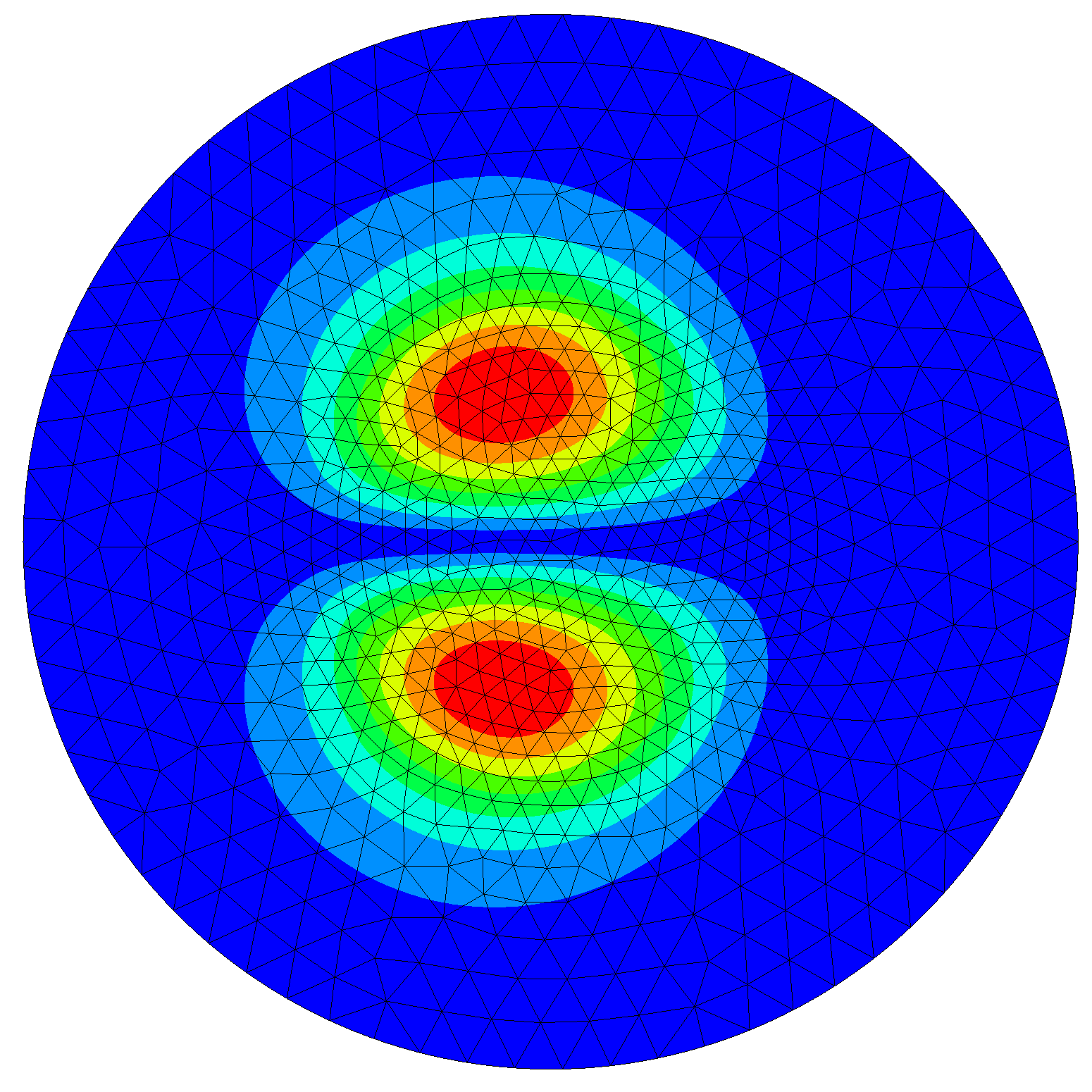}
  \includegraphics[width=0.19\textwidth]{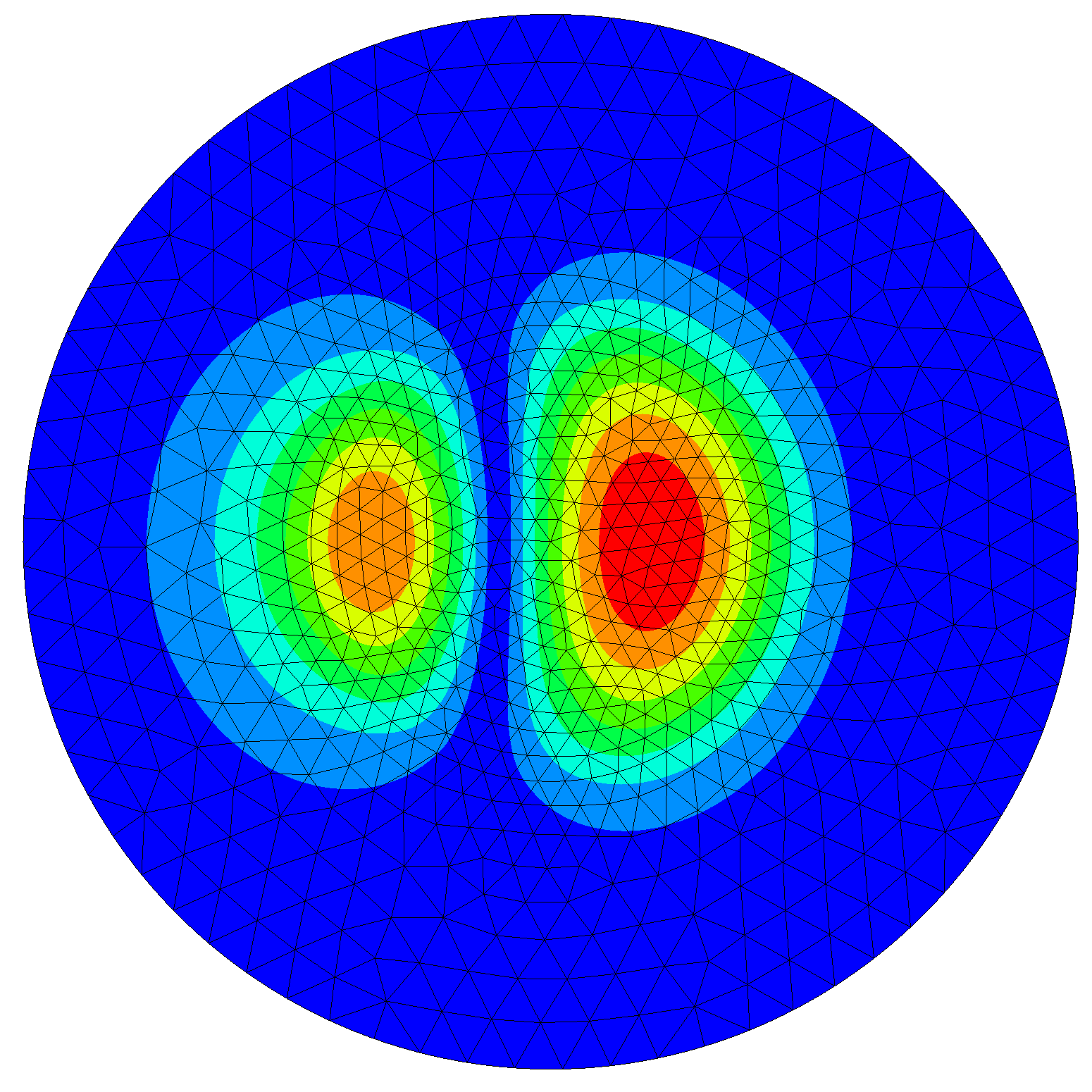}
  \includegraphics[width=0.19\textwidth]{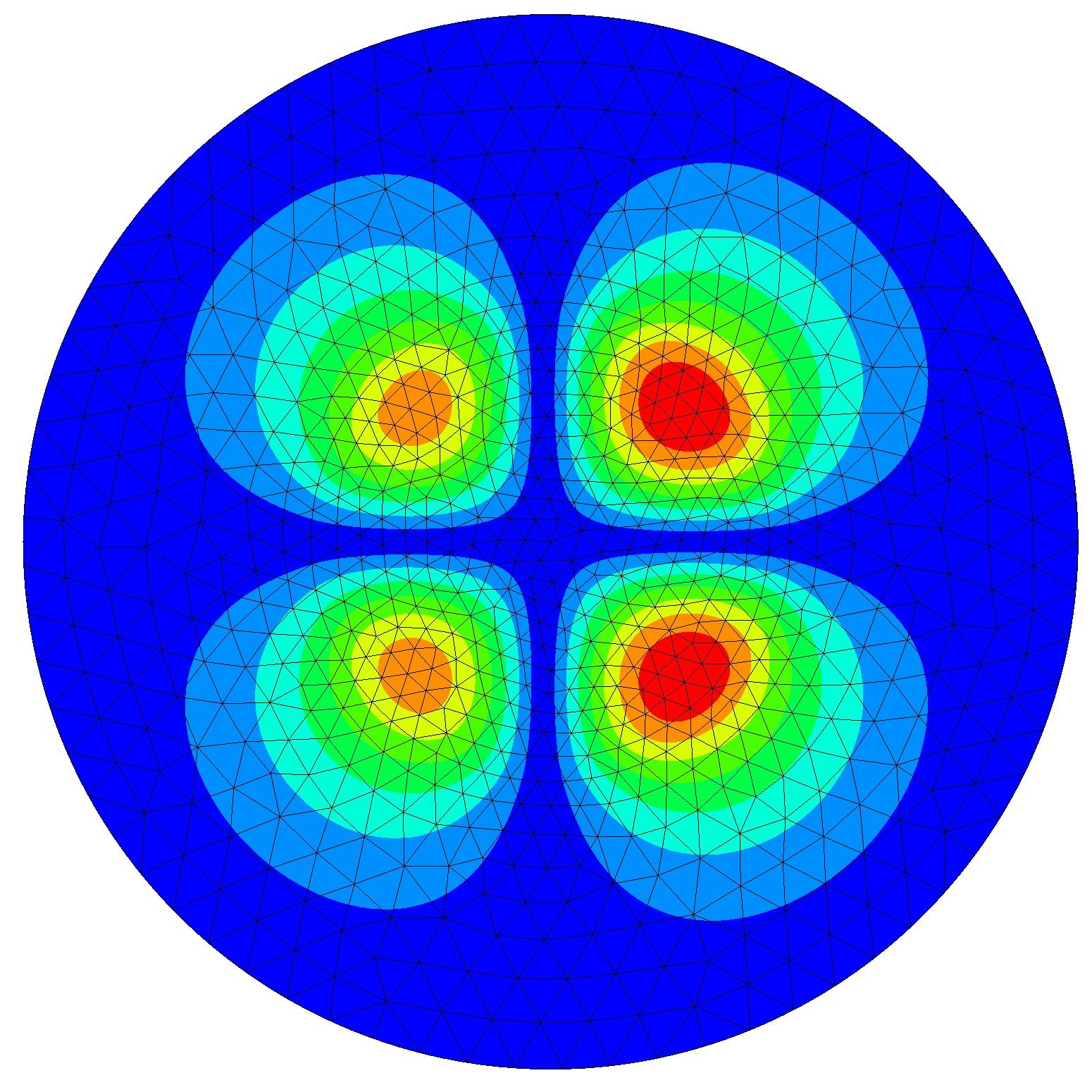}
  \includegraphics[width=0.19\textwidth]{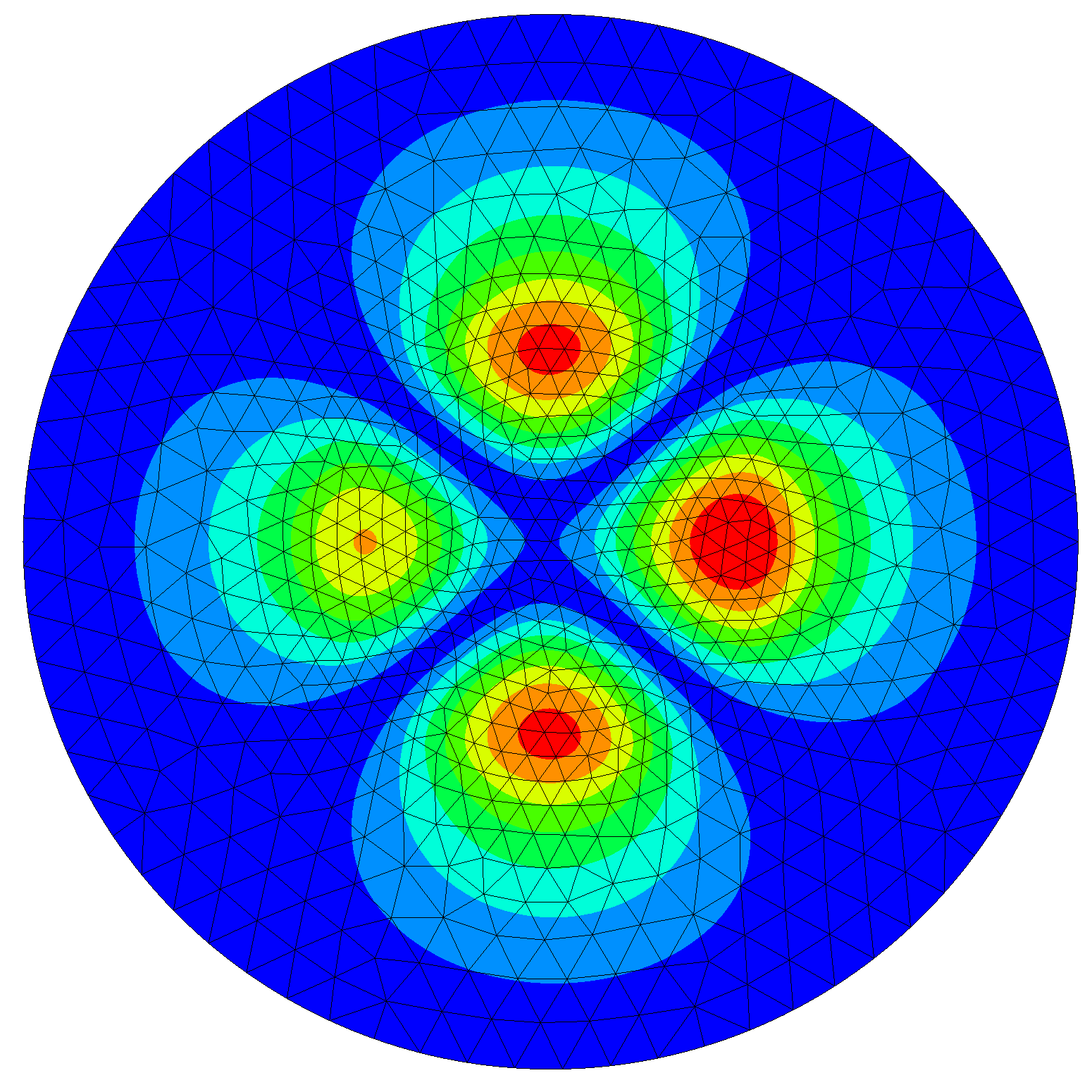}

\caption{Magnitudes of the five  eigenmodes on the cross section of the toroidal waveguide of \eqref{eq:torus-eg-prm}. Results from the quadratic and linear eigenproblems are indistinguishable.}
\label{fig:torus_result}
\end{figure}

\begin{table}[ht!]
  \centering
  \begin{footnotesize}
  \begin{tabular}{c|ccccccc}
    \toprule
    h  &  ne  &  ndof  &  $\beta_1^2$  &  $\beta_2^2$  &  $\beta_3^2$  &  $\beta_4^2$  &  $\beta_5^2$ \\
    \midrule
    1.6 & 72 & 605 & 13.896655945 & 7.417666545 & 6.466435528 & 0.866180693 & 0.865656552 \\
    0.8 & 91 & 763 & 13.896678708 & 7.417750941 & 6.466551350 & 0.866426447 & 0.865937597 \\
    0.4 & 349 & 2863 & 13.896688049 & 7.417771578 & 6.466567837 & 0.866493270 & 0.865999900 \\
    0.2 & 1509 & 12211 & 13.896688301 & 7.417771898 & 6.466568067 & 0.866493808 & 0.866000431 \\
    0.1 & 6108 & 49141 & 13.896688301 & 7.417771898 & 6.466568067 & 0.866493809 & 0.866000431 \\
    \bottomrule
    \end{tabular}
  \end{footnotesize}
    \caption{Eigenvalues $\beta^2$ corresponding to the modes in Figure~\ref{fig:torus_result} for the toroidal waveguide of \eqref{eq:torus-eg-prm}.
    }
    \label{tab:torus_results}
\end{table}

Figure~\ref{fig:torus_result} shows the modes computed from the 2D
linear eigenproblem~\eqref{eq:torus_lin_evp}, where the limiting
coefficient expressions for the torus from~\eqref{eq:J-limit-torus}
are substituted. As the results from the quadratic eigenvalue
problem~\eqref{eq:12} coincide up to numerical rounding error
precision, we do not show them. A convergence study can be found in
Table~\ref{tab:torus_results}. A comparison of the prior eigenvalues
for the helical case in Table~\ref{tab:cross_verification_2D_results}
show that the eigenvalues have indeed changed.

From the mode plots in Figure~\ref{fig:torus_result}, we observe a
squeezing of the mode profiles to the left (which translates to away
from the center of the torus).  Such shifts in core localization can
also be seen in the prior work of~\cite{ScherCole07}
which also considered the toroidal bending case. 
It is interesting to compare the last
two eigenmodes in the toroidal case (Figure~\ref{fig:torus_result})
with their previous analogues in the helical case (in
Figure~\ref{fig:cross_verification_2D_results}). Remnants of the typical four-leaf
pattern of the LP21 modes of the unbent fiber are still visible in the
toroidal case in Figure~\ref{fig:torus_result}, but are harder to
discern in Figure~\ref{fig:cross_verification_2D_results}
where they appear to have been partially fused together. The corresponding
eigenvalues in the last two columns of Table~\ref{tab:torus_results}
are closer to each other than their two analogues in the last two columns of Table~\ref{tab:cross_verification_2D_results}.

\subsection{Parameter study of eigenmodes for varying helix pitch}
\label{ssec:pitch}

In this subsection, we investigate how modes change
with the pitch when an optical fiber is helically coiled.  The
influence of the bend radius on mode profiles has been studied in the
toroidal case (where there is no pitch).  Hence, we only report our
investigations into the previously unknown case of variations with the
pitch, holding the bend radius fixed.  We now  consider more
realistic parameters, with all length units set to micrometers,
\begin{align*}
  r_0 = 12.5,\quad r_1 = 20r_0,\quad n_1 = 1.52,\quad n_0=\sqrt{0.1^2+n_1^2},\quad k=\frac{2\pi}{1.064},
\end{align*}
(so the $V$-number is  $V= 0.1kr_0$ and  the wavelength is $1.064$).
The bend radius
is fixed to $\ahel=15500$.

\begin{figure}
  \centering
  \begin{tabular}{ccccc}
    \includegraphics[width=0.15\textwidth]{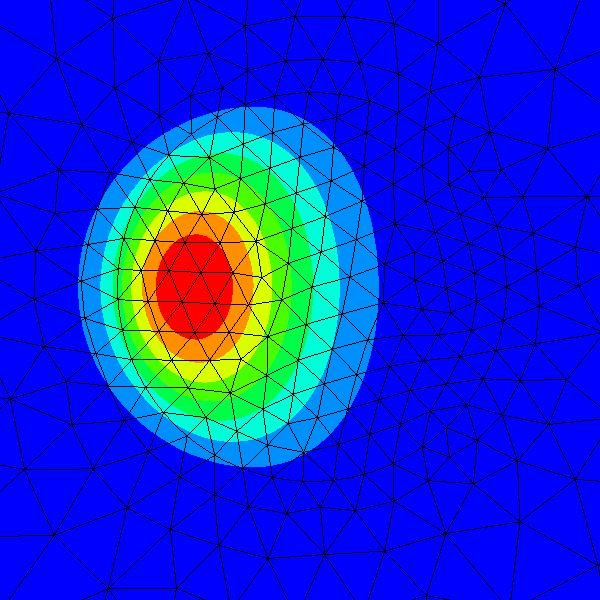}&
    \includegraphics[width=0.15\textwidth]{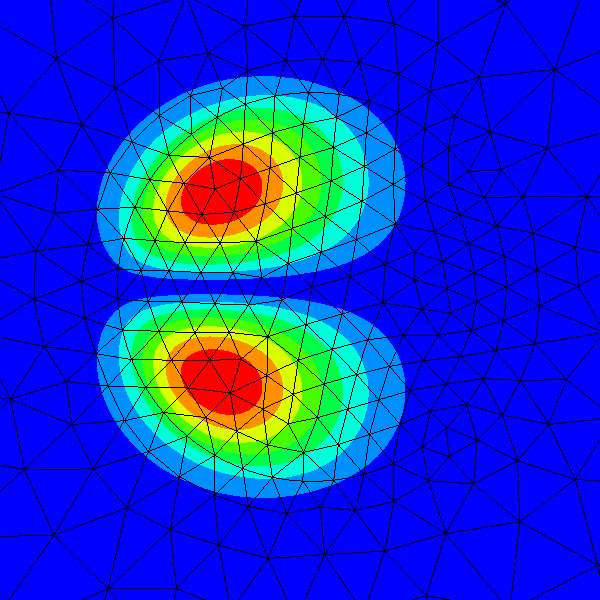}&
    \includegraphics[width=0.15\textwidth]{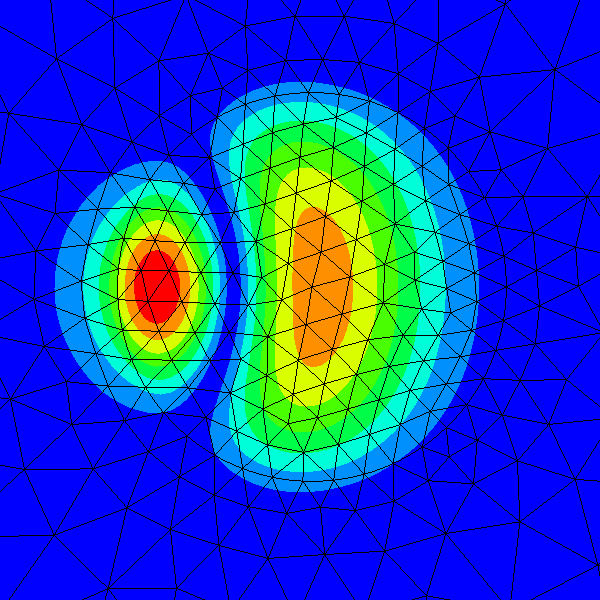}&
    \includegraphics[width=0.15\textwidth]{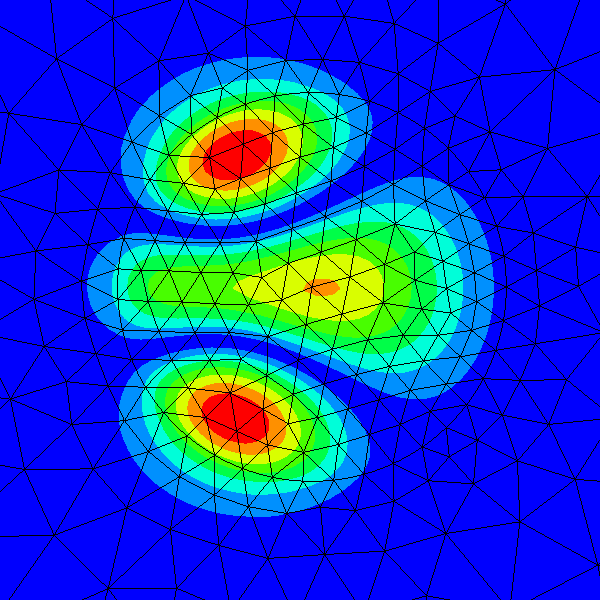}&
    \includegraphics[width=0.15\textwidth]{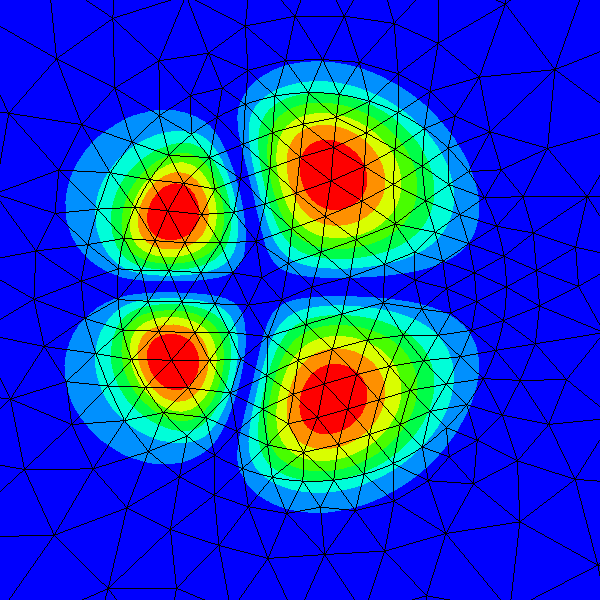}\\

    \includegraphics[width=0.15\textwidth]{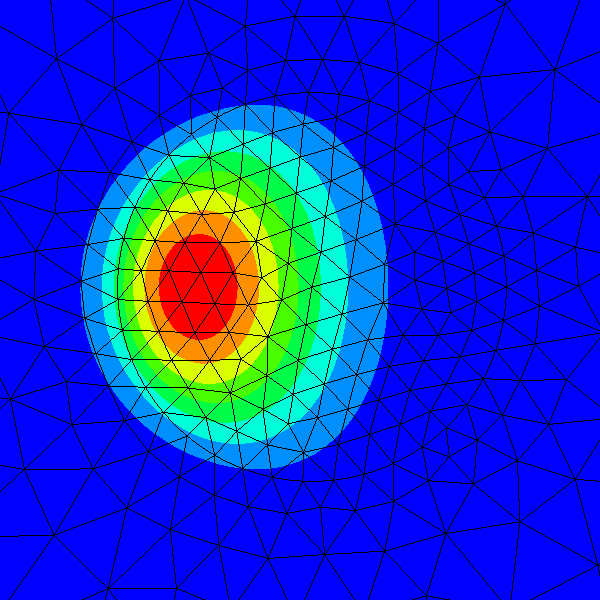}&
    \includegraphics[width=0.15\textwidth]{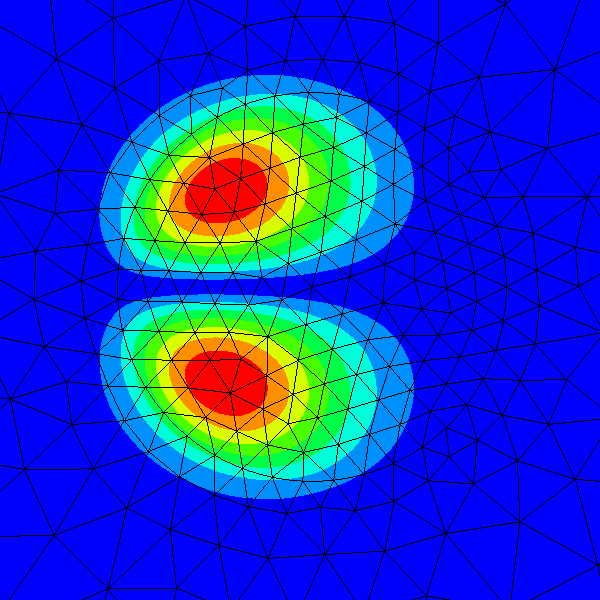}&
    \includegraphics[width=0.15\textwidth]{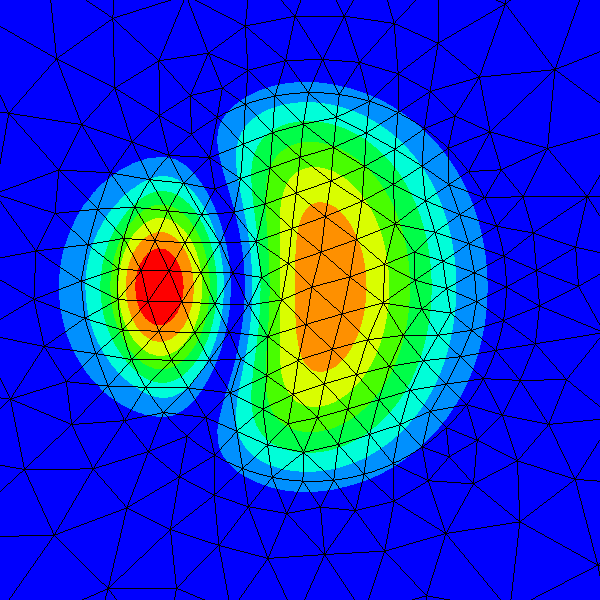}&
    \includegraphics[width=0.15\textwidth]{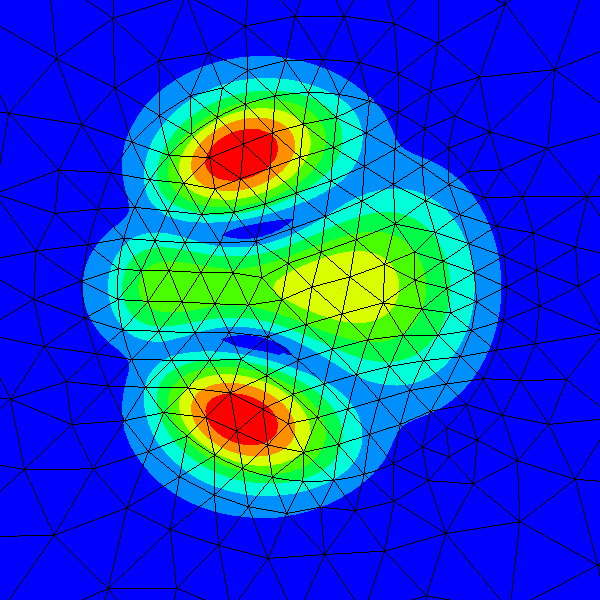}&
    \includegraphics[width=0.15\textwidth]{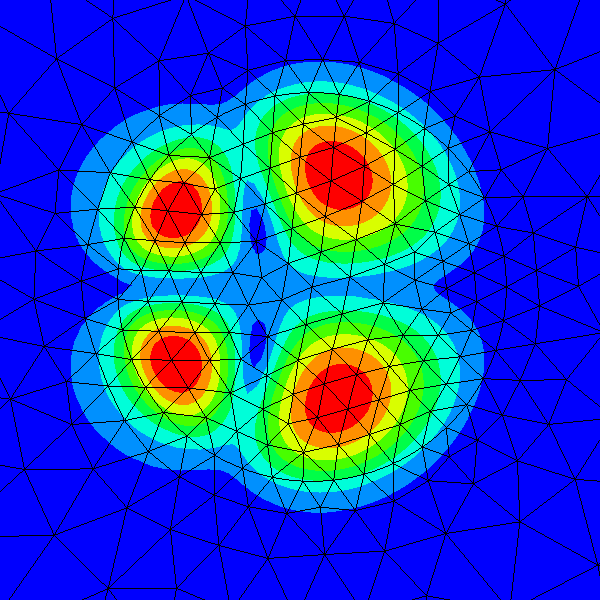}\\

    \includegraphics[width=0.15\textwidth]{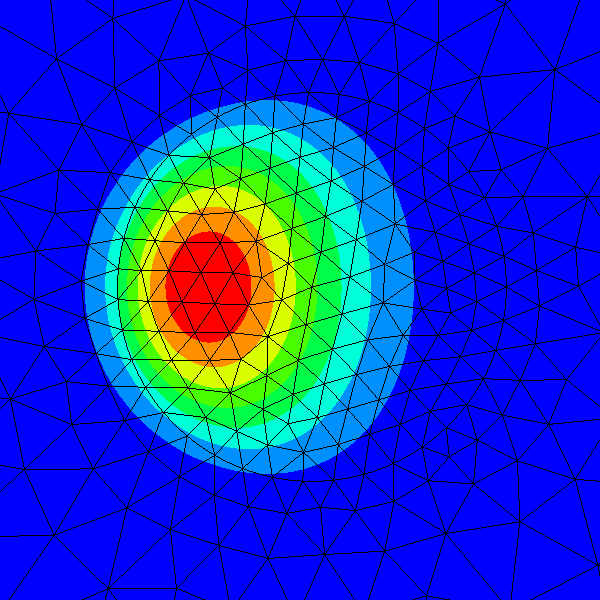}&
    \includegraphics[width=0.15\textwidth]{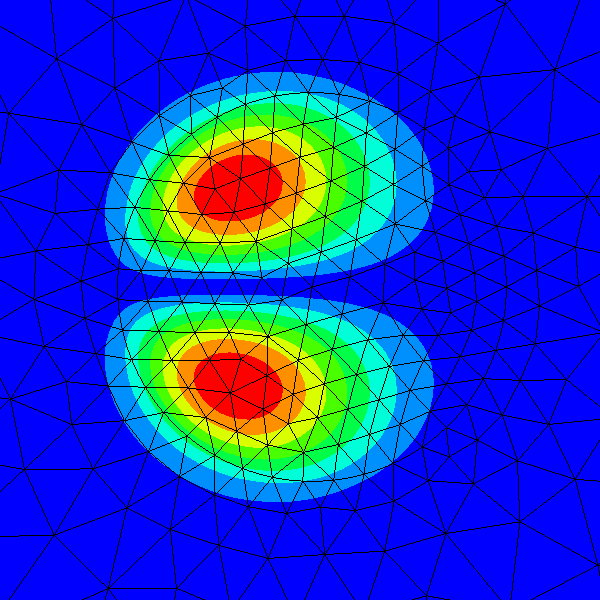}&
    \includegraphics[width=0.15\textwidth]{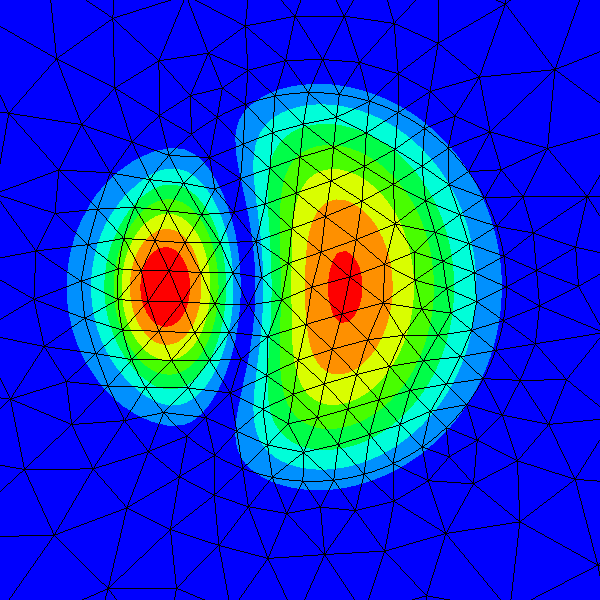}&
    \includegraphics[width=0.15\textwidth]{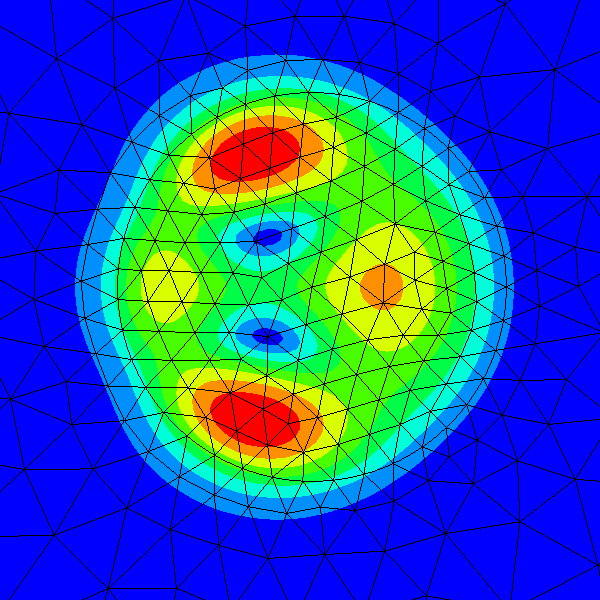}&
    \includegraphics[width=0.15\textwidth]{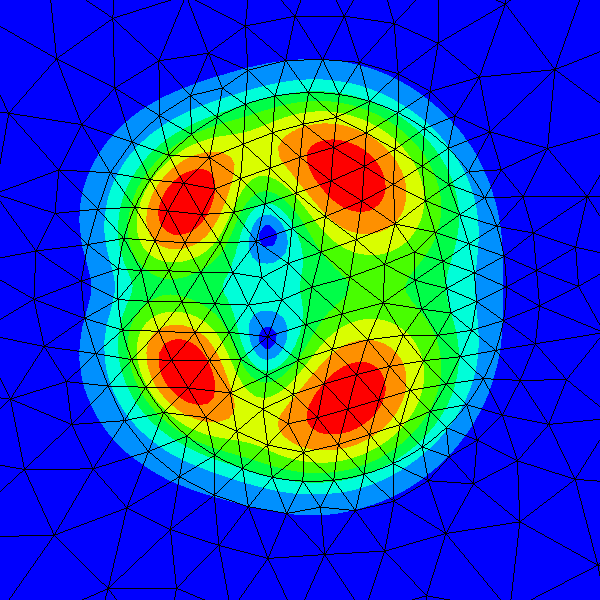}\\

    \includegraphics[width=0.15\textwidth]{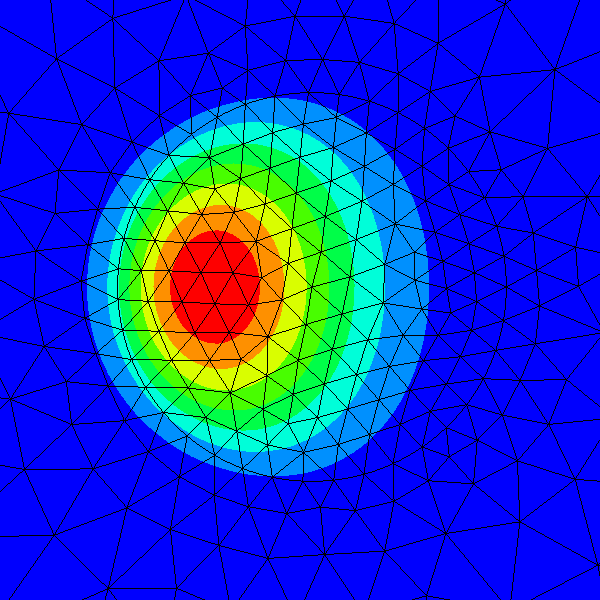}&
    \includegraphics[width=0.15\textwidth]{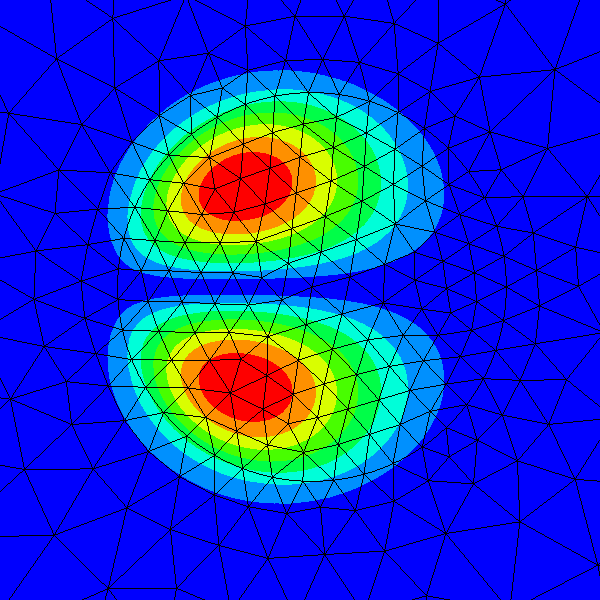}&
    \includegraphics[width=0.15\textwidth]{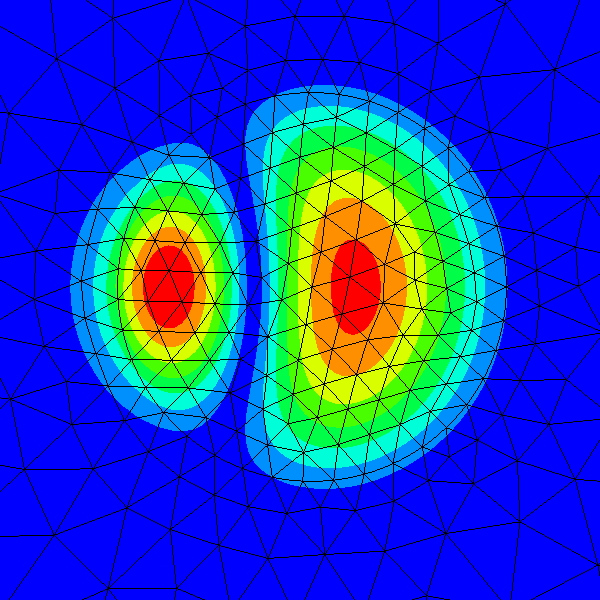}&
    \includegraphics[width=0.15\textwidth]{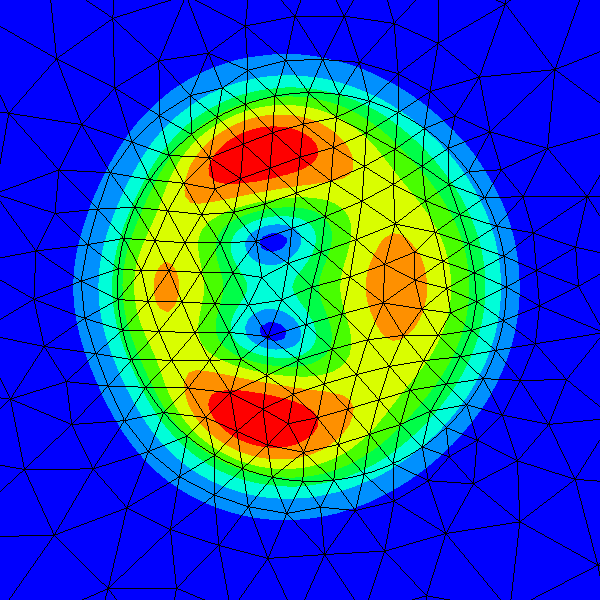}&
    \includegraphics[width=0.15\textwidth]{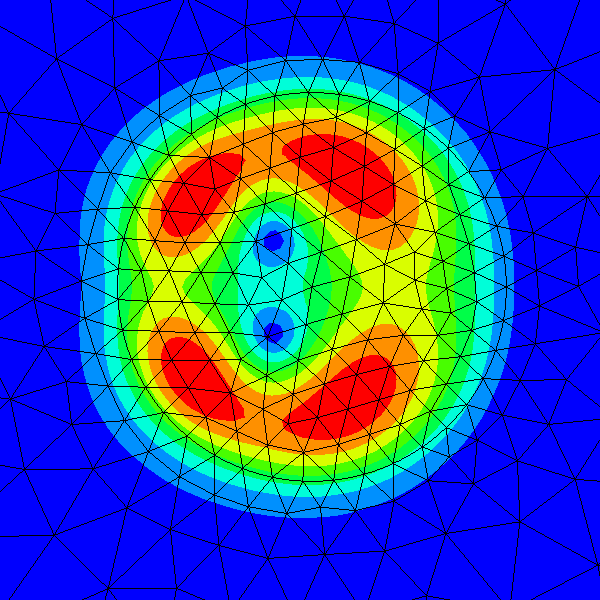}\\

    \includegraphics[width=0.15\textwidth]{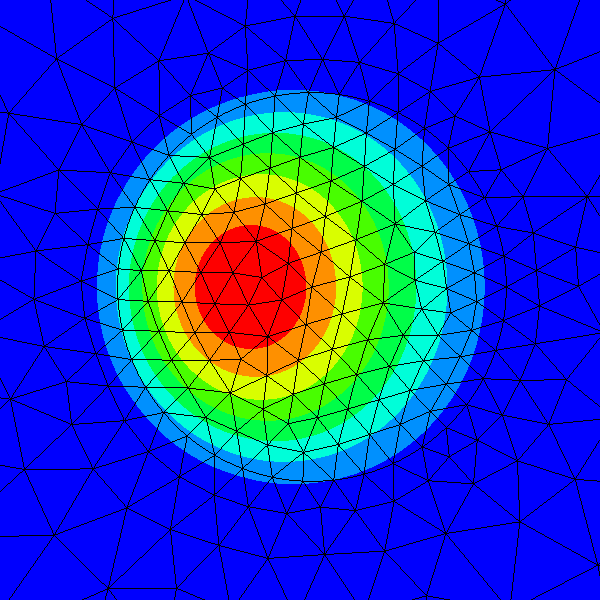}&
    \includegraphics[width=0.15\textwidth]{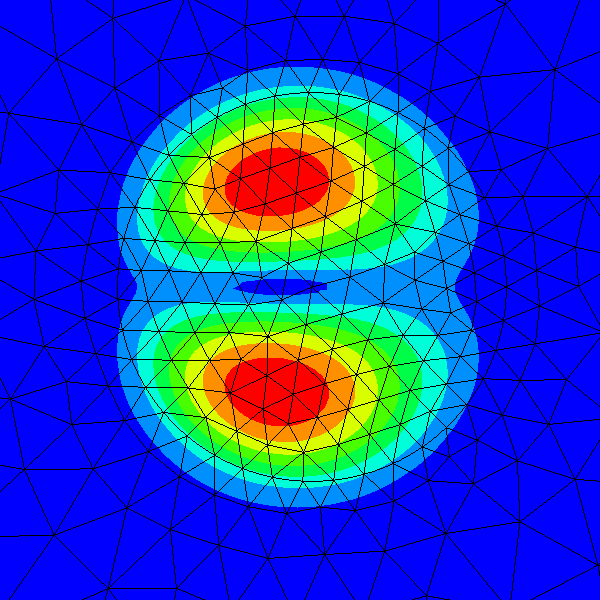}&
    \includegraphics[width=0.15\textwidth]{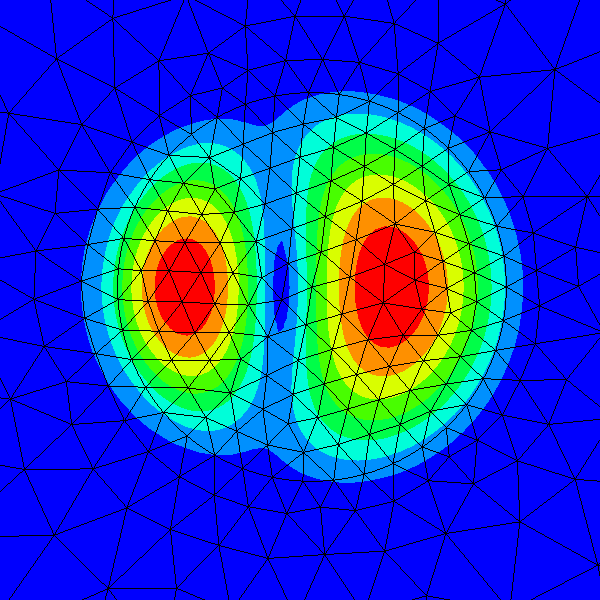}&
    \includegraphics[width=0.15\textwidth]{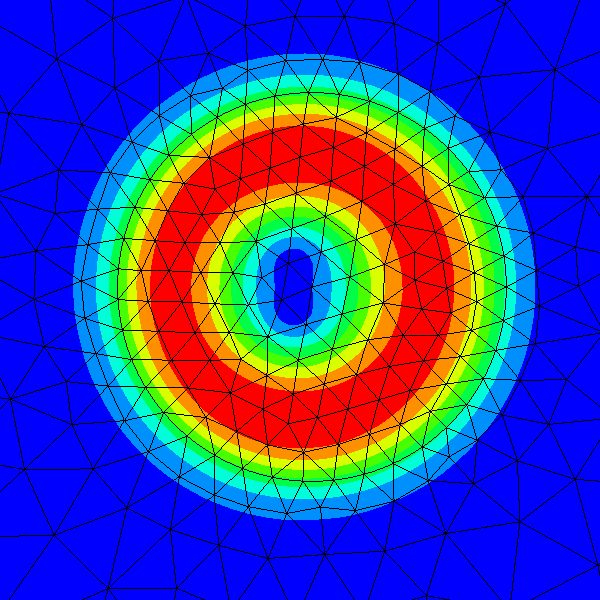}&
    \includegraphics[width=0.15\textwidth]{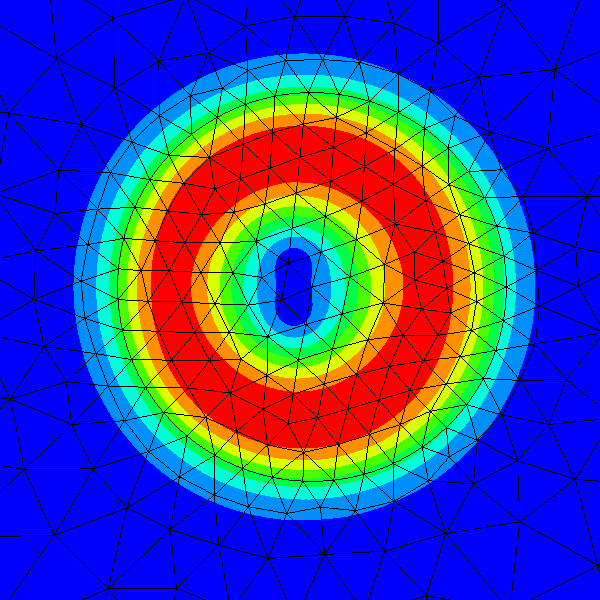}\\

    \includegraphics[width=0.15\textwidth]{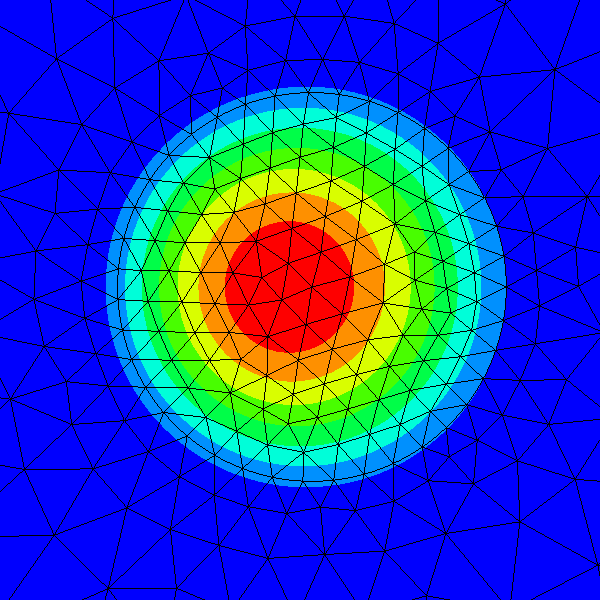}&
    \includegraphics[width=0.15\textwidth]{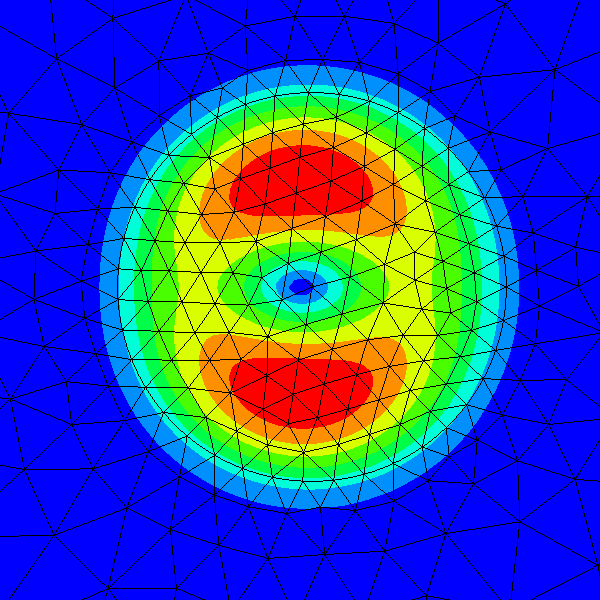}&
    \includegraphics[width=0.15\textwidth]{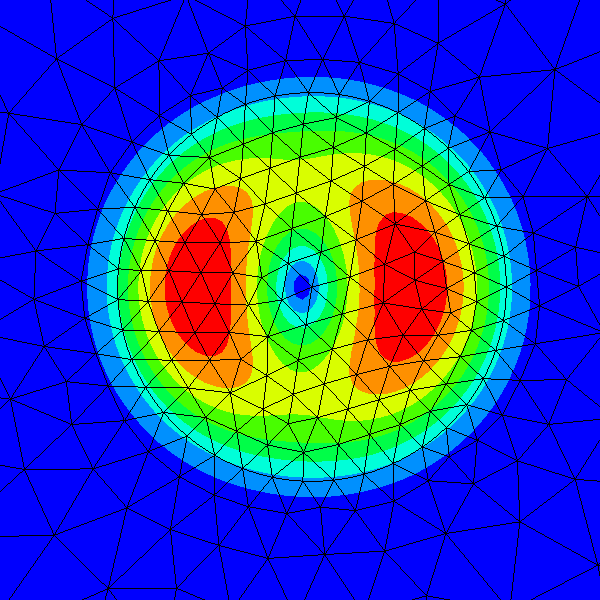}&
    \includegraphics[width=0.15\textwidth]{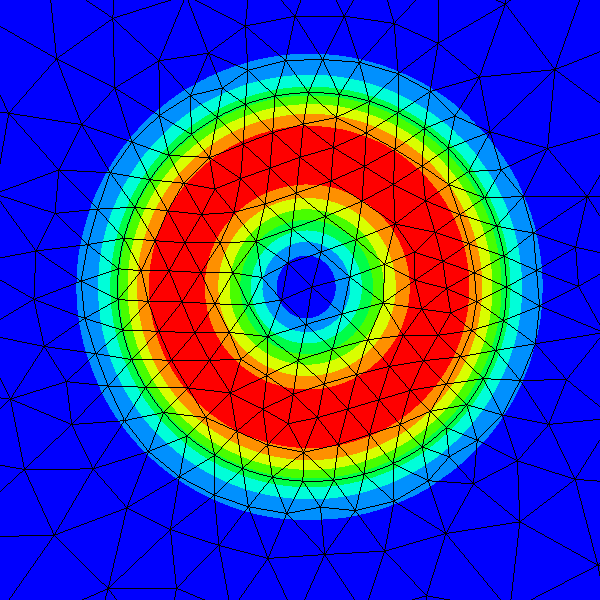}&
    \includegraphics[width=0.15\textwidth]{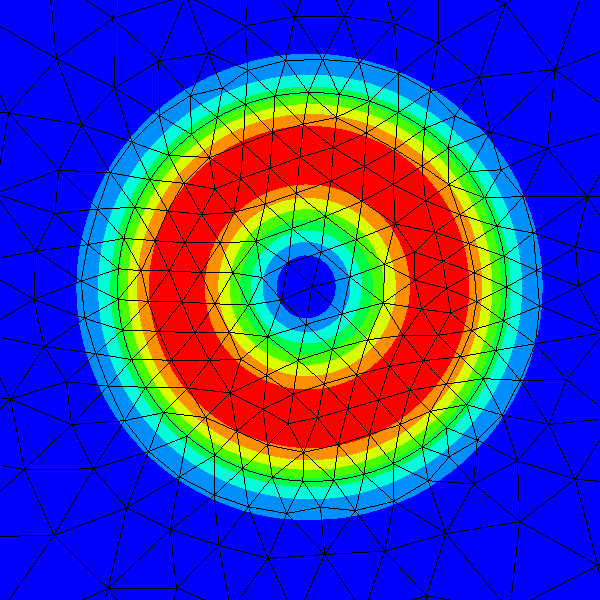}\\

    \includegraphics[width=0.15\textwidth]{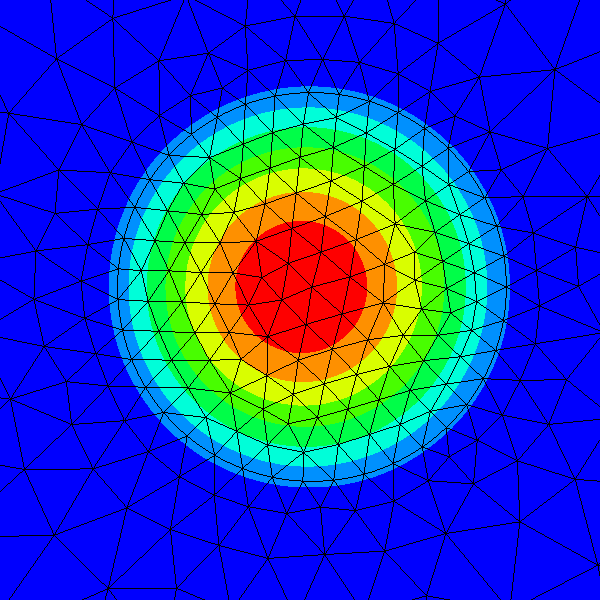}&
    \includegraphics[width=0.15\textwidth]{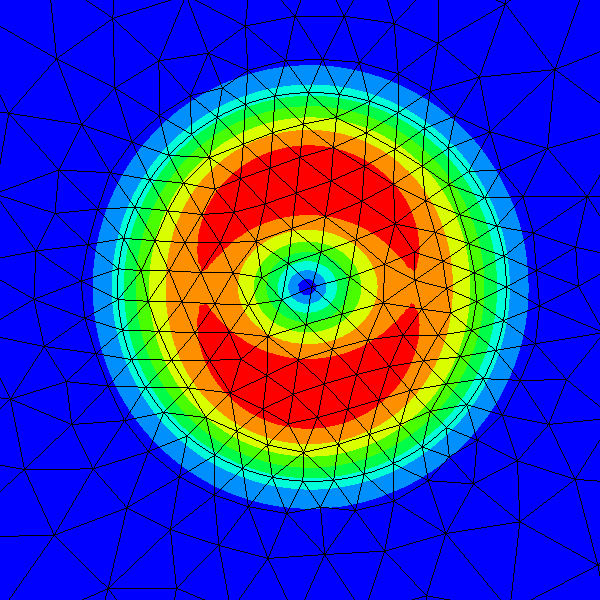}&
    \includegraphics[width=0.15\textwidth]{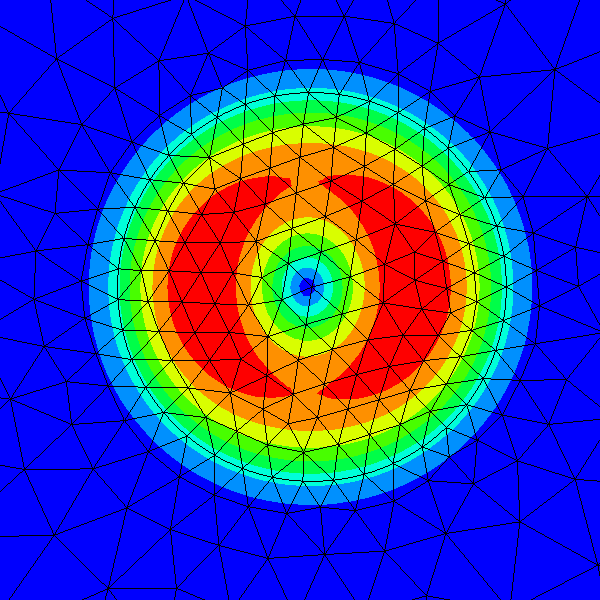}&
    \includegraphics[width=0.15\textwidth]{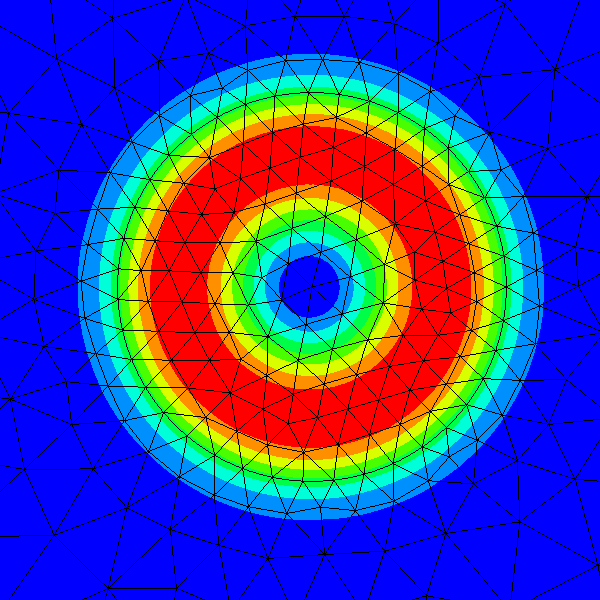}&
    \includegraphics[width=0.15\textwidth]{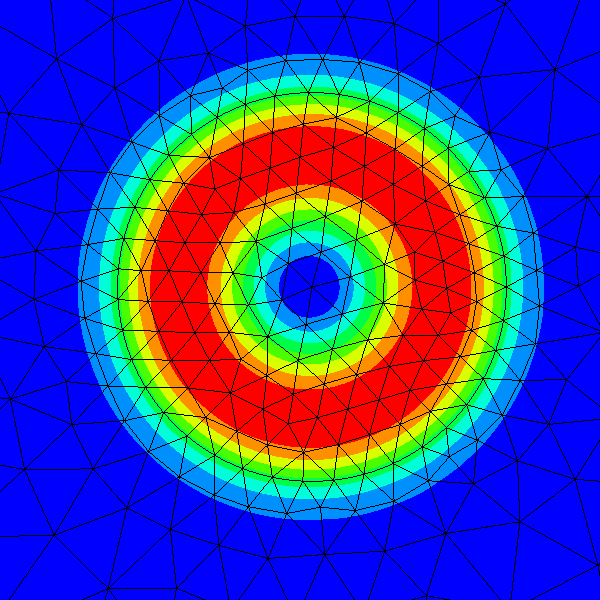}\\

    \includegraphics[width=0.15\textwidth]{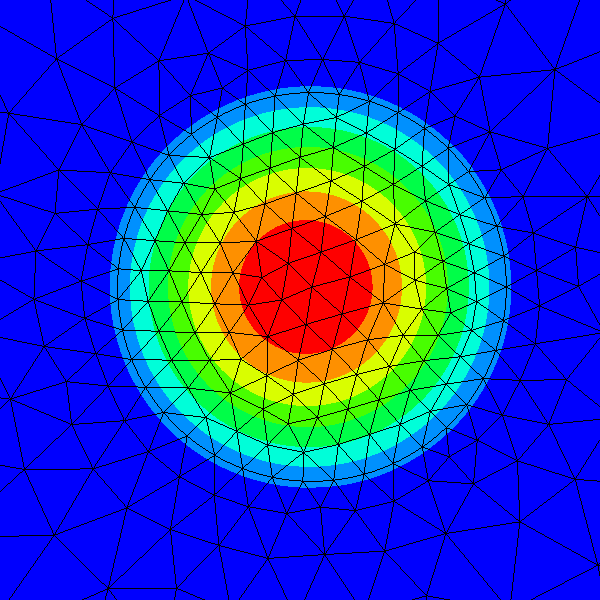}&
    \includegraphics[width=0.15\textwidth]{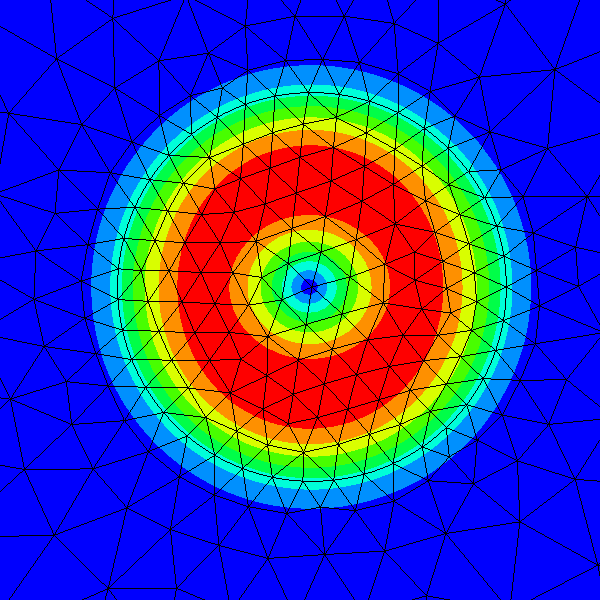}&
    \includegraphics[width=0.15\textwidth]{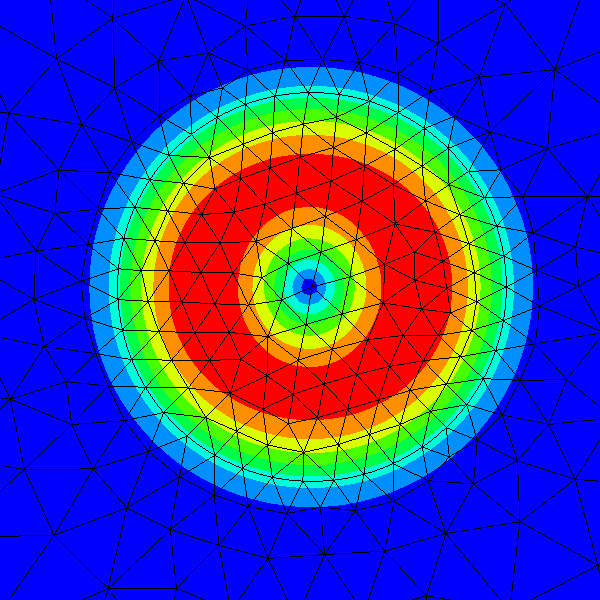}&
    \includegraphics[width=0.15\textwidth]{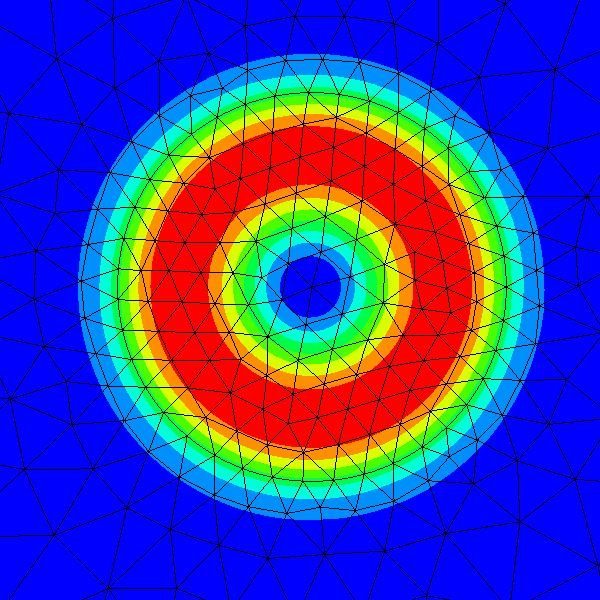}&
    \includegraphics[width=0.15\textwidth]{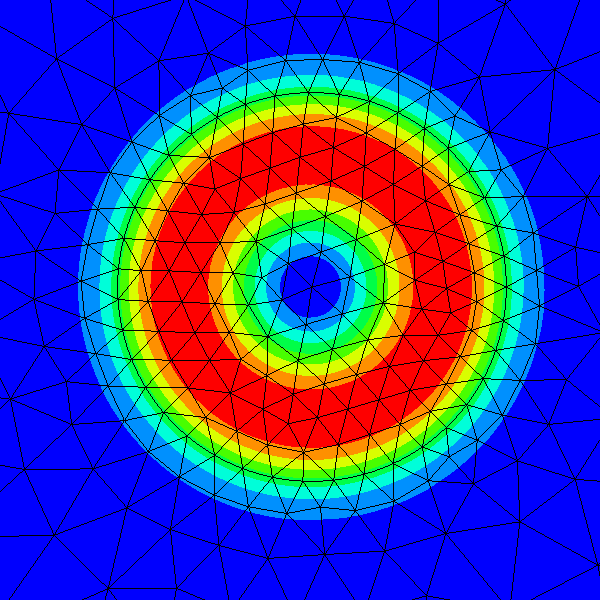}\\
  \end{tabular}
  \caption{Changes in eigenmodes are visible as pitch is varied from
    the top row to the bottom row:
    $\bhel=1,5000,10000,12500,25000,50000,75000,100000$.  Each column
    shows a different mode, with the first being the fundamental mode.}
  \label{fig:test_schermer_cole_pitch}
\end{figure}

\begin{table}[ht!]
  \centering
  \begin{footnotesize}
  \begin{tabular}{c|ccccc}
    \toprule
    $\bhel$ & $\beta_1^2$ & $\beta_2^2$ & $\beta_3^2$ & $\beta_4^2$ & $\beta_5^2$\\
    \midrule
    1      & 80.929235 & 80.873286 & 80.851811 & 80.799593 & 80.796619\\
    2500   & 80.927556 & 80.872003 & 80.851200 & 80.798888 & 80.796050\\
    5000   & 80.923099 & 80.868627 & 80.849689 & 80.797098 & 80.794604\\
    7500   & 80.917191 & 80.864234 & 80.847927 & 80.795010 & 80.792766\\
    10000  & 80.911145 & 80.859847 & 80.846440 & 80.793213 & 80.791031\\
    12500  & 80.905788 & 80.856069 & 80.845424 & 80.791893 & 80.789665\\
    25000  & 80.892417 & 80.847261 & 80.844356 & 80.789422 & 80.787356\\
    50000  & 80.888737 & 80.845129 & 80.844413 & 80.788620 & 80.787309\\
    75000  & 80.888406 & 80.844917 & 80.844453 & 80.788391 & 80.787471\\
    100000 & 80.888343 & 80.844845 & 80.844493 & 80.788276 & 80.787574\\
    250000 & 80.888314 & 80.844733 & 80.844590 & 80.788065 & 80.787778\\
    500000 & 80.888313 & 80.844697 & 80.844626 & 80.787994 & 80.787850\\
    \midrule
    $\infty$ & 80.888313 & 80.844661 & 80.844661 & 80.787922 & 80.787922\\
    \bottomrule
  \end{tabular}
\end{footnotesize}
    \caption{Observed variations in eigenvalues $\beta^2$ as pitch
      $\bhel$ is varied are reported. The last row gives the
      analytical eigenvalues for straight fiber.}
  \label{tab:test_schermer_cole_pitch}
\end{table}

We start with a pitch of $\bhel=1$. Then we increase the pitch keeping
track of the fundamental mode and the next four higher-order
modes. The results are shown in
Figure~\ref{fig:test_schermer_cole_pitch} for
$\bhel=1,5000,10000,12500,25000,50000,75000,100000$ and the
corresponding eigenvalues in
Table~\ref{tab:test_schermer_cole_pitch}. Clearly, these increasing
values of $\bhel$ lead to the upright limit case \eqref{eq:upright-1}
of vanishing torsion. As discussed previously in
\S\ref{sec:upright-limit}, in this limit, we expect to obtain the
eigenvalues of the unbent fiber. Indeed, the last few rows of
Table~\ref{tab:test_schermer_cole_pitch} show that our numerical
eigenvalues do appear to converge to those of the unbent fiber
(included in the last row) as $\bhel$ continues to increase.  We also
observe from the absolute values of the eigenfunctions plotted in
Figure~\ref{fig:test_schermer_cole_pitch} that they also appear to
converge to modulus of corresponding modes for the straight fiber.
The apparent radial symmetry of the modulus plots in the last row in
Figure~\ref{fig:test_schermer_cole_pitch} may be jarring.  While the
fundamental mode (LP01) of a straight fiber is radially symmetric, the
next higher order modes (LP11 and LP21) are not.  In fact the real and
imaginary parts (not displayed) of the second and third plots of the
last row in Figure~\ref{fig:test_schermer_cole_pitch} do show the
nonsymmetric two-leaf structure of LP11 modes, and those of fourth and
fifth plots have the four-leaf structure typical of LP21 modes. Their
magnitudes however came out radially symmetric in our computations, as
seen in the figure.  Such modulus plots for modes close to LP11 and
LP21 are perhaps atypical, but entirely possible. 
For double eigenvalues one
may only expect convergence of eigenspaces (not
individual eigenfunctions).  We have indeed verified that the errors
in the $L^2$~projections of the computed modes in
Figure~\ref{fig:test_schermer_cole_pitch} onto the span of the
corresponding LP modes of the straight fiber converge to zero as
$\bhel$ increases.

\section*{Acknowledgments}

This research was funded in part by the AFOSR grant FA9550-23-1-0103 and the Austrian Science Fund (FWF) under grant \href{https:/doi.org/10.55776/J4824}{10.55776/J4824}. This work also benefited from activities organized under the auspices of NSF RTG grant DMS-2136228. For open access purpose, the authors have applied a CC BY public copyright license to any Author Accepted Manuscript (AAM) version arising from this submission. 

\bibliographystyle{siam}
\bibliography{cites}


\end{document}